\documentclass[twoside,leqno,twocolumn]{article}

\usepackage[letterpaper]{geometry}

\usepackage{ltexpprt}

\usepackage{amsfonts}
\usepackage{epstopdf}
\usepackage{algorithmic}
\usepackage{epsfig}
\usepackage{pgfplots}
\usepackage{amsopn}
\usepackage{graphicx,epstopdf} 
\usepackage{subfig} 
\usepackage{xspace}
\usepackage{bold-extra}
\usepackage{thm-restate}

\usepackage[colorlinks=false]{hyperref}
\hypersetup{pdfborder={0 0 0.4}}  
\usepackage[nameinlink,capitalise]{cleveref}
\usepackage{lineno}

\graphicspath{{./figures/}}

\ifpdf
\DeclareGraphicsExtensions{.eps,.pdf,.png,.jpg}
\else
\DeclareGraphicsExtensions{.eps}
\fi

\newcommand{\X}{\raisebox{-2pt}{\includegraphics[scale=0.45]{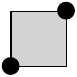}}}
\newcommand{\Tilde}{\raisebox{-2pt}{\includegraphics[scale=0.65]{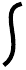}}}
\newcommand{\D}{\raisebox{-2pt}{\includegraphics[scale=0.45]{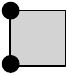}}\hspace{1pt}}
\renewcommand{\L}{\raisebox{-2pt}{\includegraphics[scale=0.45]{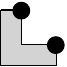}}}
\newcommand{\C}{\raisebox{-2pt}{\includegraphics[scale=0.45]{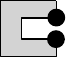}}}
\newcommand{\x}{\raisebox{-2pt}{\includegraphics[scale=0.25]{X.pdf}}}
\renewcommand{\d}{\raisebox{-2pt}{\includegraphics[scale=0.25]{D.pdf}}}
\renewcommand{\c}{\raisebox{-2pt}{\includegraphics[scale=0.25]{C.pdf}}}
\renewcommand{\l}{\raisebox{-2pt}{\includegraphics[scale=0.25]{L.pdf}}}

\newcommand{\oneB}{\raisebox{-2pt}{\includegraphics[scale=0.45]{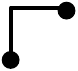}}}
\newcommand{\zeroB}{\raisebox{-2pt}{\includegraphics[scale=0.45]{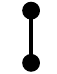}}}
\newcommand{\oneb}{\raisebox{-2pt}{\includegraphics[scale=0.25]{1-bend.pdf}}}
\newcommand{\zerob}{\raisebox{-2pt}{\includegraphics[scale=0.25]{0-bend.pdf}}}

\newcommand{\myparagraph}[1]{\smallskip\noindent\textbf{\boldmath #1}}

\newcommand{\skel}{\mathrm{skel}\xspace}
\newcommand{\rect}{\overline}

\DeclareMathOperator{\intr}{intr}
\DeclareMathOperator{\extr}{extr}
\DeclareMathOperator{\flex}{flex}

\newtheorem{definition}{Definition}
{\itshape}{\rmfamily}

\crefname{theorem}{Theorem}{Theorems}
\crefname{lemma}{Lemma}{Lemmas}
\crefname{hypothesis}{Hypothesis}{Hypotheses}
\crefname{property}{Property}{Properties}
\crefname{section}{Section}{Sections}
\crefname{subsection}{Section}{Sections}
\crefname{figure}{Fig.}{Figs.}
\crefname{equation}{Equation}{Equations}

\begin{document}
\title{\Large Optimal Orthogonal Drawings of Planar 3-Graphs in Linear Time \thanks{Work supported in part
		by MIUR Project ``AHeAD: efficient Algorithms for HArnessing networked Data'', PRIN 20174LF3T8.}}
\author{Walter Didimo\thanks{Universit\`a degli Studi di Perugia, ITALY} 
	\and Giuseppe Liotta\footnotemark[2] 
	\and Giacomo Ortali\footnotemark[2] 
	\and  Maurizio Patrignani\thanks{Roma Tre University, ITALY.}
}

\date{}

\maketitle



\maketitle

\begin{abstract}
 A \emph{planar orthogonal drawing} $\Gamma$ of a planar graph $G$ is a geometric representation of $G$ such that the vertices are drawn as distinct points of the plane, the edges are drawn as chains of horizontal and vertical segments, and no two edges intersect except at their common end-points. A \emph{bend} of $\Gamma$ is a point of an edge where a horizontal and a vertical segment meet. $\Gamma$ is \emph{bend-minimum} if it has the minimum number of bends over all possible planar orthogonal drawings of~$G$. This paper addresses a long standing, widely studied, open question: Given a planar 3-graph $G$ (i.e., a planar graph with vertex degree at most three), what is the best computational upper bound to compute a bend-minimum planar orthogonal drawing of~$G$ in the variable embedding setting? In this setting the algorithm can choose among the exponentially many planar embeddings of $G$ the one that leads to an orthogonal drawing with the minimum number of bends. We answer the question by describing an $O(n)$-time algorithm that computes a bend-minimum planar orthogonal drawing of~$G$ with at most one bend per edge, where $n$ is the number of vertices of $G$. The existence of an orthogonal drawing algorithm that simultaneously minimizes the total number of bends and the number of bends per edge was previously unknown. 
\end{abstract}







\section{Introduction}\label{se:intro}
Graph drawing is a well established research area that addresses the problem of constructing geometric representations of abstract graphs and networks~\cite{DBLP:books/ph/BattistaETT99,DBLP:reference/crc/2013gd}. It  combines flavors of topological graph theory, computational geometry, and graph algorithms. Various graphic standards have been proposed for the representation of graphs
in the plane. Usually, each vertex is represented by a point and each edge $(u,v)$ is
represented by a simple Jordan arc joining the points associated with vertices
$u$ and $v$. In an {\em orthogonal drawing} the edges are chains of horizontal ad vertical segments (see, e.g., \cref{fi:intro}).  Orthogonal drawings are among the earliest and most studied graph layout standards because of their direct application in several domains, including software engineering, database design, circuit design, and visual interfaces~\cite{DBLP:journals/jss/BatiniTT84,dl-gvdm-07,DBLP:journals/ivs/EiglspergerGKKJLKMS04,DBLP:books/sp/Juenger04,Lengauer-90}. Since the readability of an orthogonal drawing is affected by the crossings and the bends along its edges (see, e.g.~\cite{DBLP:books/ph/BattistaETT99}), a classical research subject in graph drawing studies planar (i.e. crossing-free) orthogonal drawings with the minimum number of bends. A planar $k$-graph is a planar graph with maximum vertex degree $k$. Valiant has proved that a graph has a planar orthogonal drawing if and only if it is a planar 4-graph~\cite{v-ucvc-81}.

\begin{figure}[t]
	\centering
	\subfloat[]{\label{fi:intro-a}\includegraphics[width=0.33\columnwidth]{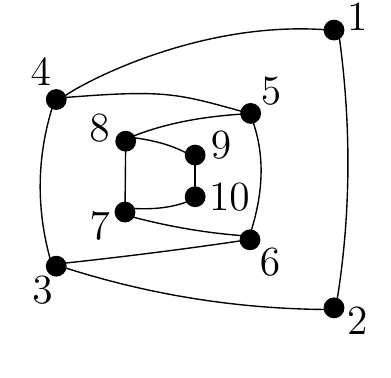}}
	\hfil
	\subfloat[]{\label{fi:intro-b}\includegraphics[width=0.33\columnwidth]{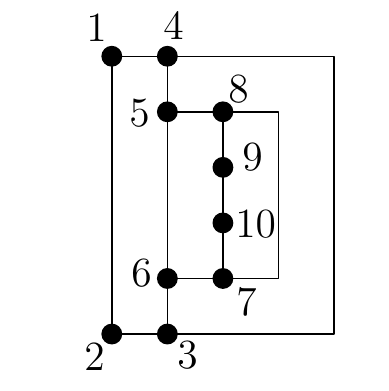}}
	\hfil
	\subfloat[]{\label{fi:intro-c}\includegraphics[width=0.33\columnwidth]{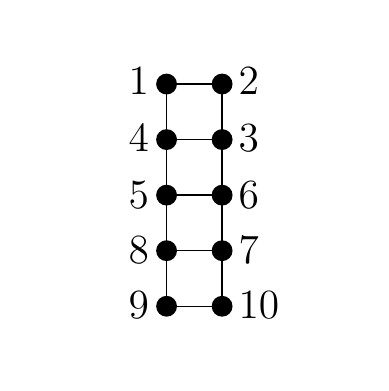}}
	\caption{(a) A planar embedded 3-graph $G$. (b)-(c) Two bend-minimum orthogonal drawings of $G$ in the fixed embedding setting and in the variable embedding setting, respectively}\label{fi:intro}
\end{figure}

A seminal paper by Storer~\cite{DBLP:conf/stoc/Storer80} conjectured that computing a planar orthogonal drawing with the minimum number of bends is computationally hard. The conjecture was proved incorrect by Tamassia~\cite{DBLP:journals/siamcomp/Tamassia87} in the so-called ``fixed embedding setting'', that is when the input graph $G$ is given with a planar embedding and the algorithm computes a bend-minimum orthogonal drawing of $G$ with the same planar embedding. Conversely, Garg and Tamassia~\cite{DBLP:journals/siamcomp/GargT01} proved the conjecture of Storer to be correct in the ``variable embedding setting'', that is when the drawing algorithm is asked to find which one of the (exponentially many) planar embeddings of $G$ gives rise to an orthogonal drawing that has the minimum number of bends.  However, if the input is a planar $3$-graph there exists a polynomial time solution~\cite{DBLP:journals/siamcomp/BattistaLV98}. Note that there are planar $3$-graphs for which a bend-minimum orthogonal drawing requires linearly many bends in the fixed embedding setting, but it has no bends in the variable embedding setting~\cite{DBLP:journals/siamcomp/BattistaLV98}. See \cref{fi:intro} for an example.

The polynomial-time algorithm presented in~\cite{DBLP:journals/siamcomp/BattistaLV98} has time complexity $O(n^5 \log n)$, where $n$ is the number of vertices of the planar $3$-graph.  Since its first publication of this algorithm more than twenty years ago, the question of establishing the best computational upper bound to the problem of computing a bend-minimum orthogonal drawing of a planar $3$-graph has been studied by several papers and mentioned as open in books, surveys, and book chapters (see, e.g.,~\cite{DBLP:conf/gd/BrandenburgEGKLM03,DBLP:conf/compgeom/ChangY17,DBLP:books/ph/BattistaETT99,DBLP:journals/siamcomp/BattistaLV98,dlt-dg-13,dlt-gd-17,DBLP:conf/gd/DidimoLP18,DBLP:journals/ieicet/RahmanEN05,DBLP:journals/siamdm/ZhouN08}). A significant improvement was presented by Chang and Yen~\cite{DBLP:conf/compgeom/ChangY17} who achieve ${\tilde{O}}(n^{\frac{17}{7}})$ time complexity by exploiting recent techniques about the efficient computation of a min-cost flow in unit-capacity networks~\cite{DBLP:conf/focs/CohenMTV17}. This complexity bound is lowered to $O(n^2)$ in~\cite{DBLP:conf/gd/DidimoLP18}, where the first algorithm that does not use a network flow approach to compute a bend-minimum orthogonal drawing of a planar $3$-graph in the variable embedding setting is presented.


In this paper we describe the first linear-time algorithm that optimizes the total number of bends when computing an orthogonal drawing of a planar $3$-graph in the variable embedding setting. Namely, in addition to minimizing the total number of bends, we optimize the number of bends per edge. This is one of the very few linear-time graph drawing algorithms in the variable embedding setting and it is the first algorithm that optimizes both the number of bends per edge and the total number of bends in an orthogonal drawing. More formally, we prove the following result.

\begin{restatable}{theorem}{thMain}\label{th:main}
	Let $G$ be an $n$-vertex planar 3-graph distinct from $K_4$. There exists an $O(n)$-time algorithm that computes an orthogonal drawing of $G$ with the minimum number of bends and at most one bend per edge in the variable embedding setting.
\end{restatable}

Concerning the bend minimization problem of orthogonal drawings in the variable embedding setting, we recall that a linear-time solution is known only for a rather restricted family of graphs, namely the 2-connected series-parallel $3$-graphs~\cite{DBLP:journals/siamdm/ZhouN08}. We also recall that the linear-time algorithm of~\cite{DBLP:journals/ieicet/RahmanEN05} for testing whether a subdivision of a 3-connected cubic graph admits an orthogonal drawing without bends does not consider the bend minimization problem.
As for optimizing the number of bends per edge, it is known that every planar $3$-graph (except $K_4$)  has an orthogonal drawing with at most one bend per edge, but the drawing algorithm that achieves this bound does not minimize the total number of bends~\cite{DBLP:journals/algorithmica/Kant96}. 

From a methodological point of view, the main ingredients for the proof of Theorem~\ref{th:main} are: $(i)$ A combinatorial argument proving the existence of a bend-minimum orthogonal drawing with at most one bend per edge for any planar 3-graph distinct from $K_4$. $(ii)$ A linear-time labeling algorithm that assigns a number to each  edge $e$ of $G$, representing the number of bends of a bend-minimum orthogonal drawing of $G$ with $e$ on the external face; the efficiency of this algorithm relies on the use of a novel data structure, called {\em bend-counter}, designed on a triconnected cubic planar graph: For each face $f$ of the graph, it returns in $O(1)$-time the minimum number of bends of an orthogonal drawing having $f$ as the external face. $(iii)$ A linear-time algorithm that, based on the previous edge labeling and on a suitable visit of SPQR-trees and block-cut-vertex trees, constructs an optimal drawing of $G$.

The remainder of the paper is organized as follows. \cref{se:preliminaries} gives basic definitions and terminology used throughout the paper. \cref{se:proof-structure} provides a high-level description of our results and illustrates the main ingredients used to prove them. \cref{se:triconnected,se:labeling,se:thgd2018-enhanced} contain details about the different ingredients illustrated in the previous section. \cref{se:conclusions} concludes the paper and suggests possible future research directions. 


\section{Preliminaries}\label{se:preliminaries} 
If $G$ is a graph, $V(G)$ and $E(G)$ denote the set of vertices and the set of edges of~$G$, respectively. We consider \emph{simple} graphs, i.e., graphs with neither self-loops nor multiple edges. The \emph{degree} of a vertex $v \in V(G)$, denoted as $\deg (v)$, is the number of its neighbors. $\Delta(G)$ denotes the maximum degree of a vertex of~$G$; if $\Delta(G) \leq h$ ($h \geq 1$), $G$ is an \emph{$h$-graph}. 

\myparagraph{Connectivity, Drawings, and Planarity.}
A graph $G$ is \emph{$1$-connected} if there is a path between any two vertices. Graph $G$ is \emph{$k$-connected}, for $k \geq 2$, if the removal of $k-1$ vertices leaves the graph $1$-connected. A $2$-connected ($3$-connected) graph is also called \emph{biconnected} (\emph{triconnected}). 

A \emph{planar drawing} of $G$ is a geometric representation in the plane such that: $(i)$ each vertex $v \in V(G)$ is drawn as a distinct point $p_v$; $(ii)$ each edge $e=(u,v) \in E(G)$ is drawn as a simple curve connecting $p_u$ and $p_v$; $(iii)$ no two edges intersect in $\Gamma$ except at their common end-vertices (if they are adjacent). A graph is \emph{planar} if it admits a planar drawing. A planar drawing $\Gamma$ of $G$ divides the plane into topologically connected regions, called \emph{faces}. The \emph{external face} of $\Gamma$ is the region of unbounded size; the other faces are \emph{internal}. The {\em cycle of a face} $f$, denoted as $C_f$, is the simple cycle consisting of the vertices and the edges along the boundary of the region of the plane identified by~$f$. A \emph{planar embedding} of $G$ is an equivalence class of planar drawings that define the same set of (internal and external) faces, and it can be described by the clockwise sequence of vertices and edges on the boundary of each face plus the choice of the external face. Graph $G$ together with a given planar embedding is an \emph{embedded planar graph}, or simply a \emph{plane graph}. If $f$ is a face of a plane graph, the {\em cycle of $f$}, denoted as $C_f$, consists of the vertices and edges along the boundary of the plane region identified by~$f$. If $\Gamma$ is a planar drawing of a plane graph $G$ whose set of faces is that described by the planar embedding of $G$, we say that $\Gamma$ \emph{preserves} this embedding, or equivalently that $\Gamma$ is an \emph{embedding-preserving drawing} of $G$. 

\myparagraph{Orthogonal Drawings and Representations.}
Let $G$ be a planar graph. In a \emph{planar orthogonal drawing} $\Gamma$ of $G$ the vertices are distinct points of the plane and each edge is a chain of horizontal and vertical segments. A graph $G$ admits a planar orthogonal drawing if and only if it is a planar $4$-graph, i.e., $\Delta(G) \leq 4$. A \emph{bend} of $\Gamma$ is a point of an edge where a horizontal and a vertical segment meet. $\Gamma$ is \emph{bend-minimum} if it has the minimum number of bends over all planar embeddings of $G$. A graph $G$ is {\em rectilinear planar} if it admits a planar orthogonal drawing without bends. Rectilinear planarity testing is NP-complete for planar $4$-graphs~\cite{DBLP:journals/siamcomp/GargT01}, but it is polynomial-time solvable for planar $3$-graphs~\cite{DBLP:conf/compgeom/ChangY17,DBLP:journals/siamcomp/BattistaLV98} and linear-time solvable for subdivisions of planar triconnected cubic graphs~\cite{DBLP:journals/ieicet/RahmanEN05}. 
Very recently a linear-time algorithm for rectilinear planarity testing of biconnected planar $3$-graphs has been presented~\cite{DBLP:conf/cocoon/Hasan019}.
By extending a result of Thomassen~\cite{Th84} about $3$-graphs that have a rectilinear drawing with all rectangular faces, Rahman et al.~\cite{DBLP:journals/jgaa/RahmanNN03} characterize rectilinear plane $3$-graphs (see Theorem~\ref{th:RN03}). For a plane graph $G$, let $C_o(G)$ be its external cycle ($C_o(G)$ is simple if $G$ is biconnected). Also, if $C$ is a simple cycle of $G$, $G(C)$ is the plane subgraph of $G$ that consists of $C$ and of the vertices and edges inside $C$.
A \emph{chord} of $C$ is an edge $e \notin C$ that connects two vertices of $C$: If $e$ is embedded outside $C$ it is an \emph{external chord}, otherwise it is an \emph{internal chord}.
An edge $e=(u,v)$ is a \emph{leg} of $C$ if exactly one of the vertices $u$ and $v$ belongs to $C$; such a vertex is a \emph{leg vertex} of $C$: If $v$ is embedded inside $C$, $(u,v)$ is an \emph{internal leg} of $C$; else it is an \emph{external leg}. $C$ is a \emph{$k$-extrovert cycle} of $G$ if it has exactly $k$ external legs and no external chord. $C$ is a \emph{$k$-introvert cycle} if it has exactly $k$ internal legs and no internal chord. In the following, for the sake of brevity, if $C$ is a $k$-extrovert ($k$-introvert) cycle, we simply refer to the $k$ external (internal) legs of $C$ as the \emph{legs}~of~$C$.

If the $k$ legs of a $k$-extrovert (resp. $k$-introvert) cycle $C$ are all attached to the same vertex $u \notin C$, then we say that $C$ is \emph{trivial}.
Clearly, a cycle $C$ may be $k$-extrovert and $k'$-introvert at the same time, for two (possibly coincident) constants $k$ and $k'$.
Since we are going to extensively exploit the interplay between these two types of cycles, in this paper we adopt a different terminology that, in our opinion, better relates them.
\cref{fi:introvert-extrovert-example} depicts different cycles of the same plane graph. In \cref{fi:introvert-extrovert-example-a} cycle $C_1$ has four external legs and four internal legs, thus it is 4-extrovert and 4-introvert. In \cref{fi:introvert-extrovert-example-b} cycle $C_2$ has two external legs and no internal leg, hence it is just 2-extrovert. In \cref{fi:introvert-extrovert-example-c} cycle $C_3$ is 6-extrovert and 2-introvert. We remark that $k$-extrovert cycles are called \emph{$k$-legged cycles} in~\cite{DBLP:journals/jgaa/RahmanNN03} and that $k$-introvert cycles are called \emph{$k$-handed cycles} in~\cite{DBLP:journals/ieicet/RahmanEN05}. By using this terminology we rephrase a characterization given in of~\cite{DBLP:journals/jgaa/RahmanNN03} as follows.

\begin{figure}[tb]
	\centering
	\subfloat[]{\label{fi:introvert-extrovert-example-a}\includegraphics[width=0.32\columnwidth]{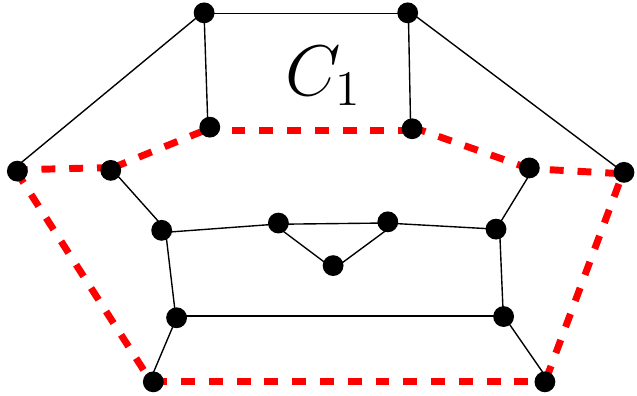}}
	\hfil
	\subfloat[]{\label{fi:introvert-extrovert-example-b}\includegraphics[width=0.32\columnwidth]{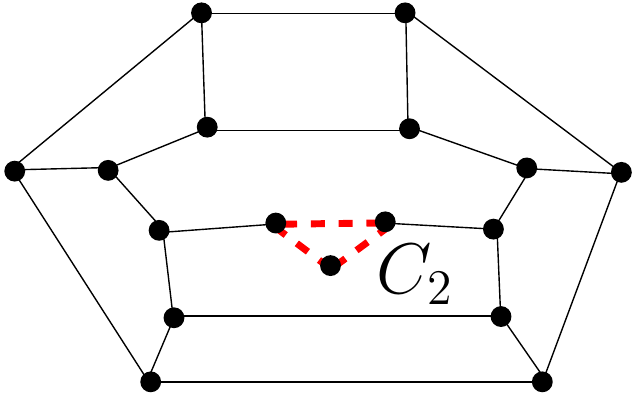}}
	\hfil
	\subfloat[]{\label{fi:introvert-extrovert-example-c}\includegraphics[width=0.32\columnwidth]{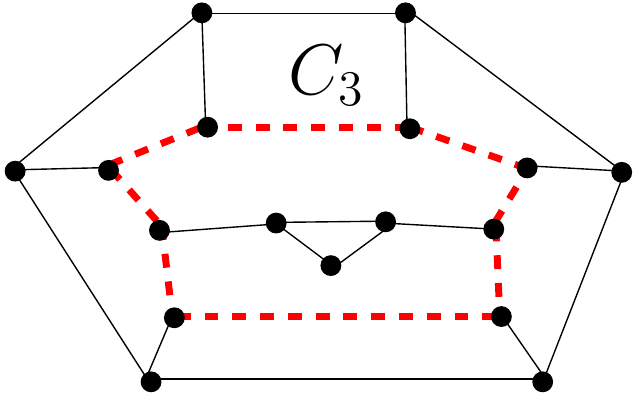}}
	\hfil
	\caption{Different cycles (dashed) of the same plane graph: (a) Cycle $C_1$ is 4-extrovert and 4-introvert. (b) Cycle $C_2$ is 2-extrovert. (c) Cycle $C_3$ is 6-extrovert and 2-introvert.}\label{fi:introvert-extrovert-example}
\end{figure}

\begin{theorem}[\cite{DBLP:journals/jgaa/RahmanNN03}]\label{th:RN03}
	Let $G$	be a biconnected plane $3$-graph. $G$ admits an orthogonal drawing without bends if and only if: $(i)$ $C_o(G)$ has at least four degree-2 vertices; $(ii)$ each $2$-extrovert cycle has at least two degree-2 vertices; $(iii)$ each $3$-extrovert cycle has at least one degree-2 vertex.
\end{theorem}


Intuitively, in an orthogonal drawing each cycle of $G$ must have at least four reflex angles in its outside, also called \emph{corners}. Condition~$(i)$ guarantees that there are at least four corners on the external face of~$G$. Conditions~$(ii)$ and~$(iii)$ reflect the fact that two (resp. three) corners of a 2-extrovert (resp. a 3-extrovert) cycle coincide with its leg vertices.

A plane graph that satisfies the conditions of Theorem~\ref{th:RN03} will be called a \emph{good} (plane) graph. A \emph{bad} cycle is any $2$-extrovert and any $3$-extrovert cycle that does not satisfy Conditions~$(ii)$ and~$(iii)$ of Theorem~\ref{th:RN03}, respectively. Since in this paper we only consider planar drawings, we often use the term ``orthogonal drawing'' in place of ``planar orthogonal drawing''.

The partial description of a plane orthogonal drawing $\Gamma$ in terms of the left and right bends along each edge and of the geometric angles at each vertex is an \emph{orthogonal representation} and is denoted by~$H$. This description abstracts from the vertex and bend coordinates of $\Gamma$. In other words, $H$ represents the class of orthogonal drawings having the same sets of vertex angles and bends as $\Gamma$. 



Since for a given orthogonal representation $H$, an orthogonal drawing of $H$ can be computed in linear time~\cite{DBLP:journals/siamcomp/Tamassia87},  the bend-minimization problem for orthogonal drawings can be studied at the level of orthogonal representation. Hence, from now on we focus on orthogonal representations rather than on orthogonal drawings.
%
%
%
Given an orthogonal representation $H$, we denote by $b(H)$ the total number of bends of $H$. 
If $v$ is a vertex of $G$, a \emph{$v$-constrained} bend-minimum orthogonal representation $H$ of $G$ is an orthogonal representation that has $v$ on its external face and that has the minimum number of bends among all the orthogonal representations with $v$ on the external face. Analogously, for an edge $e$ of $G$, an \emph{$e$-constrained} bend-minimum orthogonal representation of $G$ has $e$ on its external face and has the minimum number of bends among all the orthogonal representations with $e$ on the external face.

\myparagraph{SPQR-Trees.} Let $G$ be a biconnected graph. An \emph{SPQR-tree} $T$ of $G$ represents the decomposition of $G$ into its triconnected components and can be computed in linear time~\cite{DBLP:books/ph/BattistaETT99,DBLP:conf/gd/GutwengerM00,DBLP:journals/siamcomp/HopcroftT73}. See Fig.\ref{fi:spqr-tree} for an illustration. Each triconnected component corresponds to a node $\mu$ of $T$; the triconnected component itself is called the \emph{skeleton} of $\mu$ and denoted as $\skel(\mu)$. A node $\mu$ of $T$ can be of one of the following types: $(i)$ \emph{R-node}, if $\skel(\mu)$ is a triconnected graph; $(ii)$ \emph{S-node}, if $\skel(\mu)$ is a simple cycle of length at least three; $(iii)$ \emph{P-node}, if $\skel(\mu)$ is a bundle of at least three parallel edges; $(iv)$ \emph{Q-nodes}, if it is a leaf of $T$; in this case $\mu$ represents a single edge of $G$ and $\skel(\mu)$ consists of two parallel edges.
Neither two $S$- nor two $P$-nodes are adjacent in~$T$. A \emph{virtual edge} in $\skel(\mu)$ corresponds to a tree node $\nu$ adjacent to $\mu$ in $T$. If $T$ is rooted at one of its Q-nodes $\rho$, every skeleton (except the one of $\rho$) contains exactly one virtual edge that has a counterpart in the skeleton of its parent: This virtual edge is the \emph{reference edge} of $\skel(\mu)$ and of $\mu$, and its endpoints are the \emph{poles} of $\skel(\mu)$ and of $\mu$. The edge corresponding to $\rho$ is the \emph{reference edge} of $G$, and $T$ is the SPQR-tree of $G$ \emph{with respect to $e$}. For every node $\mu \neq \rho$ of $T$, the subtree $T_\mu$ rooted at $\mu$ induces a subgraph $G_\mu$ of $G$ called the \emph{pertinent graph} of $\mu$: The edges of $G_\mu$ correspond to the Q-nodes (leaves) of $T_\mu$. Graph $G_\mu$ is also called a \emph{component} of $G$ with respect to the reference edge $e$, namely $G_\mu$ is a P-, an R-, or an S-component depending on whether $\mu$ is a P-, an R-, or an S-component, respectively.

\begin{figure}[t]
	\centering
	\subfloat[]{\label{fi:spqr-tree-a}}{\includegraphics[width=0.33\columnwidth]{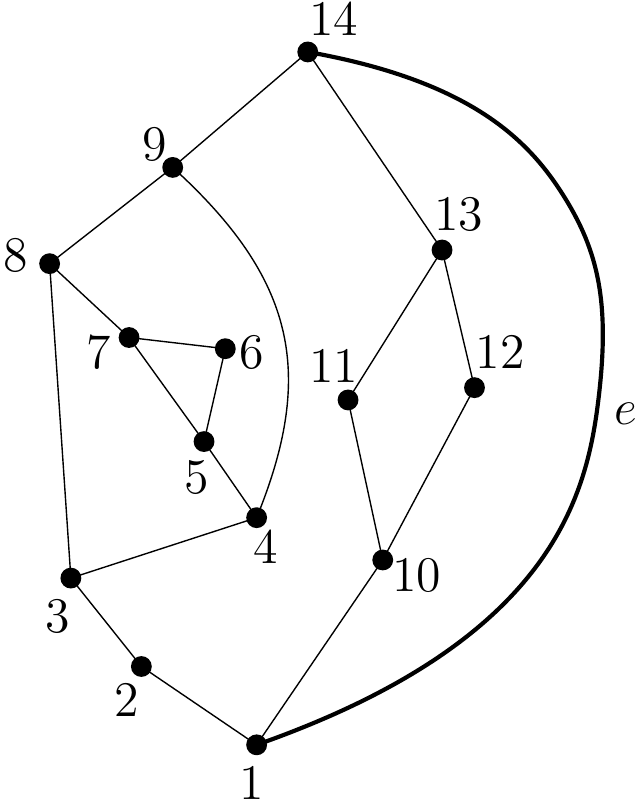}}
	\hfil
	\subfloat[]{\label{fi:spqr-tree-b}}{\includegraphics[width=0.33\columnwidth]{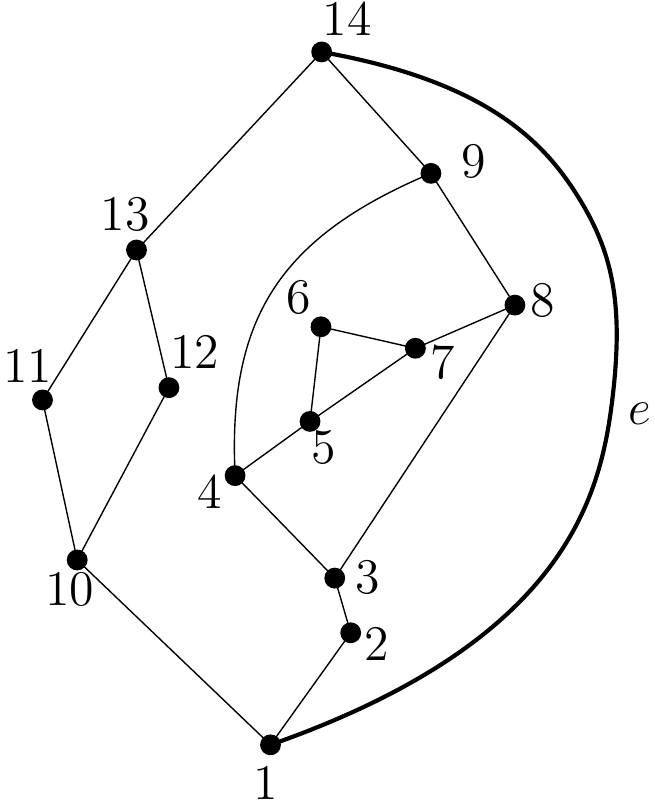}}
	\hfil
	\subfloat[]{\label{fi:spqr-tree-c}}{\includegraphics[width=0.8\columnwidth]{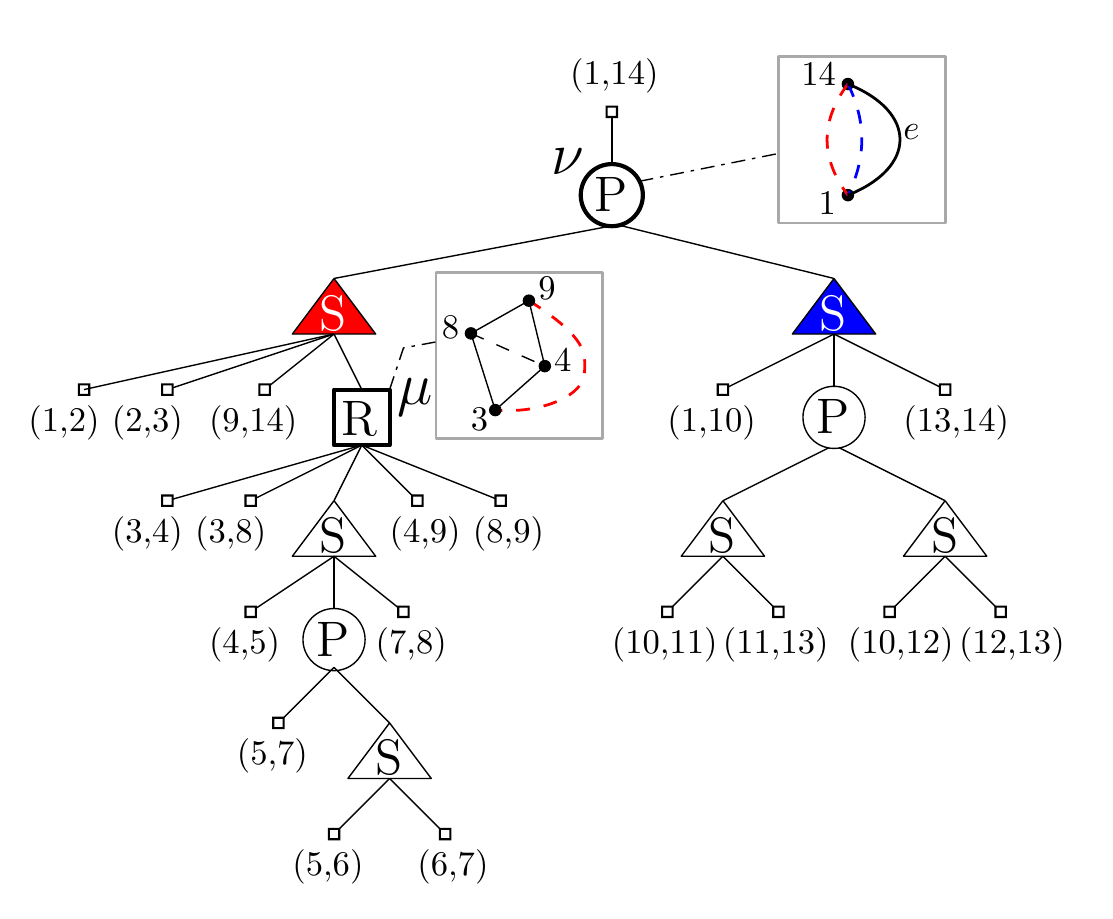}}
	\caption{(a)-(b) Two different embeddings of plane 3-graph $G$. (c) The SPQR-tree of $G$ with respect to $e$; the skeletons of nodes $\nu$ and $\mu$ are shown. The embedding in (b) has been obtained from the embedding in (a) by changing the embeddings of $\skel(\nu)$ and $\skel(\mu)$.}\label{fi:spqr-tree}
\end{figure}

The SPQR-tree $T$ rooted at a Q-node $\rho$ implicitly describes all planar embeddings of $G$ with the reference edge on the external face; they are obtained by combining the different planar embeddings of the skeletons of P- and R-nodes. For a P-node $\mu$, the embeddings of $\skel(\mu)$ are the different permutations of its non-reference edges; for an R-node $\mu$, $\skel(\mu)$ has two possible embeddings, obtained by flipping it (but its reference edge) at its poles.
The child of $\rho$ and its pertinent graph are the \emph{root child} of $T$ and the \emph{root child component} of $G$, respectively. An \emph{inner node} of $T$ is neither the root nor the root child of $T$. The pertinent graph of an inner node is an \emph{inner component} of $G$. The next lemma gives basic properties of $T$ when $\Delta(G) \leq 3$.

\begin{lemma}[\cite{DBLP:conf/gd/DidimoLP18}]\label{le:spqr-tree-3-graph}
	Let $G$ be a biconnected planar $3$-graph and $T$ be its SPQR-tree with respect to a reference edge $e$. The following properties hold: {\em \bf \textsf{T1}} Each P-node $\mu$ has exactly two children, one being an S-node and the other being an S- or a Q-node; if $\mu$ is the root child, both its children are S-nodes. {\em \bf \textsf{T2}} Each child of an R-node is either an S-node or a Q-node. {\em \bf \textsf{T3}} For each inner S-node $\mu$, the edges of $\skel(\mu)$ incident to the poles of $\mu$ are (real) edges of $G$. Also, there cannot be two incident virtual edges in $\skel(\mu)$.
\end{lemma}

\section{Structure of the Proof of Theorem~\ref{th:main}}\label{se:proof-structure}
%
Let $G$ be a planar 3-graph. An {\em optimal orthogonal representation} is an orthogonal representation of $G$ that has the minimum number of bends and at most one bend per edge in the variable embedding setting. We assume that $G$ is 1-connected, since otherwise we can independently compute optimal orthogonal drawings of its connected components.
The proof of Theorem~\ref{th:main} is based on three main ingredients. In this section we describe these ingredients and show how they are used to prove Theorem~\ref{th:main}.

\myparagraph{First~ingredient:~Representative~shapes.}
We show the existence of an optimal orthogonal representation of a biconnected planar 3-graph whose components have one of a constant number of possible ``orthogonal shapes'' that we call {\em representative shapes}. As a consequence, we can restrict the search space for an optimal orthogonal representation of a planar 3-graph to such a set of representative shapes.

Let $T$ be the SPQR-tree of $G$ with a reference edge~$e$ as its root. Let $H$ be an orthogonal representation of $G$ with $e$ on the external face. For a node $\mu$ of $T$, denote by $H_\mu$ the restriction of $H$ to the pertinent graph $G_\mu$ of $\mu$. We also call $H_\mu$ a \emph{component of $H$ with respect to $e$}. We say that $H_\mu$ is an S-, P-, Q-, or R-component with respect to $e$ depending on whether $\mu$ is an S-, P-, Q-, or R-node, respectively. If $\mu$ is the root child of $T$, $H_\mu$ is the \emph{root child component} of $H$ with respect to $e$; if $\mu$ is not the root child, $H_\mu$ is an \emph{inner component} of $H$ with respect to~$e$.

Let $u$ and $v$ be the two poles of $\mu$ and let $p_l$ and $p_r$ be the two paths from $u$ to $v$ on the external boundary of $H_\mu$, one walking clockwise and the other counterclockwise. If $\mu$ is a P- or an R-node, $p_l$ and $p_r$ are edge disjoint. If $\mu$ is a Q-node, both $p_l$ and $p_r$ coincide with the single edge represented by the Q-node. If $\mu$ is an S-node, $p_l$ and $p_r$ share some edges and they may coincide when $H_\mu$ is just a sequence of edges.  Also, the poles $u$ and $v$ of an S-node $\mu$ can either have inner degree one or two; the \emph{inner degree} of a pole is its degree in~$H_{\mu}$. Precisely, the poles of an S-node $\mu$ have both inner degree one if $\mu$ is an inner node, while they may have inner degree two if $\mu$ is the root child. We define two {\em alias vertices} $u'$ and $v'$ of the poles $u$ and $v$ of an S-nodes. If the inner degree of $u$ is one, $u'$ coincides with $u$. If the inner degree of $u$ is two, the alias vertex $u'$  subdivides the edge of $H$ that is incident to $u$ and that does not belong to $H_{\mu}$;  in this case, the {\em alias edge} of $u$ is the edge connecting $u$ to $u'$. The definition of the alias vertex $v'$ of $v$ and of alias edge of $v$ is analogous.

Let $p$ be a path between any two vertices in $H$. The \emph{turn number} of $p$, denoted as $t(p)$, is the absolute value of the difference between the number of right and the number of left turns encountered along $p$. A turn along $p$ is caused either by a bend along an edge of $p$ or by an angle of $90^\circ$ or $270^\circ$ at a vertex of~$p$. 

\begin{lemma}[\cite{DBLP:journals/siamcomp/BattistaLV98}]\label{le:k-spiral}
	Let $H_\mu$ be an S-component and let $p_1$ and $p_2$ be any two paths in $H_\mu$ between its alias vertices. 
	Then $t(p_1)=t(p_2)$.
\end{lemma}

Based on Lemma~\ref{le:k-spiral}, the orthogonal shape of an S-component can be described in terms of the turn number of any path $p$ between its two alias vertices. 

As for P-components and R-components, their orthogonal shapes are described in terms of the turn numbers of the two paths $p_l$ and $p_r$ connecting their poles on the external face. We consider the following shapes.

\begin{description}	
	\item[$\mu$ is a Q-node:] $H_\mu$ is \emph{0-shaped}, or \emph{\zeroB-shaped}, if it
	 is a straight-line segment; $H_\mu$ is \emph{1-shaped}, or \emph{\oneB-shaped}, if it
	 has exactly one bend.
	 
	 \item[$\mu$ is an S-node:] $H_\mu$ is a \emph{$k$-spiral}, for some integer $k$, if the turn number of any path $p$ between its two alias vertices is $t(p) = k$; if $H_\mu$ is $k$-spiral, we also say that $H_\mu$  has \emph{spirality} $k$.

	 \item[$\mu$ is either a P-node or an R-node:]  
	 $H_\mu$ is \emph{D-shaped}, or \emph{\D-shaped}, if $t(p_l)=0$ and $t(p_r)=2$, or vice versa; 
	 $H_\mu$ is \emph{X-shaped}, or \emph{\X-shaped}, if $t(p_l)=t(p_r)=1$;
	 $H_\mu$ is \emph{L-shaped}, or \emph{\L-shaped}, if $t(p_l)=3$ and $t(p_r)=1$, or vice versa; 
	 $H_\mu$ is \emph{C-shaped}, or \emph{\C-shaped}, if $t(p_l)=4$ and $t(p_r)=2$, or vice versa.
	 \end{description}

\smallskip

The next theorem identifies, for each component $H_{\mu}$ with respect to $e$, the set of its representative shapes. An implication of the theorem is the existence of a bend-minimum orthogonal representation of any planar 3-graph distinct from $K_4$ such that there is at most one bend per edge.

\begin{restatable}{theorem}{thShapes}\label{th:shapes}
	Let $G$ be a biconnected planar $3$-graph distinct from $K_4$. $G$ admits a bend-minimum orthogonal representation $H$ such that for every edge $e$ in the external face of $H$ the following properties hold for the orthogonal components of $H$ with respect to $e$.
	
	\smallskip\noindent{\em \bf \textsf{O1}} Every Q-component is either \zeroB- or \oneB-shaped, that is, every edge has at most one bend.
	
	\noindent{\em \bf \textsf{O2}} Every P-component or R-component is either \D- or \X-shaped if it is an inner component and it is either \L- or \C-shaped if it is the root child component.
	
	\noindent{\em \bf \textsf{O3}} Every S-component has spirality at most four.
	
	\noindent{\em \bf \textsf{O4}} Every component has the minimum number of bends within its shape.
\end{restatable}

We remark that Kant shows that every planar 3-graph (except $K_4$) has an orthogonal representation with at most one bend per edge~\cite{DBLP:journals/algorithmica/Kant96}, but the total number of bends is not guaranteed to be the minimum. On the other hand, in~\cite{DBLP:conf/gd/DidimoLP18} it is shown how to compute a bend-minimum orthogonal representation of a planar 3-graph in the variable embedding setting with constrained shapes for its orthogonal components, but there can be more than one bend per edge. \cref{fi:opt-orth-bend-min-2-bends,fi:opt-orth-1-bend,fi:opt-orth-bend-min-1-bend} show
different orthogonal representations of the same planar $3$-graph, with different levels of optimality in terms of edge bends. The representation in \cref{fi:opt-orth-bend-min-2-bends} is optimal in terms of total number of bends but has some edges with two bends. The representation in \cref{fi:opt-orth-1-bend} has at most one bend per edge, but it does not minimize the total number of bends. The representation in \cref{fi:opt-orth-bend-min-1-bend} is optimal both in terms of total number of bends and in terms of maximum number of bends per edge.

\begin{figure}[tb]
	\centering
	\subfloat[]{\label{fi:opt-orth-bend-min-2-bends}\includegraphics[width=0.25\columnwidth]{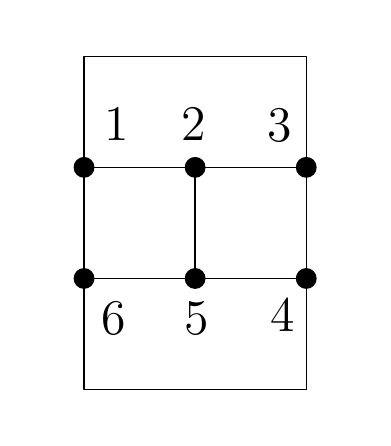}}
	\hfill
	\subfloat[]{\label{fi:opt-orth-1-bend}\includegraphics[width=0.25\columnwidth]{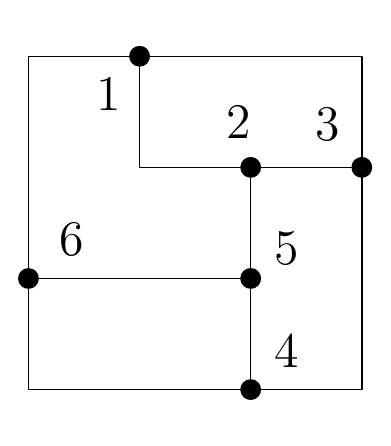}}
	\hfill
	\subfloat[]{\label{fi:opt-orth-bend-min-1-bend}\includegraphics[width=0.25\columnwidth]{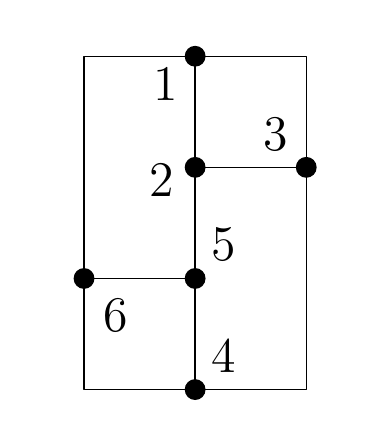}}
	\caption{(a) Bend-minimum orthogonal representation with at most 2 bends per edge. (b) Orthogonal representation with at most 1 bend per edge that is not bend-minimum. (c) Optimal orthogonal representation.}\label{fi:opt-orth}
\end{figure}

\myparagraph{Second ingredient: Labeling algorithm and Bend-Counter.} The second ingredient for the proof of Theorem~\ref{th:main} is a linear-time labeling algorithm that applies to 1-connected planar 3-graphs. Each edge $e$ of a block (biconnected component) $B$ of $G$ is labeled with the number $b(e)$ of bends of an $e$-constrained bend-minimum orthogonal representation of~$B$ with at most one bend per edge, if such a representation exists; if all $e$-constrained bend-minimum orthogonal representations require two bends on some edge we set $b(e)$ equal to~$\infty$. The labeling easily extends to the vertices of~$B$. Namely, for each vertex $v$ of~$B$, $b(v)$ is the minimum of the labels associated with the edges of~$B$ incident to $v$.
The labeling of the vertices is used in the drawing algorithm when we compose the orthogonal representations of the blocks of a 1-connected graph. We also label each block $B$ of $G$ with the number of bends $b(B)$ of an optimal $B$-constrained orthogonal representation of $G$, i.e., an orthogonal representation of $G$ such that at least one edge of $B$ is on the external face (such a representation always exists if $G$ is different from $K_4$). 

Let $T$ be the SPQR-tree of a block of $G$ rooted at an arbitrary edge~$e^*$. We perform a first $O(n)$-time bottom-up visit $T$, which can be regarded as a pre-processing step. Based on Theorem~\ref{th:shapes}, for each node $\mu$ of $T$ it suffices to consider its $O(1)$ representative shapes and, for each such shape, compute the minimum number of bends. Namely, for each node $\mu$ and for each representative shape $\sigma$ of $\mu$ (i.e., those in Theorem~\ref{th:shapes}), we compute the minimum number of bends $b_{e^*}^{\sigma}(\mu)$ of an orthogonal component $H_{\mu}$ with representative shape~$\sigma$. When $\mu$ is the root of $T$ (i.e., $\mu$ is the Q-node associated with $e^*$), the label of $e^*$ is $b(e^*)=\min\{b_{e^*}^{\zerob}(\mu),b_{e^*}^{\oneb}(\mu)\}$, where $b_{e^*}^{\zerob}(\mu)$ (resp. $b_{e^*}^{\oneb}(\mu)$) corresponds to the number of bends of an optimal $e^*$-constrained representation of $G$ with $e^*$ drawn with zero bends (resp. one bend).  
After this pre-processing step, for every other edge $e \neq e^*$, the values $b_e^{\sigma}(\mu)$ are efficiently computed by rooting $T$ at the Q-node associated with $e$ and by performing a bottom-up visit of $T$. For each node $\mu$, the values $b_e^{\sigma}(\mu)$ are computed in $O(1)$ time.   
%
%
%
%
The computation of $b_{e}^{\sigma}(\mu)$ is particularly tricky when $\mu$ is an R-node, whose skeleton is a triconnected cubic planar graph. In this case, each virtual edge of $\mu$ always corresponds to an S-component $\nu$; the spirality of $\nu$ can be increased up to a certain value without introducing extra bends along the real edges of $\skel(\nu)$.
For this property, a virtual edge of $\skel(\mu)$ is called a `flexible edge'. We aim at computing an orthogonal representation of $\skel(\mu)$ of minimum cost, where the cost is the number of bends on the real edges of $\skel(\mu)$ plus the extra bends possibly added to its flexible edges.
We use the following.

\begin{theorem}\label{th:bend-counter}
	Let $G$ be an $n$-vertex planar triconnected cubic graph, which may have some flexible edges.
	There exists a data structure that can be computed in $O(n)$ time and such that, for each possible face $f^*$ of $G$, it returns in $O(1)$ time the cost of a cost-minimum orthogonal representation of $G$ with $f^*$ as external face. 
\end{theorem}

The data structure of Theorem~\ref{th:bend-counter}, called \texttt{Bend-Counter}, is described in \cref{se:triconnected}.
We recall that the problem of computing orthogonal drawings of graphs with flexible edges is also studied by Bl\"asius et al.~\cite{DBLP:journals/algorithmica/BlasiusKRW14,DBLP:journals/comgeo/BlasiusLR16}, who however consider computational questions different from the ones in this paper.
%
%
Theorem~\ref{th:bend-counter} is used in the proof of the following.


\begin{theorem}\label{th:key-result-2}
	Let $G$ be a biconnected planar $3$-graph. 
    There is an $O(n)$-time algorithm that labels each edge $e$ of $G$ with the number $b(e)$ of bends of an optimal $e$-constrained orthogonal representation of $G$, if such a representation exists.
\end{theorem}

Finally, we extend the ideas of Theorem~\ref{th:key-result-2} to label the blocks in a block-cut-vertex tree of a 1-connected planar 3-graph $G$. 

\begin{theorem}\label{th:1-connected-labeling}
	Let $G$ be a 1-connected planar $3$-graph different from $K_4$. 
	There exists an $O(n)$-time algorithm that labels each block $B$ of $G$ with the number $b(B)$ of bends of an optimal $B$-constrained orthogonal representation~of~$G$.
\end{theorem}

\myparagraph{Third ingredient: Computing the drawing.} The third ingredient is the drawing algorithm. When $G$ is biconnected, we use Theorem~\ref{th:key-result-2} and choose an edge~$e$ such that $b(e)$ is minimum (the label of all the edges is~$\infty$ only when $G = K_4$). We then construct an optimal orthogonal representation of $G$ with $e$ on the external face by visiting the SPQR-tree of $G$ rooted at~$e$. We prove the following. 

\begin{restatable}{theorem}{thGdEnhanced}\label{th:gd2018-enhanced}
	Let $G$ be a biconnected planar $3$-graph different from $K_4$ and let $e$ be an edge of $G$ whose label $b(e)$ is minimum. There exists an $O(n)$-time algorithm that computes an optimal 
	orthogonal representation of~$G$ with $b(e)$ bends.
\end{restatable}

For 1-connected graphs, we use the next theorem to suitably merge the orthogonal representations of the different blocks of the block-cut-vertex tree.

\begin{restatable}{theorem}{thgdenhanced-v}\label{th:gd2018-enhanced-v}
	Let $G$ be a biconnected planar 3-graph and $v$ be a designated vertex of $G$ with $\deg(v) \leq 2$. An optimal $v$-constrained orthogonal representation $H$ of $G$ can be computed in $O(n)$ time. Also, $H$ has an angle larger than $90^\circ$ at $v$ on the external face.
\end{restatable}

\myparagraph{Proof of Theorem \ref{th:main}.}
	We use Theorem~\ref{th:1-connected-labeling} and choose a block $B$ such that $b(B)$ is minimum. We root the block-cut-vertex tree (BC-tree) $\cal T$ of $G$ at $B$ and compute an optimal orthogonal representation $H$ of $B$ by using Theorem~\ref{th:gd2018-enhanced}. 
	Let $v$ be a cut-vertex of $G$ that belongs to $H$ and let $B_v$ be a child block of $v$ in $\cal T$. Denote by $H_v$ an optimal $v$-constrained orthogonal representation of $B_v$. 
	Since $deg(v) \leq 2$ in $B_v$, by Theorem~\ref{th:gd2018-enhanced-v} we can assume that the angle at $v$ on the external face of $H_v$ is larger than $90^\circ$.
	
	Since $deg(v) \leq 2$ in $H$ there is a face of $H$ where $v$ forms an angle larger than $90^\circ$. Also, if $deg(v) = 2$ in $H$ then $B_v$ is a trivial block and if $deg(v) = 1$ in $H$ then $B$ is a trivial block. It follows that $H_v$ can always be inserted into the face of $H$ where $v$ forms an angle larger than $90^\circ$, yielding an optimal orthogonal representation of the graph formed by the two blocks $B$ and $B_v$. 
	Any other block of $G$ can be added by recursively applying this procedure, so to get an optimal orthogonal representation of $G$ with $b(B)$ bends. 
	
	Since (i) computing the labels of all blocks of $G$ requires $O(n)$ time (Theorem~\ref{th:1-connected-labeling}); (ii) computing an optimal orthogonal representation for the root block $B$ requires time proportional to the size of $B$ (Theorem~\ref{th:gd2018-enhanced}); and (iii) computing an optimal $v$-constrained orthogonal representation of each block $B_v$ requires time proportional to the size of $B_v$ (Theorem~\ref{th:gd2018-enhanced-v}), the theorem follows. 
	
	\medskip
	The rest of the paper is devoted to the proofs of Theorems~\ref{th:shapes}--\ref{th:gd2018-enhanced-v}.

\section{Representative~Shapes:~Proof~of~Theorem~\ref{th:shapes}}\label{se:thshapes}
To prove Theorem~\ref{th:shapes} we first concentrate on Property~\textsf{O1} (\cref{sse:O1}) and then we discuss Properties~\textsf{O2}-\textsf{O4} (\cref{sse:O2-O4}), whose proof is mostly inherited by~\cite{DBLP:conf/gd/DidimoLP18}. 
%
We need a few further definitions. Given an orthogonal representation $H$, we denote by $\rect{H}$ the orthogonal representation obtained from $H$ by replacing each bend with a dummy vertex. $\rect{H}$ is called the \emph{rectilinear image} of $H$ and a dummy vertex in $\rect{H}$ is a \emph{bend vertex}. By definition $b(\rect{H})=0$. The representation $H$ is also called the \emph{inverse} of $\rect{H}$.
If $w$ is a degree-$2$ vertex with neighbors $u$ and $v$, \emph{smoothing} $w$ is the reverse operation of an edge subdivision, i.e., it replaces the two edges $(u,w)$ and $(w,v)$ with the single edge $(u,v)$. In particular, if $H$ is an orthogonal representation of a graph $G$ and $\rect{G}$ is the underlying graph of $\rect{H}$, the graph $G$ is obtained from $\rect{G}$ by smoothing all its bend vertices.

\subsection{Proof of Property~\textsf{O1}.}\label{sse:O1}
We prove that, if $G$ is a biconnected graph distinct from the complete graph $K_4$, for any desired designated vertex $v$ of $G$, there always exists a $v$-constrained bend-minimum representation $H$ of $G$ with at most one bend per edge (\cref{le:1-bend}). Clearly, the $v$-constrained orthogonal representation that has the minimum number of bends over all possible choices for the vertex $v$, satisfies Property~\textsf{O1}.   

We start by rephrasing an intermediate result already proved in~\cite{DBLP:conf/gd/DidimoLP18}, which shows the existence of a bend-minimum orthogonal representation with at most two bends per edge for any biconnected planar $3$-graph and for any designated edge on the external face. For completeness we report the entire proof of this result, revised according to our new terminology.

\begin{lemma}\label{le:2-bends}
	Let $G$ be a biconnected planar $3$-graph and let $e$ be a designated edge of~$G$. There exists an $e$-constrained bend-minimum orthogonal representation of $G$ with at most two bends per edge.
\end{lemma}
\begin{proof}
	Let $H$ be an $e$-constrained bend-minimum orthogonal representation of $G$ and suppose that there is an edge $g$ of $H$ (possibly coincident with $e$) with at least three bends. Let $\rect{H}$ be the rectilinear image of $H$ and $\rect{G}$ its underlying plane graph. Since $b(\rect{H})=0$, $\rect{G}$ is a good plane graph. Denote by $v_1$, $v_2$, and $v_3$ three bend vertices in $\rect{H}$ that correspond to three bends of $g$ in $H$. We distinguish between two cases.
	
	\smallskip\paragraph{Case 1: $g$ is an internal edge of $H$.} Let $\rect{G'}$ be the plane graph obtained from $\rect{G}$ by smoothing $v_1$. $\rect{G'}$ is still a good plane graph. Indeed, if this were not the case there would be a bad cycle in $\rect{G'}$ that contains both $v_2$ and $v_3$, which is impossible because no bad cycle can contain two vertices of degree two. Hence, there exists an (embedding-preserving) orthogonal representation $\rect{H'}$ of $\rect{G'}$ without bends; the inverse $H'$ of $\overline{H'}$ is a representation of $G$ such that $b(H') < b(H)$, a contradiction.
	
	\smallskip\paragraph{Case 2: $g$ is an external edge of $H$.} If $C_o(\rect{G})$ contains more than four vertices of degree two, then we can smooth vertex $v_1$ and apply the same argument as above to contradict the optimality of $H$ (note that, such a smoothing does not violate Condition $(i)$ of Theorem~\ref{th:RN03}). Suppose vice versa that $C_o(\rect{G})$ contains exactly four vertices of degree two (three of them being $v_1$, $v_2$, and $v_3$). In this case, just smoothing $v_1$ violates Condition~$(i)$ of Theorem~\ref{th:RN03}. However, we can smooth $v_1$ and subdivide an edge of $C_o(\rect{G}) \cap C_o(G)$; such an edge corresponds to an edge with no bend in $H$, and it exists because $C_o(G)$ has at least three edges and, by hypothesis, at most four bends, three of which on the same edge. The resulting plane graph $\rect{G''}$ still satisfies the three conditions of Theorem~\ref{th:RN03} and admits a representation $\rect{H''}$ without bends; the inverse of $\rect{H''}$ is a bend-minimum orthogonal representation of $G$ with at most two bends per edge.
\end{proof}

Note that, if $v$ is any vertex of $G$, \cref{le:2-bends} holds in particular for any edge $e$ incident to $v$. Thus, the following corollary immediately holds by iterating \cref{le:2-bends} over all edges incident to $v$ and by retaining the bend-minimum representation.

\begin{corollary}\label{co:2-bends}
	Let $G$ be a biconnected planar $3$-graph and let $v$ be a designated vertex of $G$. There exists a $v$-constrained bend-minimum orthogonal representation of $G$ with at most two bends per edge.
\end{corollary}

The next lemma will be used to prove the main result of this section, but it is also of independent interest. It claims that for a biconnected graph with at least five vertices, a bend-minimum orthogonal representation with at most two bends per edge can be searched among those planar embeddings that have an external face with at least four vertices (i.e, the external face is not a 3-cycle), even under the constraint that a given vertex is designated to be on the external face.

\begin{lemma}\label{le:external-face}
	Let $G$ be a biconnected planar $3$-graph with $n \geq 5$ vertices and let $v$ be a designated vertex of $G$. There exists a $v$-constrained bend-minimum orthogonal representation of $G$ with at most two bends per edge and at least four vertices on the external face.
\end{lemma}
\begin{proof}	
	By \cref{co:2-bends} there exists a $v$-constrained bend-minimum orthogonal representation $H$ of $G$ with at most two bends per edge. Embed $G$ in such a way that its planar embedding coincides with the planar embedding of $H$.  If the external face of $G$ contains at least four vertices, the statement holds. Otherwise, the external boundary of $G$ is a 3-cycle with vertices $u$, $v$, $w$ and edges $e_{uv}$, $e_{vw}$, $e_{wu}$ ($v$ is the designated vertex). Let $\rect{G}$ be the underlying plane graph of the rectilinear image $\rect{H}$ of $H$. Recall that since $\rect{H}$ has no bends, $\rect{G}$ is a good plane graph.
	For an edge $e$ of $G$, denote by $\rect{e}$ the subdivision of $e$ with bend vertices in $\rect{G}$  (if $e$ has no bend in $H$, then $e$ and $\rect{e}$ coincide). Since $G$ is biconnected, at least two of its three external vertices have degree three. The following (non-symmetric) cases are possible:
	
	\smallskip\paragraph{Case 1: $\deg(u)=\deg(v)=\deg(w)=3$.} Refer to \cref{fi:external-face-1}. In this case, $H$ has at least four bends on the external face, and hence two of them are on the same edge. Denote as $e_u$, $e_v$, and $e_w$ the internal edges incident to $u$, $v$, and $w$, respectively. Since $G$ is not $K_4$, at most two of $e_u$, $e_v$, and $e_w$ can share a vertex. Assume that $e_v$ does not share a vertex with $e_u$ (otherwise, we relabel the vertices, exchanging the identity of $u$ and $w$). Also, without loss of generality, we can assume that one of the two edges incident to $v$, say $e_{uv}$, has two bends. Indeed, if this not the case, $e_{wu}$ has two bends and we can simply move one of these two bends from $e_{wu}$ to $e_{uv}$. Since $G$ cannot contain 2-extrovert cycles that contain an external edge, this transformation still guarantees that the resulting plane graph is good. Let $e_u$ and $e_v$ be the internal edges of $G$ incident to $u$ and to $v$, respectively. Consider the plane graph $\rect{G'}$ obtained from $\rect{G}$ by rerouting $\rect{e_{uv}}$ in such a way that $w$ becomes an internal vertex. 
	
	If at least one among $e_u$ and $e_v$ has a bend in $H$, then $\rect{G'}$ is a good plane graph. Namely: The external face of $\rect{G'}$ still contains at least four vertices of degree two; the new 2-extrovert cycle passing through $u$, $v$, and $w$ contains at least two bend vertices (e.g., those of $\rect{e_{uv}}$); any other 2- or 3-extrovert cycle of $\rect{G'}$ is also present in $\rect{G}$ and contains in $\rect{G'}$ the same number of degree-2 vertices as in  $\rect{G}$. Therefore, in this case, $\rect{G'}$ has an embedding-preserving orthogonal representation $\rect{H'}$ without bends, and the inverse $H'$ of $\rect{H'}$ is a $v$-constrained bend-minimum orthogonal representation of $G$ with at most two bends per edge. This because $v$ is still on the external face of $H'$ and each edge of $G$ has the same number of bends in $H$ and in $H'$. Also, $H'$ has at least four vertices on the external face.
	
	Suppose now that neither $e_u$ nor $e_v$ has a bend in $H$. We know that at least one among $e_{wu}$ and $e_{vw}$ has a bend in $H$, say $e_{wu}$. Let $\rect{G''}$ be the plane graph obtained from $\rect{G'}$ by smoothing a bend vertex of $\rect{e_{wu}}$ and by subdividing $e_u$ with a bend vertex. This guarantees that $\rect{G''}$ has at least four vertices of degree two on the external face. Also, since each 2- or 3-extrovert cycle of $\rect{G''}$ that passes through $u$ but not through $\rect{e_{uv}}$ must contain $\rect{e_{wu}}$ and $\rect{e_u}$, then $\rect{G''}$ is still a good plane graph and admits a rectilinear orthogonal representation $\rect{H''}$. With the same arguments as in the previous case, the inverse $H''$ of $\rect{H''}$ is a $v$-constrained bend-minimum orthogonal representation of $G$ with at most two bends per edge and at least four vertices on the external face.                 
	
	\begin{figure}[tb]
		\centering
		\subfloat[Case 1]{\label{fi:external-face-1}\includegraphics[width=\columnwidth]{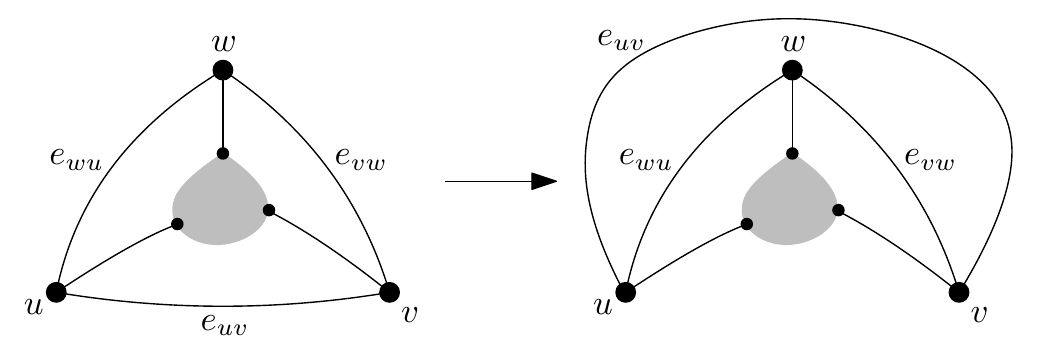}}
		\hfil
		\subfloat[Case 2]{\label{fi:external-face-2}\includegraphics[width=\columnwidth]{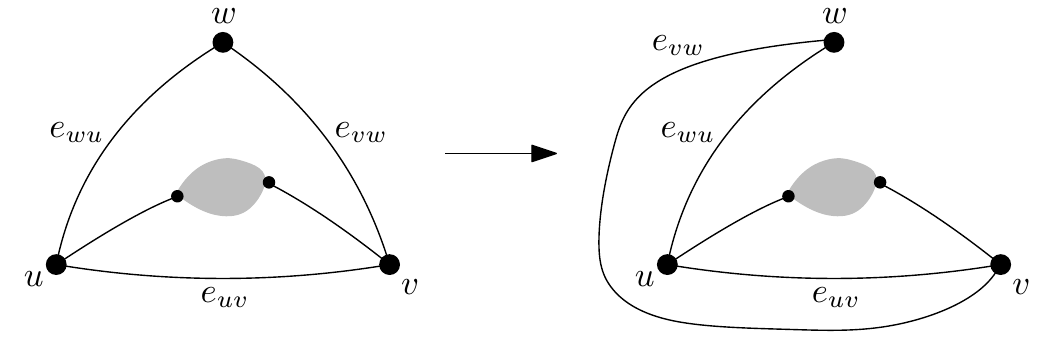}}
		\hfil
		\subfloat[Case 3]{\label{fi:external-face-3}\includegraphics[width=\columnwidth]{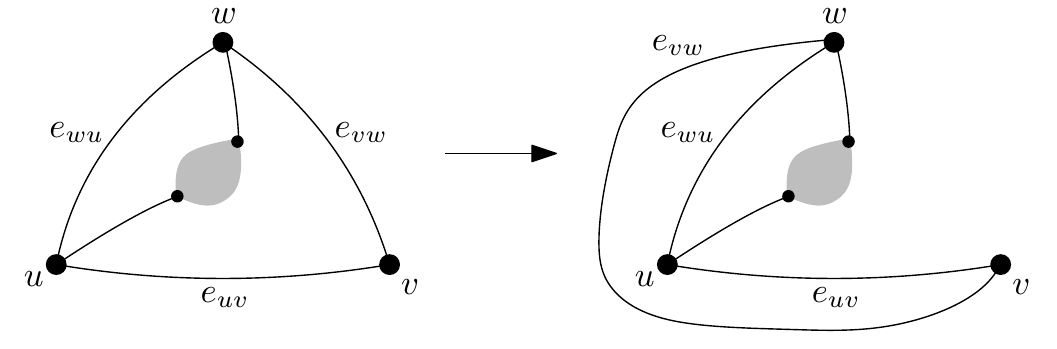}}
		\caption{Illustration for the proof of \cref{le:external-face}.}\label{fi:external-face}
	\end{figure}

	\smallskip\paragraph{Case 2: $\deg(u)=\deg(v)=3$ and $\deg(w)=2$.} Refer to \cref{fi:external-face-2}. In this case $H$ must have at least three bends on the external face. Let $e_u$ and $e_v$ be the internal edges of $G$ incident to $u$ and to $v$, respectively. Consider the plane graph $\rect{G'}$ obtained from $\rect{G}$ by rerouting $\rect{e_{vw}}$ in such a way that $u$ becomes an internal vertex. We must consider two subcases:
	\begin{itemize}
		\item Each of the external edges $e_{uv}$, $e_{vw}$, $e_{wu}$ of $G$ has a bend in $H$. If at least one among $e_u$ and $e_v$ has a bend in $H$, then $\rect{G'}$ remains a good plane graph and has a rectilinear representation $\rect{H'}$. The inverse $H'$ of $\rect{H'}$ is a $v$-constrained bend-minimum orthogonal representation of $G$ with at most two bends per edge and at least four external vertices. If neither $e_u$ nor $e_v$ has a bend in $H$, then, with the same argument as above, we can move a bend vertex from $\rect{e_{uv}}$ to $e_u$, i.e., we smooth a bend vertex from $\rect{e_{uv}}$ and subdivide $e_u$ with a bend vertex. The resulting plane graph is still good, and from it we can still get a $v$-constrained bend-minimum orthogonal representation of $G$ with at most two bends per edge and at least four external vertices. 
		
		\item One of the external edges $e_{uv}$, $e_{vw}$, $e_{wu}$ of $G$ has no bend in $H$. In this case, at least one of these three edges has two bends and another one may have no bend. Assume for example that $e_{uv}$ has two bends and $e_{wu}$ has no bend (the other cases are handled similarly or they are easier). If $e_u$ (resp. $e_v$) has no bend in $H$, we move one of the two bend vertices of $\rect{e_{uv}}$ on $e_u$ (resp. $e_v$). As in the previous cases, this transformation guarantees that the resulting plane graph $\rect{G''}$ is a good plane graph, and from it we get a $v$-constrained bend-minimum orthogonal representation of $G$ with at most two bends per edge and at least four external vertices.
	\end{itemize}

     \smallskip\paragraph{Case 3: $\deg(u)=\deg(w)=3$ and $\deg(v)=2$.} Refer to \cref{fi:external-face-3}. Also in this case $H$ must have at least three bends on the external face. Let $e_u$ and $e_w$ be the internal edges of $G$ incident to $u$ and to $w$, respectively. Consider again the plane graph $\rect{G'}$ obtained from $\rect{G}$ by rerouting $\rect{e_{vw}}$ in such a way that $u$ becomes internal. The analysis follows the line of Case 2, where the roles of $v$ and $w$ are exchanged.   
\end{proof}

The next step towards \cref{le:1-bend} are two technical results, namely \cref{le:internal-edge} and \cref{le:external-edge}.
They can be used to prove that, given a $v$-constrained bend-minimum orthogonal representation of a biconnected $3$-graph with more than four vertices and at most \emph{two} bends per edge (which exists by \cref{co:2-bends}), we can iteratively transform it into a new $v$-constrained bend-minimum orthogonal representation with at most \emph{one} bend per edge. The transformation of \cref{le:internal-edge} is used to remove internal edges with two bends, while \cref{le:external-edge} is used to remove external edges with two bends.

\begin{lemma}\label{le:internal-edge}
	Let $G$ be a biconnected planar $3$-graph with $n \geq 5$ vertices, $v$ be a designated vertex of $G$, and $H$ be a $v$-constrained bend-minimum orthogonal representation of $G$ with at most two bends per edge and at least four vertices on the external face. If $e$ is an \emph{internal} edge of $H$ with two bends, there exists a new $v$-constrained bend-minimum orthogonal representation $H^*$ of $G$ such that: (a) $e$ has at most one bend in $H^*$; (b) every edge $e' \neq e$ has at most two bends in $H^*$, and $e'$ has two bends in $H^*$ only if it has two bends in $H$; (c) $H^*$ has at least four vertices on the external face.
\end{lemma}
\begin{proof}
	Let $\rect{H}$ be the rectilinear image of $H$ and denote by $v_1$ and $v_2$ the bend vertices of $H$ associated with the bends of $e$.  By Theorem~\ref{th:RN03} and since $H$ has the minimum number of bends, $e$ necessarily belongs to a 2-extrovert cycle $C$ of $H$. Indeed, if $e$ does not belong to a 2-extrovert cycle, then we can smooth from the underlying plane graph $\rect{G}$ of $\rect{H}$ one of $v_1$ and $v_2$. The resulting plane graph $\rect{G'}$ is a good plane graph and then it admits an orthogonal representation $\rect{H'}$ without bends; the inverse $H'$ of $\rect{H'}$ is an orthogonal representation of $G$ with less bends than $H$, a contradiction. In the following, we call \emph{free edge} an edge of $G$ without bends in $H$. We distinguish between three cases:
	
	\smallskip\paragraph{Case 1: $C$ does not share $e$ with other 2-extrovert cycles of $H$.} In this case, $H$ must have a free edge $g$ that belongs to $C$, otherwise smoothing one of $v_1$ and $v_2$ from $\rect{G}$ we obtain a good plane graph; this would imply the existence of an orthogonal representation of $G$ with less bends than $H$, a contradiction. Now, consider the plane graph $\rect{G^*}$ obtained from $\rect{G}$ by smoothing $v_1$ and by subdividing $g$ with a new (bend) vertex. $\rect{G^*}$ is a good plane graph and thus it admits an orthogonal representation $\rect{H^*}$ without bends. The inverse $H^*$ of $\rect{H^*}$ is an orthogonal representation of $G$ that satisfies Properties~(a) and~(b). Also $b(H^*)=b(H)$, thus $H^*$ is bend-minimum. Finally, since $H^*$ has the same planar embedding as $H$, $H^*$ is $v$-constrained and Property~(c) is also guaranteed.
	
	\smallskip\paragraph{Case 2: $C$ shares $e$ and at least another edge with a 2-extrovert cycle $C'$ of $H$.} In this case, $C$ and $C'$ must share a free edge $g$, otherwise, as in the previous case, smoothing one of $v_1$ and $v_2$ from $\rect{G}$ we obtain a good plane graph and, in turns, an orthogonal representation of $G$ with less bends than $H$; a contradiction. As above, $H^*$ is obtained from $H$ by removing a bend from $e$ and by adding a bend along $g$; also $H^*$ still has the same planar embedding as $H$. 
	
	\begin{figure}[tb]
		\centering
		\subfloat[]{\label{fi:nested}\includegraphics[width=0.33\columnwidth]{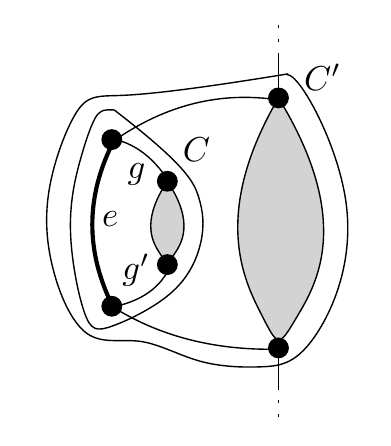}}
		\hfil
		\subfloat[]{\label{fi:nested-1}\includegraphics[width=0.33\columnwidth]{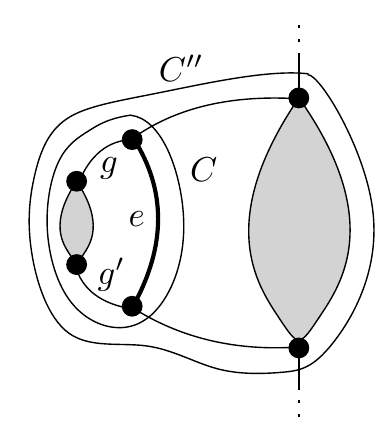}}
		\hfil
		\subfloat[]{\label{fi:interlaced}\includegraphics[width=0.33\columnwidth]{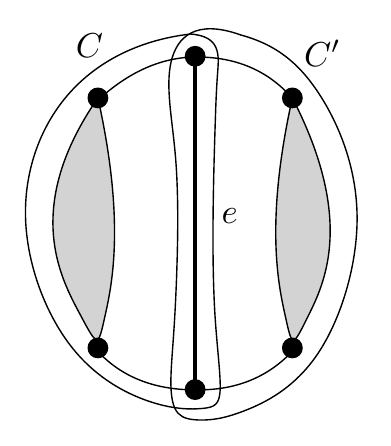}}
		\hfil
		\caption{Illustration for the proof of \cref{le:internal-edge}. (a) Two nested 2-extrovert cycles $C$ and $C'$ that share edge $e$ only. (b) Flipping $C$ at its leg vertices. (c) Two interlaced 2-extrovert cycles $C$ and $C'$ that share edge $e$ only; the external face of the graph consists of $(C \cup C') \setminus \{e\}$.}\label{fi:internal}
	\end{figure}

	\smallskip\paragraph{Case 3: $C$ shares only $e$ with a 2-extrovert cycle $C'$ of $H$.} In this case, $C$ and $C'$ can have two  possible configurations in the embedding of $H$:
		\begin{itemize}
			\item $C$ and $C'$ are \emph{nested}, i.e., one of the two cycles is inside the other; refer to \cref{fi:nested}). Without loss of generality, assume that $C$ is inside $C'$ (the argument is symmetric in the opposite case). Let $g$ and $g'$ be the two edges of $C$ adjacent to $e$. We first claim that each of these two edges is free in $H$. Indeed, suppose by contradiction that $g$ has a bend and let $\rect{G}$ be the underlying plane graph of $\rect{H}$. Since $G$ has vertex-degree at most three, $g$ cannot belong to other $2$-extrovert or $3$-extrovert cycles other than $C$. Thus, smoothing from $\rect{G}$ the bend vertex associated with the bend of $g$ we obtain a good plane graph, which contradicts the fact that $H$ is bend-minimum. The same argument applies to $g'$. Now, consider the plane graph $\rect{G''}$ obtained from $\rect{G}$ by flipping $C$ (namely, the cycle of $\rect{G}$ corresponding to $C$) at its leg vertices and let $C''$ be the new 2-extrovert cycle that has $C$ inside it; see \cref{fi:nested-1}). $C''$ consists of the edges of $(C' \cup C) \setminus \{e\}$.
			The other 2-extrovert and 3-extrovert cycles of $\rect{G''}$ stay the same as in $\rect{G}$. Consider the plane graph $\rect{G^*}$ obtained from $\rect{G''}$ by smoothing the two bend vertices $v_1$ and $v_2$, and by subdividing both $g$ and $g'$ (that was free in $H$) with a new (bend) vertex. Since $\rect{G^*}$ has two bend vertices along the path shared by $C$ and $C''$, and the rest of the 2-extrovert and 3-extrovert cycles are not changed with respect to $\rect{G''}$, $\rect{G^*}$ is a good plane graph. Thus $\rect{G^*}$ admits an orthogonal representation $\rect{H^*}$ without bends. The inverse $H^*$ of $\rect{H^*}$ is an orthogonal representation of $G$ that satisfies Properties~(a) and~(b), and since $b(H^*)=b(H)$, $H^*$ is still bend-minimum. Finally, all the vertices of $H$ that were on the external face remain on the external face of $H^*$. Therefore, $H^*$ is also $v$-constrained and Property~(c) is guaranteed.
			
			\item $C$ and $C'$ are \emph{interlaced}, i.e., none of the two cycles is inside the other; refer to \cref{fi:interlaced}). In this case, the external face of $H$ is formed by $(C \cup C') \setminus \{e\}$. Let $\rect{G}$ be the underlying plane graph of $\rect{H}$. By Condition~(i) of Theorem~\ref{th:RN03}, $\rect{G}$ has at least four degree-2 vertices on its external face (which can be real or bend vertices). We claim that all the degree-2 vertices are either along the external part of $C$ or along the external part of $C'$. Indeed, if the external part of $C$ contains exactly two degree-2 vertices, then the external part of $C'$ contains at least two degree-2 vertices. In this case the plane graph obtained from $\rect{G}$ by smoothing both the bend vertices associated with the bends of $e$ is still good, which contradicts the fact that $H$ is bend-minimum. Analogously, if $C$ contains exactly one degree-2 vertex, then the external part of $C'$ contains at least three degree-2 vertices; again, smoothing from $\rect{G}$ one of the two bend vertices associated with the bends of $e$ we obtain a good plane graph, a contradiction. Therefore, the only possibility is that either the external part of $C$ or the external part of $C'$ has no degree-2 vertices in $\rect{G}$. Without loss of generality assume that the external part of $C$ has no degree-2 vertices in $\rect{G}$. This implies that the external edges of $C$ are all free edges in $H$. If we smooth from $\rect{G}$ a bend vertex associated with a bend of $e$ and subdivide a free edge of $C$ with a new (bend) vertex, we obtain a good plane graph $\rect{G^*}$. Thus $\rect{G^*}$ admits an orthogonal representation $\rect{H^*}$ without bends. The inverse $H^*$ of $\rect{H^*}$ is an orthogonal representation of $G$ that satisfies Properties (a) and (b), and since $b(H^*)=b(H)$, $H^*$ is still bend-minimum. Also, since $H^*$ and $H$ have the same planar embedding, $H^*$ is still $v$-constrained and Property~(c) holds.
		\end{itemize} 
\end{proof}

\begin{lemma}\label{le:external-edge}
	Let $G$ be a biconnected planar $3$-graph with $n > 4$ vertices, $v$ be a designated vertex of $G$, and $H$ be a $v$-constrained bend-minimum orthogonal representation of $G$ with at most two bends per edge and at least four vertices on the external face. If $e$ is an \emph{external} edge of $H$ with two bends, there exists a new $v$-constrained bend-minimum orthogonal representation $H^*$ of $G$ such that: (a) $e$ has at most one bend in $H^*$; (b) every edge $e' \neq e$ has at most two bends in $H^*$ and $e'$ has two bends in $H^*$ only if it has two bends in $H$; (c) $H^*$ has at least four vertices on the external face.
\end{lemma}
\begin{proof}
	As in the proof of \cref{le:internal-edge}, a \emph{free edge} of $H$ is an edge without bends. Let $\overline{H}$ be the rectilinear image of $H$ and let $v_1$ and $v_2$ be the bend vertices of $H$ associated with the bends of $e$. Since $\overline{H}$ has no bends, its underlying graph $\overline{G}$ is a good plane graph. For simplicity, if $C$ is a cycle of $G$ we also call $C$ the cycle of $\overline{G}$ that corresponds to the subdivision of $C$ in $\overline{G}$. We distinguish between two cases:
	
	\smallskip\paragraph{Case 1: $e$ does not belong to a 2-extrovert cycle of $H$.} In this case, there must be at least a free edge on the external face of $H$. Indeed, suppose by contradiction that this is not true. By hypothesis $H$ as at least four external edges; if all these edges are not free, then there are at least five bends on the external boundary of $H$. Smoothing $v_1$ from $\rect{G}$ we get a resulting plane graph $\rect{G'}$ that is still a good plane graph, because by hypothesis $e$ does not belong to a 2-extrovert cycle of $H$ and because we still have four vertices of degree two on the external face of $\overline{G'}$. This implies that $\rect{G'}$ has an orthogonal representation $\rect{H'}$ without bends, and the inverse $H'$ of $\rect{H'}$ has less bends than $H$, a contradiction. Let $g$ be a free edge on the external face of $H$. Moving a bend from $e$ to $g$ we get the desired $v$-constrained orthogonal representation $H^*$ (note that $H^*$ has the same planar embedding as $H$).       
	
	\begin{figure}[tb]
		\centering
		\subfloat[]{\label{fi:external-edge-1}\includegraphics[width=0.38\columnwidth]{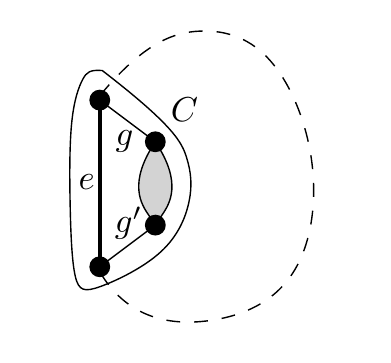}}
		\hfil
		\subfloat[]{\label{fi:external-edge-2}\includegraphics[width=0.38\columnwidth]{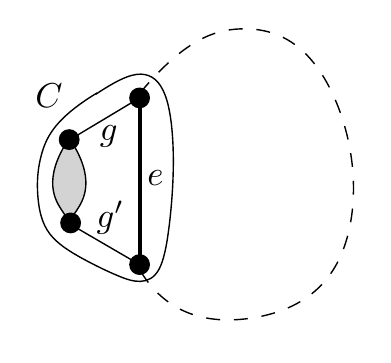}}
		\caption{Illustration for the proof of \cref{le:external-edge}. (a) A 2-extrovert cycle $C$ that shares exactly one edge $e$ on the external face; $g$ and $g'$ are free edges; the dashed curve represents the rest of the boundary of the external face. (b) Flipping $C$ around its leg vertices $g$ and $g'$ become external edges, and we can move the two bends of $e$ one on $g$ and the other on $g'$.}\label{fi:external}
	\end{figure}

	\smallskip\paragraph{Case 2: $e$ belongs to a 2-extrovert cycle $C$ of $H$.} We consider two subcases:
	\begin{itemize}
		\item $e$ is the only edge of $C$ on the external face (refer to \cref{fi:external-edge-1}). With similar arguments used in the proof of \cref{le:internal-edge}, it is easy to see that both the (internal) edges $g$ and $g'$ of $C$ incident to $e$ are free edges in $H$. Consider the plane graph $\rect{G'}$ obtained from $\rect{G}$ by flipping $C$ (namely, the cycle of $\rect{G'}$ corresponding to $C$) around its two leg vertices (see \cref{fi:external-edge-2}). The graph $\rect{G^*}$ obtained from $\rect{G'}$ by subdividing both $g$ and $g'$ with a vertex and by smoothing $v_1$ and $v_2$ is still a good plane graph. Hence, $\rect{G^*}$ admits a rectilinear orthogonal representation $\rect{H^*}$ without bends. The inverse $H^*$ of $\rect{H^*}$ has the same number of bends as $\rect{H}$. Also, edge $e$ has no bend in $H^*$, $g$ and $g'$ have one bend in $H^*$, and every other edge of $H^*$ has the same number of bends as in $H$. Finally, the external face of $H^*$ contains all the vertices of the external face of $H$. Therefore, $H^*$ is the desired $v$-constrained orthogonal representation.      
		
		\item $e$ has at least another edge $g \neq e$ on the external face. If $g$ is a free edge of $H$, then we can simply move a bend from $e$ to $g$, thus obtaining the desired $v$-constrained orthogonal representation $H^*$. Suppose now that $g$ is not a free edge. In this case there must be another free edge on the external face. Indeed, if all the edges of the external face of $G$ were not free, we could smooth $v_1$ from $\rect{G}$, and the resulting graph $\rect{G'}$ would be a good plane graph (recall that there are at least four edges on the external face and that $C$ has at least three bends in $H$ if $g$ is not free); if $\rect{H'}$ is an orthogonal representation of $\rect{G'}$ without bends, then its inverse $H'$ is an orthogonal representation of $G$ with less bends than $H$, a contradiction. Therefore, let $g'$ be a free edge on the external face. Again, we can simply move a bend from $e$ to $g'$, thus obtaining the desired $v$-constrained representation $H^*$.       
	\end{itemize} 
\end{proof}


\begin{figure}[tb]
	\centering
	\includegraphics[width=0.9\columnwidth]{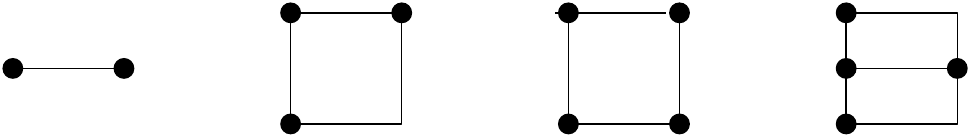}
	\caption{Bend-minimum orthogonal drawings with at most one bend per edge for a biconnected planar 3-graph distinct from $K_4$ and having at most four vertices.}\label{fi:upto4}
\end{figure}

\begin{lemma}\label{le:1-bend}
	Let $G$ be a biconnected planar $3$-graph distinct from $K_4$ and let $v$ be a designated vertex of $G$. There exists a $v$-constrained bend-minimum orthogonal representation of $G$ with at most one bend per edge. Also, if $v$ is a degree-2 vertex, such a representation can be constructed in such a way that $v$ has an angle larger than $90^\circ$ on the external face.  
\end{lemma}
\begin{proof}
	If $n \leq 4$ the statement trivially holds by choosing a planar embedding of $G$ with all the vertices on the external face; all the bend-minimum orthogonal representations with one bend per edge of non-isomorphic graphs are depicted in \cref{fi:upto4}.
	Suppose vice versa that $n > 4$. By \cref{le:external-face}, there exists a $v$-constrained bend-minimum orthogonal representation $H$ of $G$ with at most two bends per edge and at least four vertices on the external face. If all the vertices of $G$ have at most one bend in $H$, then we are done. Otherwise, starting from $H$ we can iteratively apply \cref{le:internal-edge} and \cref{le:external-edge} to yield a $v$-constrained bend-minimum orthogonal representation $H^*$ of $G$ with at most one bend per edge. In fact, thanks to Properties~(a) and~(b) of \cref{le:internal-edge} and \cref{le:external-edge}, at each iteration the number of edges with two bends decreases by at least one unit. Also, Property~(c) guarantees that we can continue the iterative process until $H^*$ is achieved. 
	Finally, suppose that $v$ is a degree-2 vertex and that $v$ has an angle of $90^\circ$ on the external face of $H^*$. Consider the underlying plane graph $\rect{G^*}$ of $\rect{H^*}$. Since $\rect{H^*}$ has no bend, $\rect{G^*}$ is a good plane graph. Rahman et al.~\cite{DBLP:journals/jgaa/RahmanNN03} show in particular that an embedding-preserving orthogonal drawing of $\rect{G^*}$ without bends can always be computed by arbitrarily choosing four degree-2 external vertices that will form reflex corners on the external face in the drawing, i.e., which will form angles of $270^\circ$ on the external face. We then compute such a representation $\rect{H^+}$ choosing $v$ as one of these four corners. The inverse $H^+$ of $\rect{H^+}$ is such that $b(H^+)=b(H^*)$ and each edge of $G$ has the same number of bends in $H^+$ and in $H^*$. Therefore, $H^+$ is a $v$-constrained bend-minimum orthogonal representation of $G$ with at most one bend per edge and with an angle larger than $90^\circ$ at $v$ on the external face.  
\end{proof}

\subsection{Proof of Properties~\textsf{O2-O4}.}\label{sse:O2-O4}

As mentioned at the beginning of this section, the proof of Properties~\textsf{O2-O3} is mostly inherited by~\cite{DBLP:conf/gd/DidimoLP18}, but it also exploits Properties~(c) and~(d) of the next technical lemma, i.e., \cref{le:shapes}.
On the other hand, Properties~(a) and~(b) of \cref{le:shapes} will be used later on in this paper to produce bend-minimum orthogonal representations of R-components with desired representative shapes. 

Let $G_\mu$ be a P- or an R-component of $G$ with a given planar embedding with the poles $\{u,v\}$ on the external face, and let $G^+_\mu$ be the plane graph obtained from $G_\mu$ by adding edge $e=(u,v)$ and choosing as external face one of the two faces incident to $e$.

\begin{restatable}{lemma}{leShapes}\label{le:shapes}
	Let $H^+_\mu$ be an orthogonal representation of $G^+_\mu$. The following properties hold:
	\begin{itemize}
		\item[(a)] If $e$ has at most 2 bends in $H^+_\mu$, then there exists an orthogonal representation $H^*_\mu$ of $G^+_\mu$ such that $e$ has 2 bends in $H^*_\mu$, $b(H^*_\mu \setminus e) \leq b(H^+_\mu \setminus e)$, and $H^*_\mu \setminus e$ is $\D$-shaped. 
		
		\item[(b)] If $e$ has exactly 3 bends in $H^+_\mu$, then there exists an orthogonal representation $H^*_\mu$ of $G^+_\mu$ such that $e$ has 3 bends in $H^*_\mu$, $b(H^*_\mu \setminus e) \leq b(H^+_\mu \setminus e)$, and $H^*_\mu \setminus e$ is $\X$-shaped.
		
		\item[(c)] If $e$ has exactly 1 bend in $H^+_\mu$, then there exists an orthogonal representation $H^*_\mu$ of $G^+_\mu$ such that $e$ has 1 bend in $H^*_\mu$, $b(H^*_\mu \setminus e) \leq b(H^+_\mu \setminus e)$, and $H^*_\mu \setminus e$ is $\L$-shaped.
		
		\item[(d)] If $e$ has 0 bends in $H^+_\mu$, then there exists an orthogonal representation $H^*_\mu$ of $G^+_\mu$ such that $e$ has 0 bends in $H^*_\mu$, $b(H^*_\mu \setminus e) \leq b(H^+_\mu \setminus e)$, and $H^*_\mu \setminus e$ is $\C$-shaped.
	\end{itemize}
	
	In all cases, every edge distinct from $e$ has in $H^*_\mu$ no more bends than it has in $H^+_\mu$ and $H^*_\mu$ can be computed from $H^+_\mu$ in linear time in the number of vertices of $G^+_\mu$.    
\end{restatable}
\begin{proof}
	We prove the different properties one by one.	
	
	\smallskip\paragraph{Property~(a).} Refer to \cref{fi:property-a} for a schematic illustration. Consider the rectilinear image $\rect{H^+_\mu}$ of $H^+_\mu$ and its underlying graph $\rect{G^+_\mu}$. By definition $\rect{H^+_\mu}$ has no bend, hence $\rect{G^+_\mu}$ is a good plane graph. Denote by $\rect{e}$ the subdivision of $e$ in $\rect{G^+_\mu}$ (note that $\rect{e}$ and $e$ coincide if $e$ has no bend).  
	Consider the plane graph $\rect{G_\mu} = \rect{G^+_\mu} \setminus \rect{e}$. Since $\mu$ is either a P- or an R-component, $\rect{G_\mu}$ is biconnected; also the end-vertices $u$ and $v$ of $e$ have degree two in $\rect{G_\mu}$, and $\rect{G_\mu}$ is also a good plane graph. 
	Rahman et al.~\cite[Corollary~6]{DBLP:journals/jgaa/RahmanNN03} describe a linear-time algorithm that constructs a no-bend orthogonal representation of $\rect{G_\mu}$ such that: $(i)$ four arbitrarily designated vertices of degree two in the external face are used as corners and $(ii)$ for each path $p$ of the external face between any two such consecutive corners in clockwise order, the turn number $t(p)$ is zero. We refer to this algorithm as \textsf{NoBend-Alg} in the following.
	Denote by $p'$ the path of the external face of $\rect{G^+_\mu}$ between $u$ and $v$ not containing $\rect{e}$. Clearly, $p'$ is also a path in $\rect{G_\mu}$. Since by hypothesis $e$ has at most 2 bends in $\rect{H^+_\mu}$, there exist at least two bend vertices $x$ and $y$ in $p'$, which correspond to degree-2 vertices of $\rect{G_\mu}$ distinct from $u$ and $v$.   
	We apply \textsf{NoBend-Alg} to construct an orthogonal representation $\rect{H'_\mu}$ of $\rect{G_\mu}$ having $u$, $v$, $x$, and $y$ as the four external designated corners. Since the external path from $u$ to $v$ in $\rect{H'_\mu}$ is zero and since $u$ and $v$ are corners, we have that the inverse $H'_\mu$ of $\rect{H'_\mu}$ is $\D$-shaped. 
	Therefore, we can construct the desired orthogonal representation $H^*_\mu$, by adding to $H'_\mu$ edge $e$ with two bends (turning in the same direction). By construction, $H^*_\mu \setminus e = H'_\mu$, hence $H^*_\mu \setminus e$ is $\D$-shaped. Also, the number of bends of $H^*_\mu \setminus e$ is not greater than the number of bends of $\rect{H^+_\mu} \setminus e$, because the bends of $H^*_\mu \setminus e$ can only correspond to degree-2 vertices of $\rect{H'_\mu}$ and the set of degree-2 vertices is the same in $\rect{H'_\mu}$ and in $\rect{H^+_\mu} \setminus \rect{e}$. 
	
	\begin{figure}[tb]
		\includegraphics[width=0.9\columnwidth]{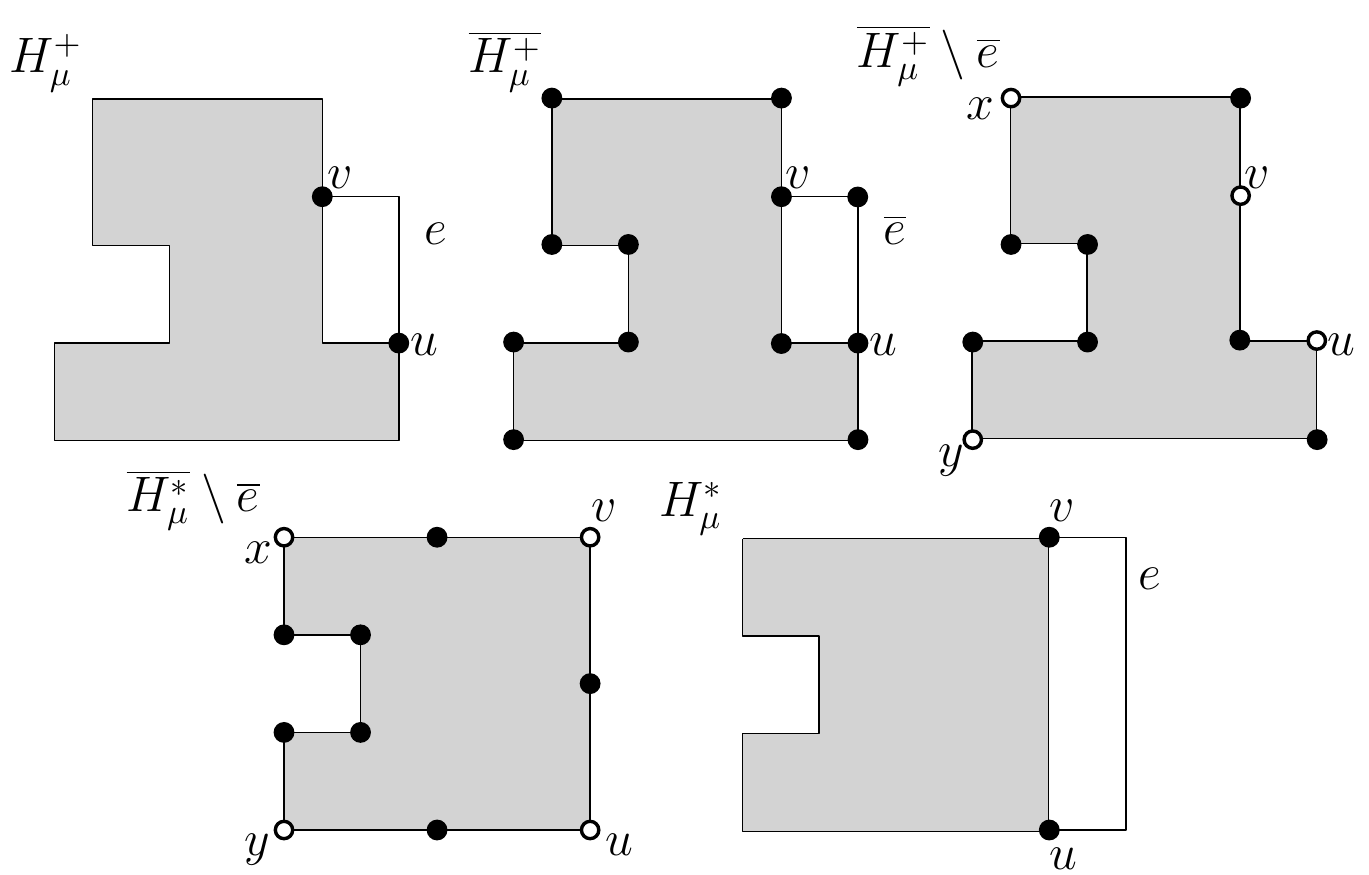}
		\caption{Illustration for the proof of Property (a).}\label{fi:property-a}
	\end{figure}
	
	\smallskip\paragraph{Property~(b).} Refer to \cref{fi:property-b} for a schematic illustration. We use the same notation as in the previous case. Denote by $p'$ and $p''$ the two paths between $u$ and $v$ on the external face of $\rect{G_\mu}$. Since by hypothesis $e$ has exactly 3 bends in $\rect{H^+_\mu}$, there exists at least one bend vertex $x$ along $p'$ and at least one bend vertex $y$ along $p''$. Both $x$ and $y$ are degree-2 vertices of $\rect{G_\mu}$ distinct from $u$ and $v$. We apply \textsf{NoBend-Alg} to construct an orthogonal representation $\rect{H'_\mu}$ of $\rect{G_\mu}$ having $u$, $x$, $v$, and $y$ as the four external designated corners. Since the external path between any two such consecutive corners is zero, we have that the turn number from $u$ to $v$ on the external face is equal to one in each of the two possible circular directions (clockwise and counterclockwise). This implies that the inverse $H'_\mu$ of $\rect{H'_\mu}$ is $\X$-shaped. Therefore, we can construct the orthogonal representation $H^*_\mu$, by adding to $H'_\mu$ edge $e$ with three bends. By construction, $H^*_\mu \setminus e = H'_\mu$, hence $H^*_\mu \setminus e$ is $\X$-shaped. Also, the number of bends of $H^*_\mu \setminus e$ is not greater than the number of bends of $\rect{H^+_\mu} \setminus \rect{e}$, because the bends of $H^*_\mu \setminus e$ can only correspond to degree-2 vertices of $\rect{H'_\mu}$.  
	
	\begin{figure}[tb]
		\includegraphics[width=0.9\columnwidth]{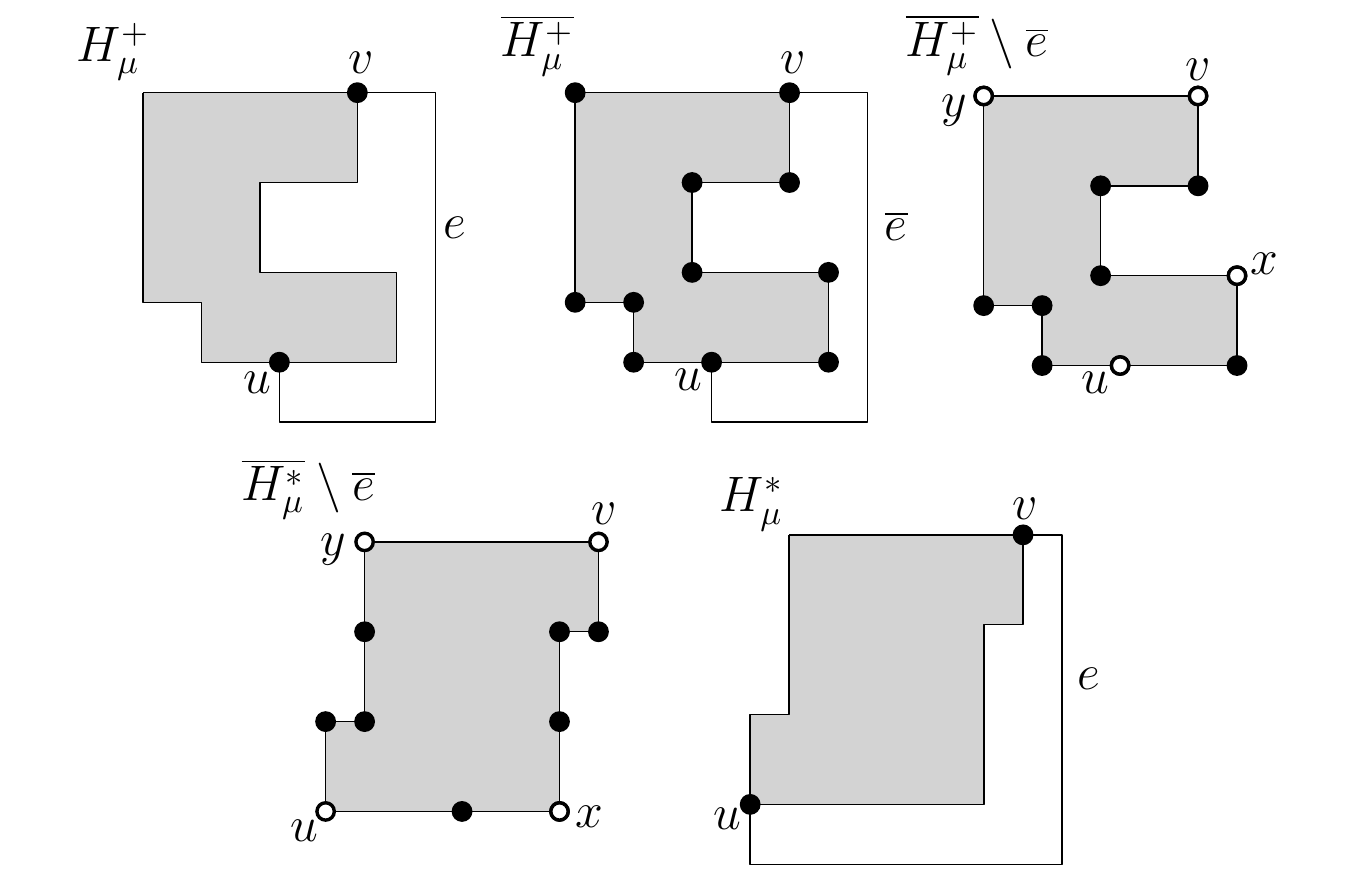}
		\caption{Illustration for the proof of Property (b).}\label{fi:property-b}
	\end{figure}
	
	\smallskip\paragraph{Property~(c).} Refer to \cref{fi:property-c} for a schematic illustration. We still use the same notation as in the previous cases. Denote by $p'$ the path of the external face of $\rect{G^+_\mu}$ between $u$ and $v$ not containing $\rect{e}$. Since by hypothesis $e$ has exactly 1 bend in $H^+_\mu$, there exist at least three bend vertices $x, z, y$ along $p'$ in $\rect{H^+_\mu}$, corresponding to degree-2 vertices of $\rect{G^+_\mu}$. In particular, we choose $x$ as the first degree-2 vertex encountered along $p'$ while moving counterclockwise from $v$ and we choose $y$ as the first degree-2 vertex along $p'$ in the clockwise direction from $u$. Vertex $z$ is an arbitrary degree-2 vertex along $p'$ between $x$ and $y$. Finally, let $w$ be the degree-2 vertex associated with the bend on $e$.  
	We apply \textsf{NoBend-Alg} to construct an orthogonal representation $\rect{H^*_\mu}$ of $\rect{G_\mu}$ having $w$, $x$, $z$, and $y$ as the four external designated corners. Since the external path between any two such consecutive corners is zero, and since no degree-2 vertex exists on the external face going from $w$ to $x$ counterclockwise, we have that $\rect{H^*_\mu}$ has an angle of $180^\circ$ at $v$ on the external face. Similarly, since no degree-2 vertex exists on the external face going from $w$ to $y$ clockwise, we have that $\rect{H^*_\mu}$ has an angle of $180^\circ$ at $u$ on the external face. This implies that the turn number of $p'$ in $\rect{H^*_\mu} \setminus \rect{e}$ is equal to three and the turn number of the other external path of $\rect{H^*_\mu} \setminus \rect{e}$ between $u$ and $v$ is equal to one. Hence, the inverse 
	$H^*_\mu \setminus e$ of $\rect{H^*_\mu} \setminus \rect{e}$ is $\L$-shaped and again the number of bends of $H^*_\mu \setminus e$ is not greater than the number of bends of $H^+_\mu \setminus e$.
	
	\begin{figure}[tb]
		\includegraphics[width=0.9\columnwidth]{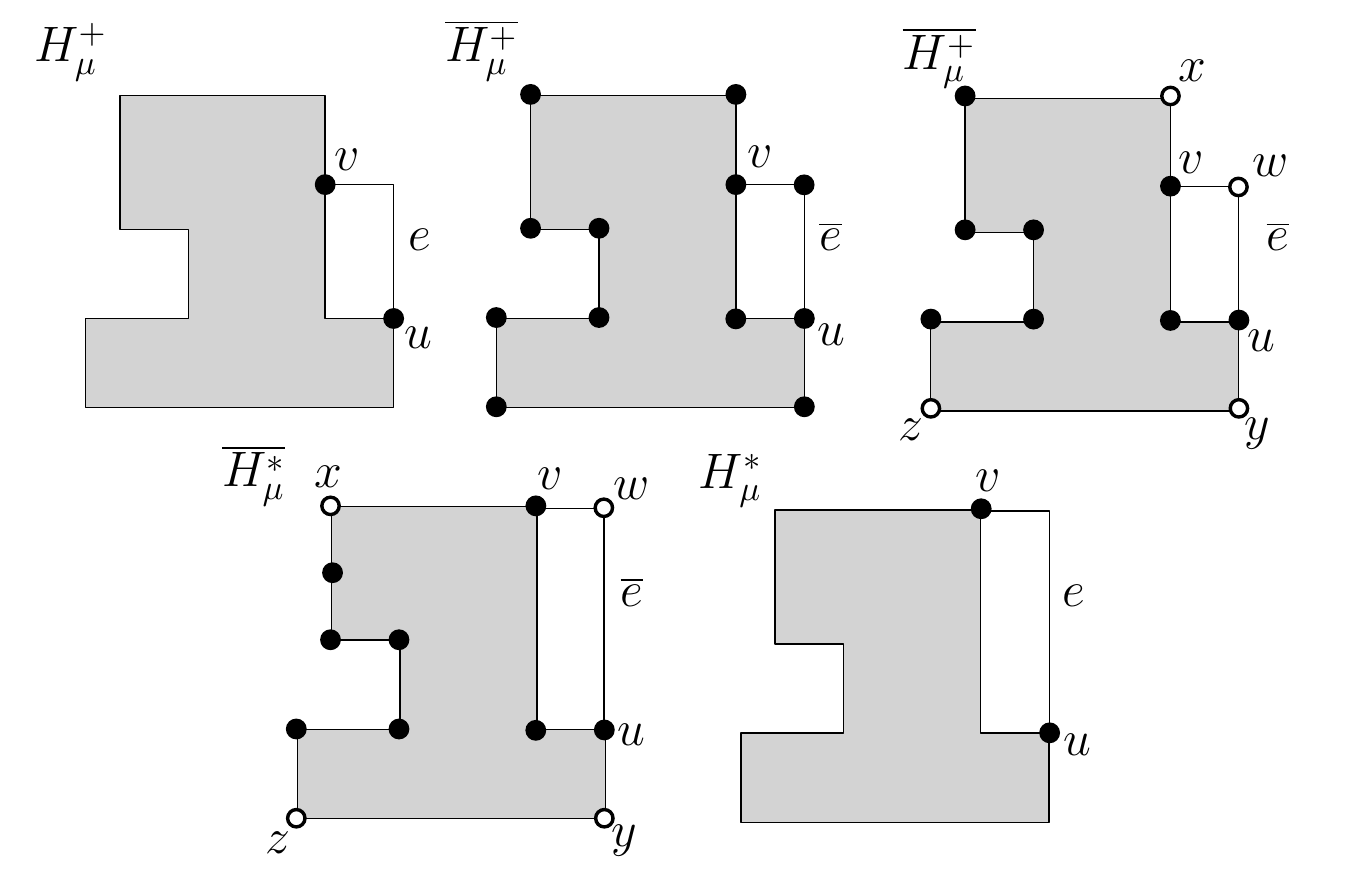}
		\caption{Illustration for the proof of Property (c).}\label{fi:property-c}
	\end{figure}
	
	\smallskip\paragraph{Property~(d).} Refer to \cref{fi:property-d} for a schematic illustration. As in the previous case, let $p'$ be the path of the external face of $\rect{G^+_\mu}$ between $u$ and $v$ not containing $\rect{e}$. Since by hypothesis $e$ has 0 bends in $H^+_\mu$, there exist at least four bend vertices $x, w, z, y$ along $p'$ in $\rect{H^+_\mu}$, corresponding to degree-2 vertices of $\rect{G^+_\mu}$. As in the proof of Property~(c), we choose $x$ as the first degree-2 vertex encountered along $p'$ while moving counterclockwise from $v$ and $y$ as the first degree-2 vertex along $p'$ in the clockwise direction from $u$. Vertices $w$ and $z$ are arbitrary degree-2 vertices along $p'$ between $x$ and $y$. We apply \textsf{NoBend-Alg} to construct an orthogonal representation $\rect{H^*_\mu}$ of $\rect{G_\mu}$ having $x$, $w$, $z$, and $y$ as the four external designated corners. For the same arguments as in the previous case, we have that $\rect{H^*_\mu}$ has an angle of $180^\circ$ at $u$ and at $v$ on the external face. This implies that the turn number of $p'$ in $\rect{H^*_\mu} \setminus \rect{e}$ is equal to four and the turn number of the other external path of $\rect{H^*_\mu} \setminus \rect{e}$ between $u$ and $v$ is equal to two. Hence, the inverse $H^*_\mu \setminus e$ of $\rect{H^*_\mu} \setminus \rect{e}$ is $\C$-shaped and again the number of bends of $H^*_\mu \setminus e$ is not greater than the number of bends of $H^+_\mu \setminus e$.
	
	\begin{figure}[tb]
		\includegraphics[width=0.9\columnwidth]{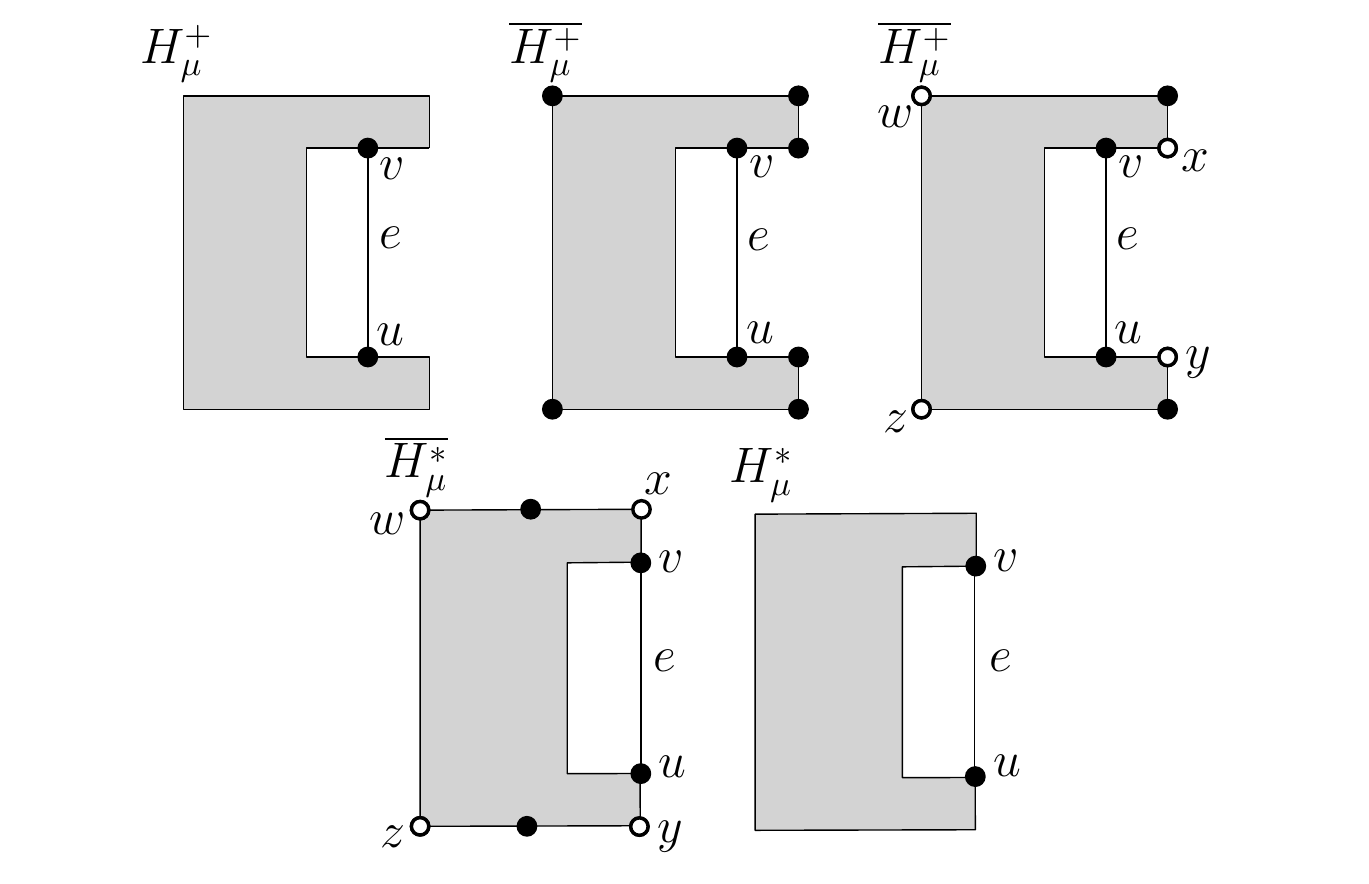}
		\caption{Illustration for the proof of Property (d).}\label{fi:property-d}
	\end{figure}

	\smallskip Concerning the time complexity, since \textsf{NoBend-Alg} runs in linear time we have that $H^*_\mu$ is computed from $H^+_\mu$ in linear time.    
\end{proof}

Let us now return to the proof of Properties~\textsf{O2-O3}. In~\cite{DBLP:conf/gd/DidimoLP18} it is proved that for any designated edge $e$ of $G$ on the external face, there exists an $e$-constrained bend-minimum orthogonal representation of $G$ with at most two bends per edge and verifying Properties~\textsf{O2-O3}. The proof of these two properties is done by starting from any $e$-constrained bend-minimum orthogonal representation $H'$ with two bends per edge, and then constructing a new orthogonal representation $H$ such that: $(1)$ $H$ has the same planar embedding as $H'$; $(2)$ each edge of $H$ has the same number of bends as in $H'$; $(3)$ $H$ verifies Properties~\textsf{O2-O3}. In particular, the construction of $H$ does not use the hypothesis that $H'$ has at most two bends per edge, and therefore it can be applied without changes by starting from an $e$-constrained bend-minimum orthogonal representation $H'$ with one bend per edge, i.e., one verifying Property~\textsf{O1}. Notice that such a property is still verified by $H$, due to Conditions $(1)$ and $(2)$. However, since in~\cite{DBLP:conf/gd/DidimoLP18} an edge could have up two bends, the root child component $\mu$ could have one among the $\D$-, $\C$-, $\X$-, and $\L$-shape if $\mu$ is a P- or an R-node, and spirality ranging from $2$ to $4$ if $\mu$ is an S-node. In our case an edge has 1 bend or 0 bends. Hence, due to Properties~(c) and~(d) of \cref{le:shapes}, if the root child component is a P- or an R-component then it always admits a bend-minimum representation that is either $\L$- or the $\C$-shaped. If the root child node component is an S-component then it has spirality 4 or 3; in fact, we can disregard higher values of spirality for a root child S-component because of the following lemma.

\begin{lemma}[\cite{DBLP:journals/siamcomp/BattistaLV98}]\label{le:elbow}
	Let $G_\mu$ be an S-component and let $H_\mu$ be a $0$-spiral bend-minimum orthogonal representation of $G_\mu$. Let $H'_\mu$ be another $k$-spiral orthogonal representation of $G_\mu$ such that $k > 0$ and $H'_\mu$ is bend-minimum among the $k$-spiral orthogonal representations of $G_\mu$. If $k=1$, $b(H'_\mu)=b(H_\mu)$, else $b(H'_\mu) \geq b(H_\mu)$.
\end{lemma}            

Concerning Property~\textsf{O4}, suppose that $H$ is a bend-minimum representation of $G$ with Properties~\textsf{O1-O3}. Let $H_\mu$ be any component of $H$, with a given representative shape, and suppose by contradiction that there exists a different representation $H'_\mu$ with the same shape as $H_\mu$ but such that $b(H'_\mu) < b(H_\mu)$. By Theorem~3 of~\cite{DBLP:conf/gd/DidimoLP18} it is possible to replace $H_\mu$ with $H'_\mu$ in $H$, still having a valid orthogonal representation $H'$ of $G$ and having the same set of bends along the edges of $H \setminus H_\mu$ (although it is intuitive, refer to~\cite{DBLP:conf/gd/DidimoLP18} for a formal definition of how to replace $H_\mu$ with $H'_\mu$). This implies that $b(H') < b(H)$, a contradiction.

\section{The Bend Counter Data Structure: Proof of Theorem~\ref{th:bend-counter}}\label{se:triconnected}

The \texttt{Bend-Counter} data structure is used by the labeling
algorithm to compute and update in $O(1)$ time the value
$b_{e}^\sigma(\mu)$ for each R-node $\mu$ of $T$ and for each
representative shape $\sigma$ for $\mu$. We recall that, by
Theorem~\ref{th:shapes} if $\mu$ is an inner R-component, it is either \X- or
\D-shaped. If $\mu$ is a child of the root $e$ of $T$, an optimal
orthogonal representation of $G$ is constructed by computing an optimal representation
of $G_\mu$ with $e$ on its external face. In both
cases, it is interesting to observe that $\skel(\mu)$ is a triconnected
cubic graph. Since Rahman et al.~\cite{DBLP:journals/jgaa/RahmanNN99}
show a linear-time algorithm that computes an optimal orthogonal representation
of a triconnected cubic planar graph in the fixed embedding setting, one
could be tempted to use that algorithm (with some minor adaptations) to
compute the values $b_{e}^{\d}(\mu)$ and $b_{e}^{\x}(\mu)$ for $\mu$.

There are however two main drawbacks in following this idea. Firstly,
the algorithm by Rahman et al.~\cite{DBLP:journals/jgaa/RahmanNN99}
minimizes the number of bends in the fixed embedding setting, that is
for a chosen external face. The planar embeddings of $\skel(\mu)$ depend
on the choice of the root of $T$, and since the labeling algorithm roots
the SPQR tree at each possible Q-node, if we had to re-execute $O(n)$ times the
algorithm by Rahman et al.~\cite{DBLP:journals/jgaa/RahmanNN99}
we would not have a linear-time labeling procedure. Secondly, when
performing a bottom-up visit of $T$ rooted at $e$, the value
$b_{e}^\sigma(\mu)$ for an R-node $\mu$ is computed by summing up the
bends along the real edges of $\skel(\mu)$ to the number of bends that
the pertinent graphs associated with its virtual edges carry. Hence,
when computing $b_{e}^\sigma(\mu)$ one has to optimize the bends along
the real edges and treat the virtual edges as edges that can be freely
bent. Even in the fixed embedding setting, the algorithm in~\cite{DBLP:journals/jgaa/RahmanNN99} is not designed for a graph
where one can bend some edges without impacting on the overall cost of
the optimization procedure.

The \texttt{Bend-Counter} makes it possible to overcome the above difficulties. It is
constructed in $O(n)$ time on a suitably chosen planar embedding of a
triconnected planar 3-graph $G$. Each edge $e$ of $G$ is given a
non-negative integer $\flex(e)$ specifying the \emph{flexibility} of
$e$.  An edge $e$ is called \emph{flexible} if $\flex(e) > 0$ and
\emph{inflexible} if $\flex(e) = 0$. Let $H$ be an orthogonal representation of $G$ and let $c(e)$ be the number of bends of $e$ in $H$ exceeding $\flex(e)$. The \emph{cost} $c(H)$ of $H$ is the sum of $c(e)$ for all edges $e$ of $G$. Note that, if all the edges of $G$ are inflexible, the cost of $H$ coincides with its total number of bends, i.e., $c(H)=b(H)$. For any face $f$ of $G=(V,E)$, the
\texttt{Bend-Counter} returns in $O(1)$ time the cost of a bend-minimum
orthogonal representation of $G$ with at most one bend per edge and
$f$ as external face. As a consequence, the use of the \texttt{Bend-Counter} makes it possible, for any possible choice of the root of $T$, to compute in $O(1)$ time the values  $b_{e}^\sigma(\mu)$, for any representative shape $\sigma$ of $\mu$ (see Section~\ref{se:labeling} for details). Since by \cref{le:spqr-tree-3-graph} every virtual edge of $\skel(\mu)$ corresponds to an S-component and since by Property~\textsf{O3} of Theorem~\ref{th:shapes} the spirality of each S-component is at most
four, from now on we assume $\flex(e) \leq 4$ for every edge
$e$ of~$G$.


As a preliminary step towards constructing the \texttt{Bend-Counter}, in \cref{sse:fixed-embedding} we study the fixed embedding scenario and extend to graphs with flexible edges some of the ideas of~\cite{DBLP:journals/jgaa/RahmanNN99}. \cref{sse:variable-embedding} focuses on the variable embedding setting and discusses how to construct the \texttt{Bend-Counter}.

\subsection{Triconnected Plane Graphs with Flexible Edges.}\label{sse:fixed-embedding}

Let $G=(V,E)$ be a triconnected cubic plane graph, and let $\flex(e) \leq 4$ denote the flexibility of any edge $e \in E$.
We aim at computing the minimum cost $c(G)$ of an orthogonal representation of $G$ that preserves the given embedding.
In other words, $c(G) = \min\{c(H): H \textrm{~is an embedding-preserving}$ $\textrm{~orthogonal representation of~} G\}$.
We start with some useful definitions and properties. Two cycles $C$ and $C'$ of $G$ are \emph{independent} if $G(C)$ and $G(C')$ have no common vertices, and they are \emph{dependent}
otherwise. If  $C'$ and $C$ are dependent and $C \subseteq G(C')$, $C$ is a \emph{descendant} of $C'$ and $C'$ is an \emph{ancestor} of $C$;  also, $C$ is a \emph{child-cycle} of $C'$ if it is not a descendant of another descendant of $C'$. \cref{fi:independent-cycles-example} depicts different cycles of the same plane graph $G$. In \cref{fi:independent-cycles-example-a}, $C_1$ and $C_2$ are two dependent cycles and $C_1 \subseteq G(C_2)$; hence $C_1$ is a descendant of $C_2$. In \cref{fi:independent-cycles-example-b} the cycles $C_1$ and $C_3$ are dependent but they have no inclusion relationship. Finally, the cycles $C_1$ and $C_4$ of \cref{fi:independent-cycles-example-c} are independent. The following lemma rephrases Lemma~1 of~\cite{DBLP:journals/jgaa/RahmanNN99}.

\begin{figure}[tb]
	\centering
	\subfloat[]{\label{fi:independent-cycles-example-a}\includegraphics[width=0.33\columnwidth]{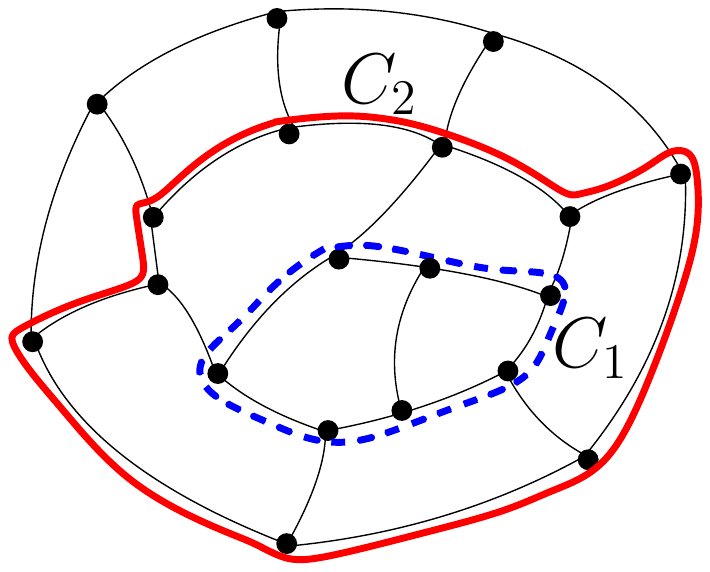}}
	\hfil
	\subfloat[]{\label{fi:independent-cycles-example-b}\includegraphics[width=0.33\columnwidth]{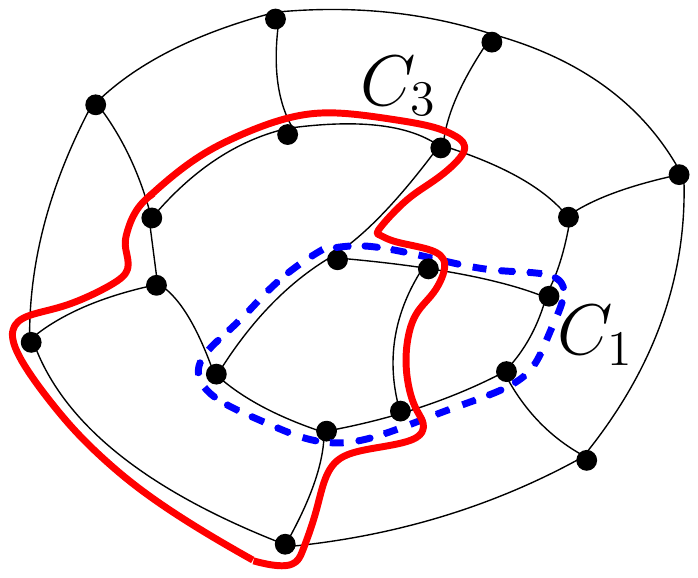}}
	\hfil
	\subfloat[]{\label{fi:independent-cycles-example-c}\includegraphics[width=0.33\columnwidth]{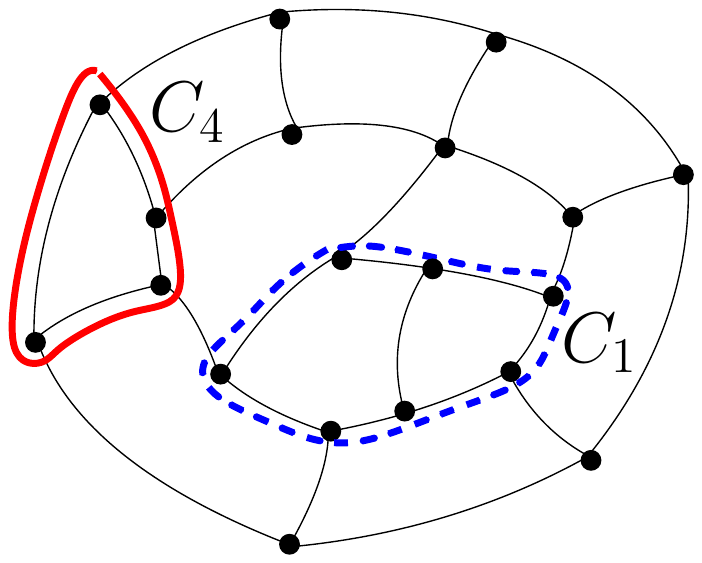}}
	\hfil
	\caption{Relationships of cycles in a plane graph $G_f$. (a) Cycle $C_1$ is a descendant of cycle $C_2$. (b) $C_1$ and $C_2$ are dependent cycles with no inclusion relationship. (c) $C_1$ and $C_2$ are independent cycles.}\label{fi:independent-cycles-example}
\end{figure}

\begin{lemma}[\cite{DBLP:journals/jgaa/RahmanNN99}]\label{le:independent-child-cycles}
	Any two 3-extrovert child-cycles of a 3-extrovert cycle $C$ are independent.
\end{lemma}

\cref{le:independent-child-cycles} implies that if $C$ is a 3-extrovert cycle of $G$, the inclusion relationships among all the 3-extrovert cycles in $G(C)$ (including $C$) can be described by a \emph{genealogical tree} denoted as $T_C$~\cite{DBLP:journals/jgaa/RahmanNN99}; the root of $T_C$ corresponds to $C$ and any other node of $T_C$ corresponds to a descendant cycle of $C$. The three leg vertices of $C$ split it into three edge-disjoint paths,
called the \emph{contour paths} of $C$. These contour paths are given a color in the set
$\{red, green, orange\}$ by the following recursive rule.


\medskip

\noindent\textsc{3-Extrovert Coloring Rule}: Let $T_C$ be the genealogical tree rooted at $C$. The three contour paths of $C$ are colored according to the following two cases.
	\begin{enumerate}
		\item $C$ has no contour path that contains either a flexible edge or a green contour path of a child-cycle of $C$ in $T_f$; in this case the three contour paths of $C$ are colored green.
		\item Otherwise, let $P$ be a contour path of $C$.
		\begin{itemize}
			\item[(a)] If $P$ contains a flexible edge then $P$ is colored orange.
			\item[(b)] If $P$ does not contain a flexible edge and it contains a green contour path of a child-cycle of $C$ in $T_f$, then $P$ is colored green.
			\item[(c)] Otherwise, $P$ is colored red.
		\end{itemize}
	\end{enumerate}
\smallskip

\cref{fi:rahman_colouration} shows a plane triconnected cubic graph and the genealogical trees $T_{C_{12}}$ and $T_{C_{11}}$, rooted at the 3-extrovert cycles $C_{12}$ and $C_{11}$. In this and in the following figures, an edge with $k$ \Tilde~symbols represents a flexible edge with $\flex(e)=k$.
The red-green-orange coloring of the contour paths of all 3-extrovert cycles of the graph are shown. The leaves of $T_{C_{12}}$ are $C_1$, $C_2$, and $C_3$; their contour paths are all green.  The internal node $C_4$ of $T_{C_{12}}$ has a contour path sharing an edge with a green contour path of its child-cycle $C_1$; therefore, this contour path of $C_4$ is green and the other two are red. Similarly, each of $C_5$ and $C_6$ has a green contour path and two red contour paths. The internal node $C_7$ has no contour path sharing edges with a green contour path of one of its child-cycles. The root $C_{12}$ has one contour path sharing an edge with a green contour path of its child-cycle $C_5$; hence, this contour path is green and the other two are red. The leaves of $T_{C_{11}}$ are $C_8$, $C_9$, and $C_{10}$. The contour paths of $C_8$ and $C_9$ are all green, while $C_{10}$ has a contour path with a flexible edge, hence, it has one orange contour path and two red contour paths. Finally, $C_{11}$ has a contour path with a flexible edge and two contour paths containing two green contour paths of its child-cycles $C_9$ and $C_8$, hence, it has one orange and two green contour paths.

\begin{figure}[tb]
	\centering
	\subfloat[]{\label{fi:rahman_colouration-a}\includegraphics[width=0.5\columnwidth]{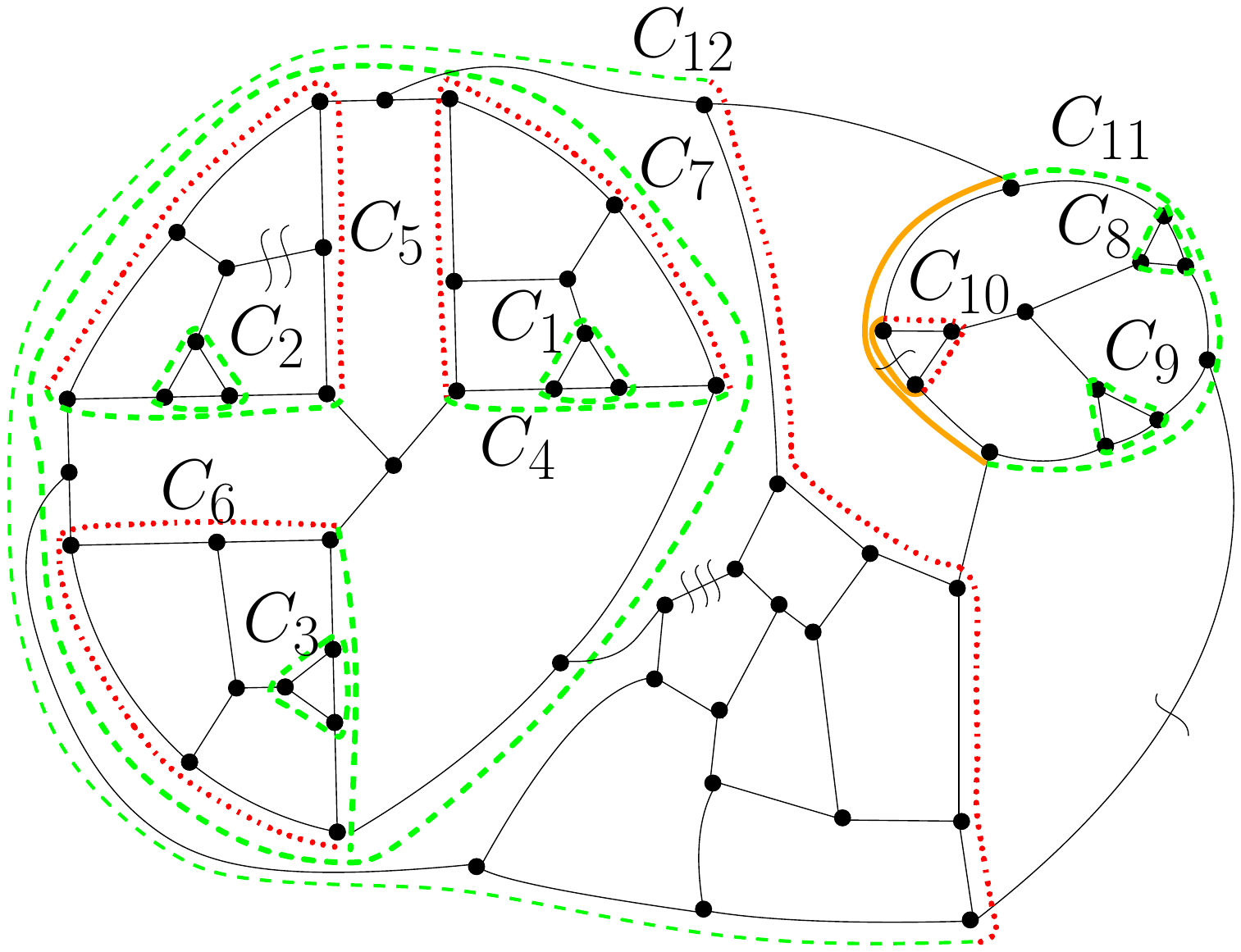}}
	\hfil
	\subfloat[]{\label{fi:rahman_colouration-b}\includegraphics[width=0.5\columnwidth]{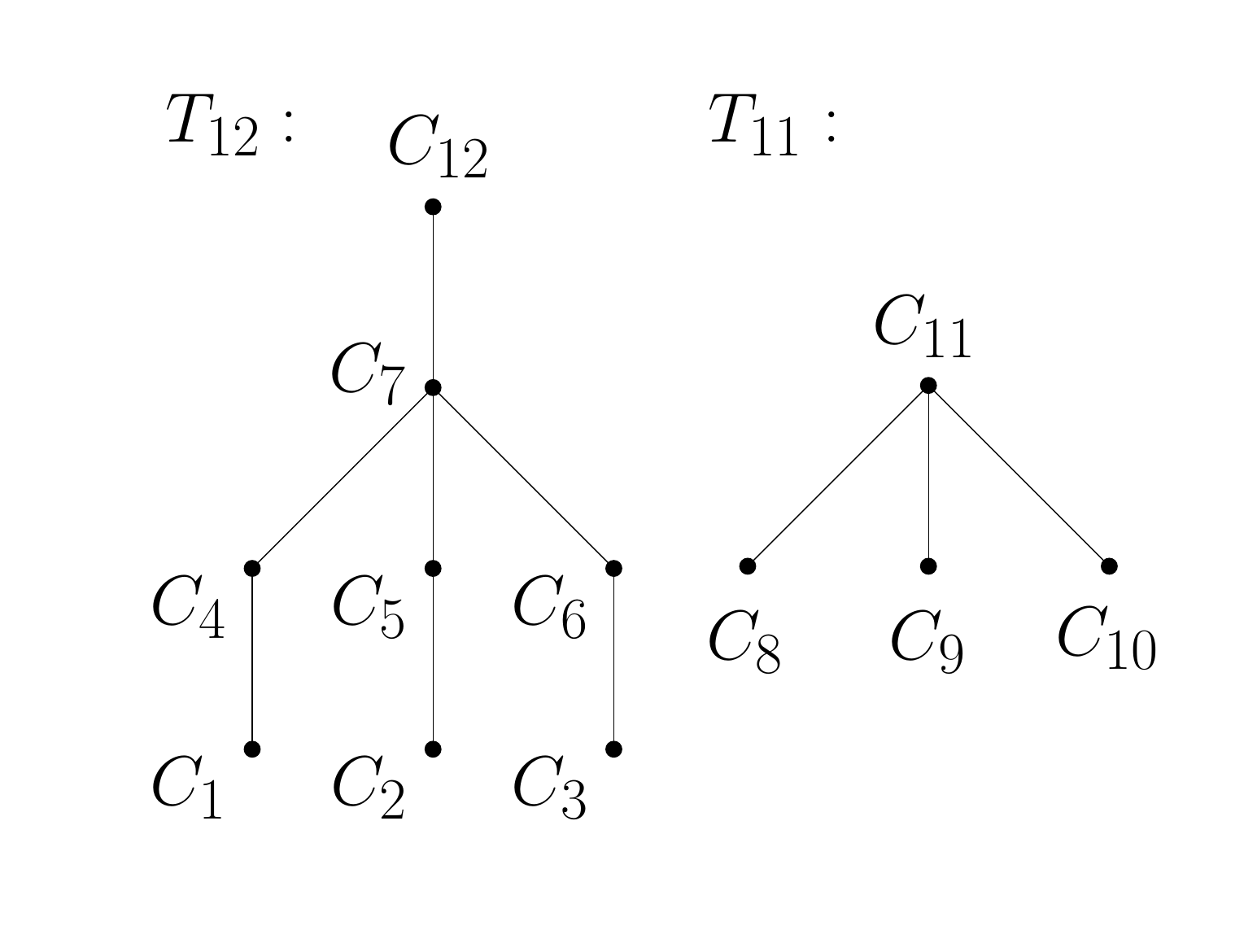}}
	\hfil
	\caption{(a) A plane triconnected cubic graph and the red-green-orange coloring of the contour paths of its 3-extrovert cycles. The red color is represented by a dotted line, the green color by a dashed line, and the orange color by a continuous line. (b) The genealogical trees $T_{C_{12}}$ and $T_{C_{11}}$.}\label{fi:rahman_colouration}
\end{figure}

A \emph{demanding 3-extrovert cycle} is a 3-extrovert cycle of $G$ that does not have an orange contour path and that does not share edges with a green contour path of its child-cycles. For example, in \cref{fi:rahman_colouration-a}, the leaves $C_1$, $C_2$, $C_3$, $C_8$ and $C_9$ of $T_{C_{12}}$ and $T_{C_{11}}$ correspond to demanding 3-extrovert cycles,
since they do not have child-cycles and they do not contain flexible edges. Also $C_7$ is a demanding cycle. All other 3-extrovert cycles are not demanding.

In \cref{sse:reference-embedding} we show how to color in linear time the contour paths of all 3-extrovert cycles of $G$ (\cref{le:3-extrovert-coloring-linear}) and how to compute in linear time the demanding 3-extrovert cycles of $G$ (\cref{le:demanding-3-extrovert-lineartime}).

\smallskip
We now give a theorem that extends a result of~\cite{DBLP:journals/jgaa/RahmanNN99}; it shows how to construct a bend-minimum orthogonal representation of $G$ with at most one bend per demanding cycle. To this aim, we define the \emph{flexibility} of $f$, denoted by $\flex(f)$. Intuitively, every orthogonal representation of $G$ needs four bends on the external face (because $G$ has no degree-2 vertices); $\flex(f)$ is a measure of how much the external face $f$ can take advantage of flexible edges to accommodate these bends without increasing the cost of the orthogonal representation.

Let $m_f$ be the number of flexible edges on $f$. Also, for an edge $e$ of $f$ with $\flex(e) \geq 1$, let $f'$ be the face incident to $e$ other than $f$; we denote by $x_e$ the sum of the flexibilities of the edges of $f'$ distinct from $e$ plus the number of demanding 3-extrovert cycles incident to $f'$ and not incident to $f$ (i.e., such that $f'$ is a leg face and $f$ is not a leg face).

\begin{definition}\label{de:flex-f}
The flexibility of $f$ is defined based on the following cases:

\begin{description}
	\item[1.] $m_f = 0$. In this case $\flex(f)=0$
	
	\item[2.] $m_f = 1$. Let $e_0$ be the unique flexible edge of $f$
	\begin{description}
		\item[2.a.] If $\flex(e_0) \leq 2$ then  $\flex(f)=\flex(e_0)$
		\item[2.b.] Else
		\begin{description}
			\item[2.b.i.] If $\flex(e_0) - x_{e_0} \leq 2$ then $\flex(f) = \flex(e_0)$
			\item[2.b.ii.] Else $\flex(f) = x_{e_0}+2$
		\end{description}
	\end{description}
	\item[3.] $m_f = 2$. Let $e_0$ and $e_1$ be the two flexible edges of $f$. We distinguish the following subcases:
		\begin{description}
			\item[3.a.] Edges $e_0$ and $e_1$ are not adjacent
			\begin{description}
				\item[3.a.i.] If $\flex(e_0) \geq 3$ and $\flex(e_1) = 1$
				\begin{description}
					\item[3.a.i.$\alpha$.] If $x_{e_0} \geq 1$ then $\flex(f)=4$
					\item[3.a.i.$\beta$.] Else $\flex(f)=3$
				\end{description}
				\item[3.a.ii.] Else, $\flex(f) = \min\{4,\flex(e_0)+\flex(e_1)\}$
			\end{description}
			\item[3.b.] Edges $e_0$ and $e_1$ share a vertex $v$. Let $e_2$ be the edge incident to $v$ different from $e_0$ and $e_1$
			\begin{description}
				\item[3.b.i.] If $\flex(e_0) + \flex(e_1) \leq 3$, then $\flex(f) = \flex(e_0)+\flex(e_1)$
				\item[3.b.ii.] If $\flex(e_0) \geq 2$ and $\flex(e_1) \geq 2$
				\begin{description}
					\item[3.b.ii.$\alpha$.] If $x_{e_0}+x_{e_1}-2\cdot\flex(e_2) \geq 1$, then $\flex(f) = 4$
					\item[3.b.ii.$\beta$.] Else, $\flex(f)=3$
				\end{description}
				\item[3.b.iii.] Else, assume $\flex(e_0) \geq 3$ and $\flex(e_1) = 1$
				\begin{description}
				\item[3.b.iii.$\alpha$.] If $x_{e_0}-\flex(e_2) \geq 1$, then $\flex(f)=4$
				\item[3.b.iii.$\beta$.] If we have that $x_{e_0}-\flex(e_2) = 0$, $x_{e_1} - \flex(e_2) \geq 1$, and $\flex(e_2) \geq 1$, then $\flex(f)=4$
				\item[3.b.iii.$\gamma$.] Else, $\flex(f)=3$
				\end{description}
			\end{description}
		\end{description}
	
	\item[4.] $m_f \geq 3$. $\flex(f) = \min\{4,\sum_{e \in C_f}\flex(e)\}$.
\end{description}

\end{definition}

Note that, by definition, $0 \leq \flex(f) \leq 4$. Denote by $D(G)$ the set of demanding 3-extrovert cycles of $G$ that do not share edges with other demanding 3-extrovert cycles of $G$ and denote by $D_{f}(G) \subseteq D(G)$
those cycles of $D(G)$ that contain an edge of the external face $f$. We prove the following.

\begin{restatable}{theorem}{thFixedEmbeddingMinBend}\label{th:fixed-embedding-min-bend}
	Let $G=(V,E)$ be an $n$-vertex plane triconnected cubic graph with given edge flexibilities and let $f$ be the external face of $G$. We have that:
	\begin{equation}\label{eq:fixed-embedding-cost}
       \;\;\;\;\; c(G) = |D(G)| + 4 - \min\{4, |D_{f}(G)| + \flex(f) \}
	\end{equation}
	Also, if the external face is not a 3-cycle with all edges of flexibility at most one, there exists a cost-minimum embedding-preserving orthogonal representation $H$ of $G$ such that each inflexible edge has at most one bend and each flexible edge $e$ has at most $flex(e)$ bends.
    Furthermore, if one is given $D(G)$ and $D_{f}(G)$ such representation~$H$ can be computed in $O(n)$ time.
	
	
\end{restatable}

\begin{proof}
	
	By Theorem~\ref{th:RN03} (where bends act like vertices of degree two), any orthogonal representation of $G$ needs four bends on the external face (Condition~$(i)$ of Theorem~\ref{th:RN03}) and a bend along the boundary of every 3-extrovert cycle (Condition~$(iii)$ of Theorem~\ref{th:RN03}). By the same theorem, this set of bends is also sufficient for an orthogonal representation of $G$.
	
	Let $H$ be a cost-minimum orthogonal representation of $G$. Let $\rect{H}$ be the rectilinear image of $H$ and  $\rect{G}$ its underlying graph.
	We can assume $\rect{G}$ to be such that smoothing any degree-2 vertex (which corresponds to a bend in $H$ on either a flexible or an inflexible edge) yields a graph $\rect{G'}$ that does not admit a rectilinear representation. In fact, if this is not the case, $\rect{G'}$ still satisfies the conditions of Theorem~\ref{th:RN03} and, therefore, there exists an orthogonal representation $H'$ with one bend less than $H$ and such that $c(H') \leq c(H)$. By iterating this simplification, we obtain a graph $\rect{G}$ where any vertex of degree two is necessary to satisfy the conditions of Theorem~\ref{th:RN03}.
	Therefore, it is sufficient to describe how to construct an orthogonal representation $H$ of $G$ such that any bend of $H$ is needed to satisfy Conditions~$(i)$ and~$(iii)$ of Theorem~\ref{th:RN03}. Actually, thanks to Theorem~\ref{th:RN03}, it is sufficient to specify the edges of $G$ that host some bends in $H$.
	
	First, let's focus on the set of the 3-extrovert cycles of $G$. This set can be partitioned into the subsets of demanding and non-demanding 3-extrovert cycles. By the \textsc{3-Extrovert Coloring Rule} a non-demanding 3-extrovert cycle $C$ either contains a flexible edge or it contains an edge $e$ shared with a demanding 3-extrovert cycle. Hence, Condition~$(iii)$ of Theorem~\ref{th:RN03} can be satisfied for $C$ without additional cost by inserting a single bend on one of its flexible edges, in the former case, or on edge $e$, in the latter case.
	

	By \cref{le:independent-child-cycles} if two demanding 3-extrovert cycles share some edges, they must include edges of the external face.
	Also, by using arguments in~\cite{DBLP:journals/jgaa/RahmanNN99} (variants of these arguments can be found also in the proof of \cref{le:fagiolo-nero-extrovert}) it is possible to prove that any edge of the external face belongs to $k-1$ demanding 3-extrovert cycles that are not independent, where $k$ is the number of non-independent demanding 3-extrovert cycles.
	It follows that the four bends of the external face used to satisfy Condition~$(i)$ can be used to also satisfy Condition~$(iii)$ for all such intersecting demanding 3-extrovert cycles. Namely, insert a bend on any edge on the external face; this bend satisfies Condition~$(iii)$ for $k-1$ of the demanding 3-extrovert cycles that share edges; then choose an edge of the remaining unsatisfied demanding 3-extrovert cycle and insert a bend along this edge; the remaining two bends can be inserted along any two arbitrary edges of the external face.
	It may be worth remarking that if there are two demanding 3-extrovert cycles that are not independent, then every demanding 3-extrovert cycle having edges on the external face is not independent; also, in this case the external face cannot have any flexible edges.

	
	From the discussion above it follows that in order to satisfy Condition~$(iii)$ of Theorem~\ref{th:RN03}: We can disregard the 3-extrovert cycles that are not independent; $H$ needs one bend for each independent demanding 3-extrovert cycle; and no edge needs more than one bend. Thus, Condition~$(iii)$ can be satisfied in such a way that the number of bends on each flexible edge does not exceed its flexibility.
	The cost of this set of bends corresponds to the term $|D(G)|$ in \cref{eq:fixed-embedding-cost}.
	
	The second part of \cref{eq:fixed-embedding-cost} derives from the need of satisfying Condition~$(i)$ of Theorem~\ref{th:RN03}. If~$f$ contains neither flexible edges nor edges of demanding 3-extrovert cycles, then $H$ must have exactly four costly bends on (inflexible) edges of~$f$. In particular, if~$f$ has at least four edges we insert at most one bend per edge; if $f$ is a 3-cycle we insert two bends on exactly one edge and the remaining two bends are on distinct edges.
	If~$f$ contains edges of independent demanding 3-extrovert cycles or flexible edges, we can save up to four units of cost for satisfying Condition~$(i)$ of Theorem~\ref{th:RN03} by exploiting $D_f(G)$ and $\flex(f)$ as described below.
	
	Since the set $D_f(G)$ contains independent demanding 3-extrovert cycles that share an edge with~$f$, a bend inserted on such edges to satisfy Condition~$(iii)$ of Theorem~\ref{th:RN03} also contributes to satisfy Condition~$(i)$, thus saving $\min\{4,|D_f(G)|\}$ units of cost in total. Hence, we can assume that a bend placed on a cycle $C \in D_f(G)$ is always located along~$f$. For the remaining $4 - \min\{4,|D_f(G)|\}$ bends (if any) we might take advantage of the flexible edges.
	To this aim, we have to consider different cases that reflect those used in the definition of $\flex(f)$. Before analyzing these cases, we make the following observation. Adding one or more bends along an edge $e=(u,v)$ of the external face (i.e., subdividing $e$) always creates one 2-extrovert cycle $C_1$ and two 3-extrovert cycles $C_2$ and $C_3$ in $\rect{G}$, in addition to those of $G$. Namely, denoted by $W$ the set of vertices of $\rect{G}$ that correspond to the bends along $e$, we have $C_1= C_o(\rect{G} \setminus W)$, $C_2 = C_o(\rect{G} \setminus W \setminus \{u\})$, and $C_3 = C_o(\rect{G} \setminus W \setminus \{v\})$. See \cref{flexible_C123}, where we indicate the placement of a bend into an edge with a cross (we use this notation in all the figures of this proof). We may need to add extra degree-2 vertices to $\rect{G}$ (i.e., bends on the edges of $G$) along these cycles to preserve Conditions~$(ii)-(iii)$ of Theorem~\ref{th:RN03}. 	Note that $C_2$ (resp. $C_3$) contains all edges of the external face $f$ of $G$ except the two (adjacent) edges of $f$ incident to $u$ (resp. $v$) and that  $C_1$ contains all edges of $C_o(G)$ except $e=(u,v)$. Consider the four bends inserted along $C_o(G)$ to satisfy Conditions~$(i)$ of Theorem~\ref{th:RN03}. By the discussion above, if these four bends are distributed on at least two non-adjacent edges and no edge receives more than two bends, then Conditions~$(ii)-(iii)$ of Theorem~\ref{th:RN03} are always satisfied for $C_1$, $C_2$, and $C_3$. Therefore, in the case analysis that follows we will consider the cycles $C_1$, $C_2$, and $C_3$ only when the four bends along $C_o(G)$ are distributed along at most two edges and these edges are adjacent or when an edge is given three (or more) bends. 
	We adopt the same notation $m_f$ and $x_e$ introduced for defining $\flex(f)$.

	\begin{figure}[tbh]
		\centering
		{\label{fi:flexible_C1}\includegraphics[width=0.32\columnwidth]{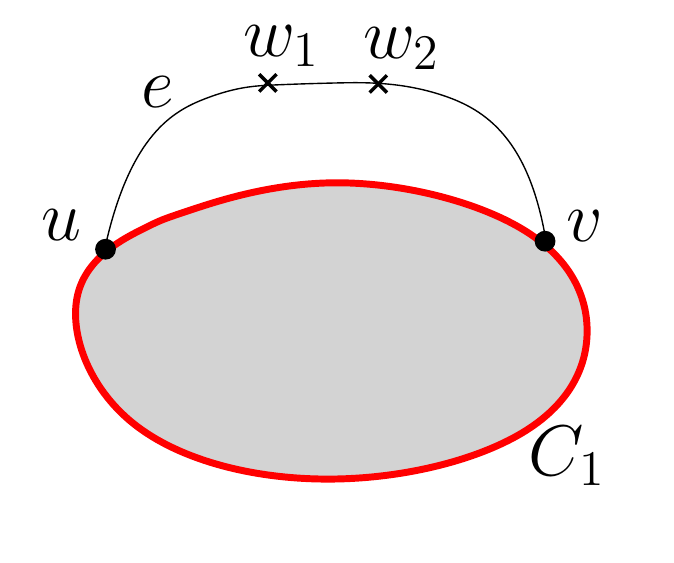}}
		{\label{fi:flexible_C2}\includegraphics[width=0.32\columnwidth]{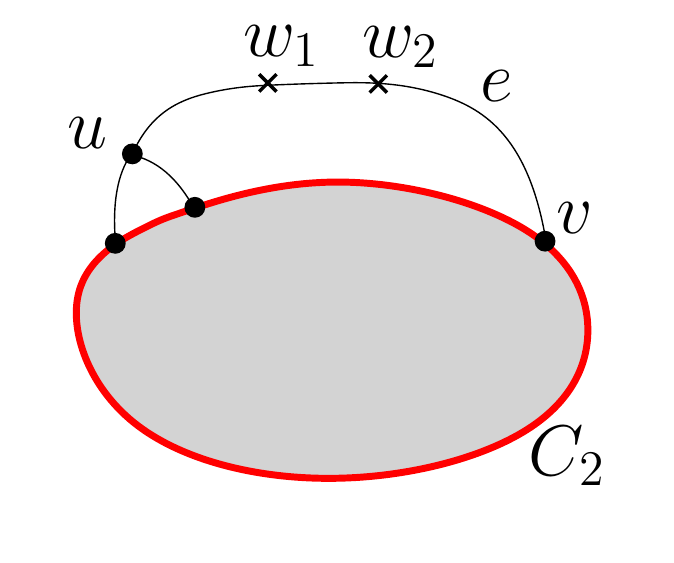}}
		{\label{fi:flexible_C3}\includegraphics[width=0.32\columnwidth]{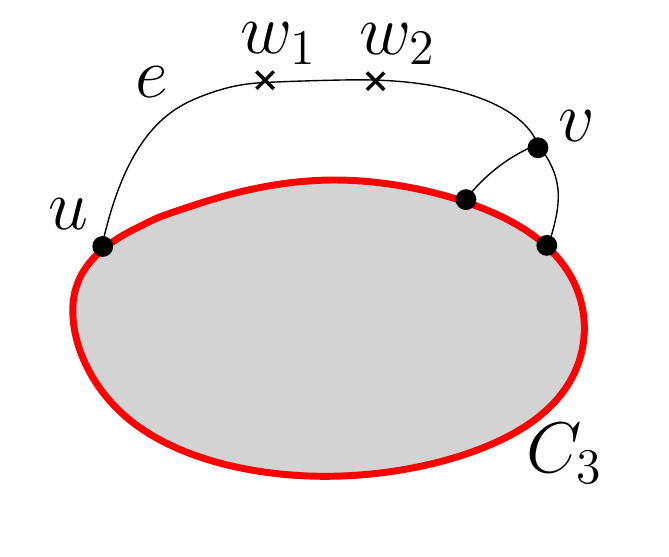}}
		\caption{
		Placing one or more bends along an edge always creates a new 2-extrovert cycle $C_1$ and two new 3-extrovert cycles $C_2$ and $C_3$.	
		}\label{flexible_C123}
	\end{figure}

%

	\smallskip\paragraph{Case $m_f = 0$.} This case is ruled out since the four bends along $C_o(G)$ are distributed along more than two edges and no edge is given three bends.
	
	\smallskip\paragraph{Case $m_f = 1$.} Refer to Case~2 of \cref{de:flex-f}. Let $e_0$ be the unique flexible edge of $f$.
	\begin{itemize}
			
			\item First, suppose that $\flex(e_0) \leq 2$ (see also \cref{flexible_mf=1_abc}). If the external face is not a 3-cycle, we instert $\flex(e_0)$ bends along $e_0$ in $H$ while the remaining bends on $C_o(G)$ are distributed on the (inflexible) edges of $C_o(G) \setminus e_0$ such that every edge receives at most one bend. If the external face is a 3-cycle, two bend are inserted along the flexible edge (possibly adding an extra cost along $e_0$) so to guarantee that the two inflexible edges of $C_o(G)$ are bent once. Since $\flex(f)=\flex(e_0)$ (Case~2.a. of \cref{de:flex-f}), \cref{eq:fixed-embedding-cost} follows.

			\begin{figure}[htb]
				\centering
				{\label{fi:flexible_mf=1_a}\includegraphics[width=0.32\columnwidth]{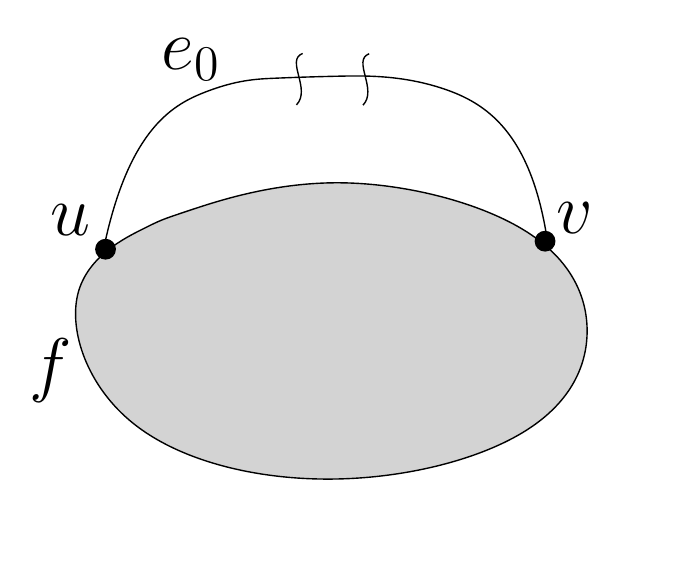}}
				\hfil
				{\label{fi:flexible_mf=1_b}\includegraphics[width=0.32\columnwidth]{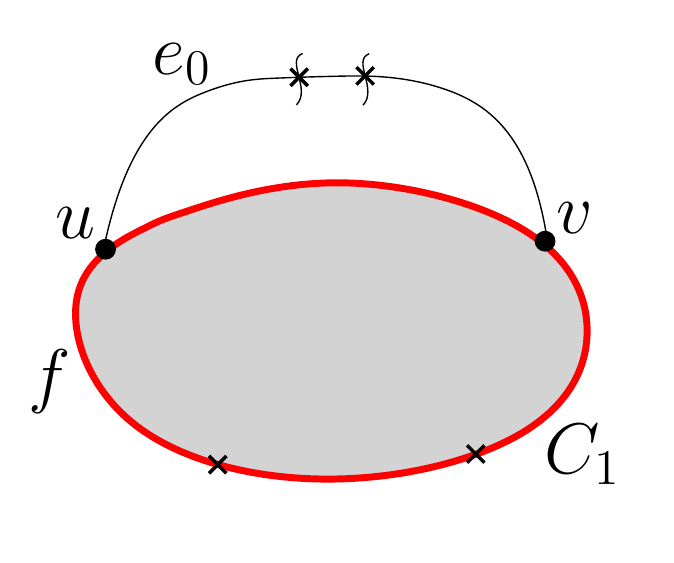}}
				\hfil
				{\label{fi:flexible_mf=1_c}\includegraphics[width=0.32\columnwidth]{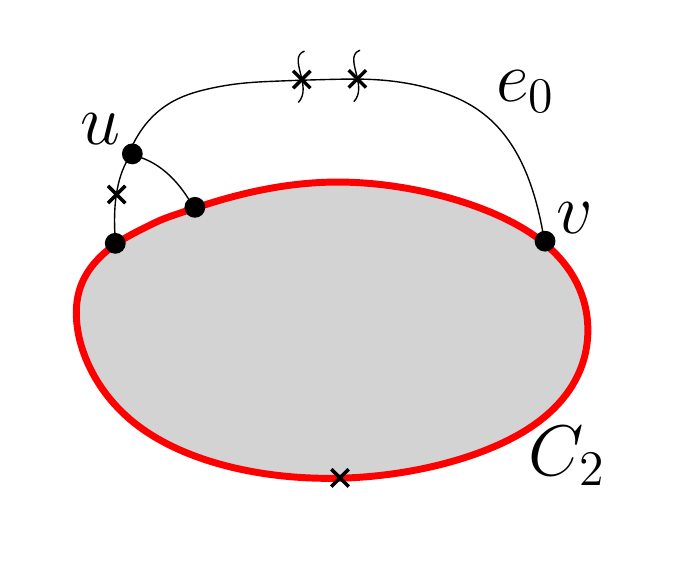}}
				\hfil
				\caption{Illustration for the proof of Theorem~\ref{th:fixed-embedding-min-bend}, when $m_f=1$ and $flex(e_0)\le 2$.
				}\label{flexible_mf=1_abc}
			\end{figure}
			
			\item Otherwise, $\flex(e_0) \geq 3$. Assume first that $\flex(e_0) = 3$. Refer to \cref{fi:flexible_mf=1_def}. In this case we could place three bends along $e_0=(u,v)$ without extra cost only if Conditions~$(ii)$--$(iii)$ of Theorem~\ref{th:RN03} are satisfied for the three cycles $C_1$, $C_2$, and $C_3$ created by the insertion of the bends along $e_0$.
Regarding Condition~$(ii)$ for $C_1$, consider the internal face $f'$ incident to $e_0$. Cycle $C_1$ consists of the path $C_o(G) \setminus e$ and the path $\Pi$ from $u$ to $v$ on $f'$ not passing through $e_0$. There are two subcases:
			
			\begin{figure}[htb]
				\centering
				\subfloat[]{\label{fi:flexible_mf=1_d}\includegraphics[width=0.32\columnwidth]{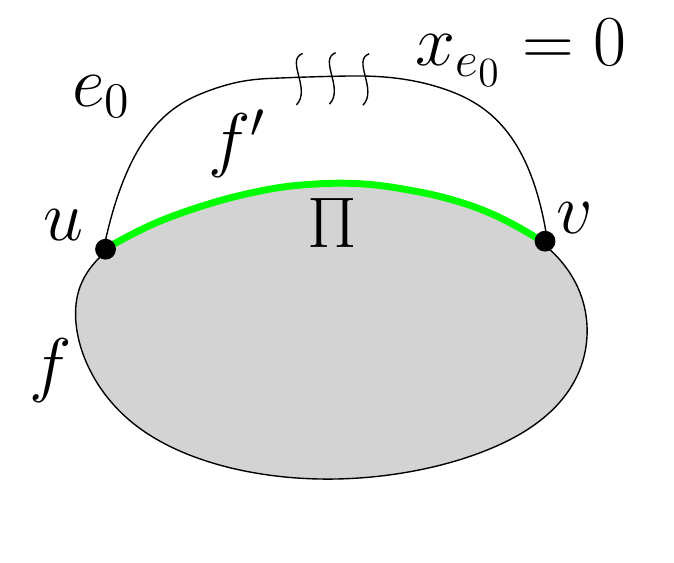}}
				\hfil
				\subfloat[]{\label{fi:flexible_mf=1_e}\includegraphics[width=0.32\columnwidth]{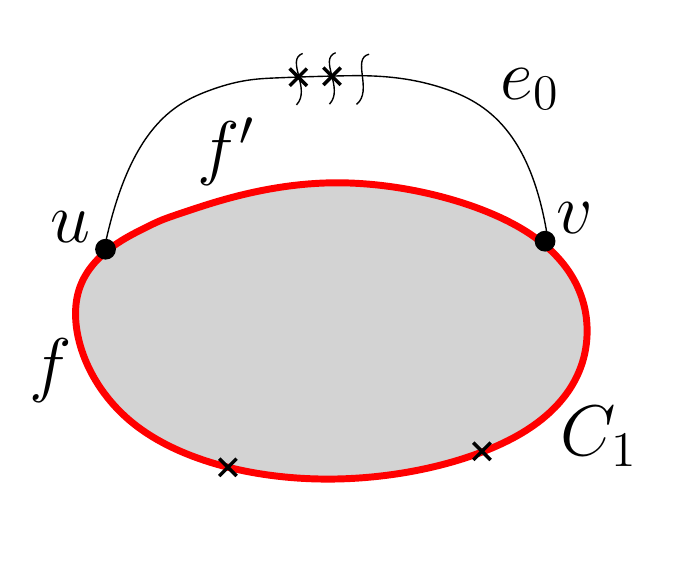}}
				\hfil
				\subfloat[]{\label{fi:flexible_mf=1_f}\includegraphics[width=0.32\columnwidth]{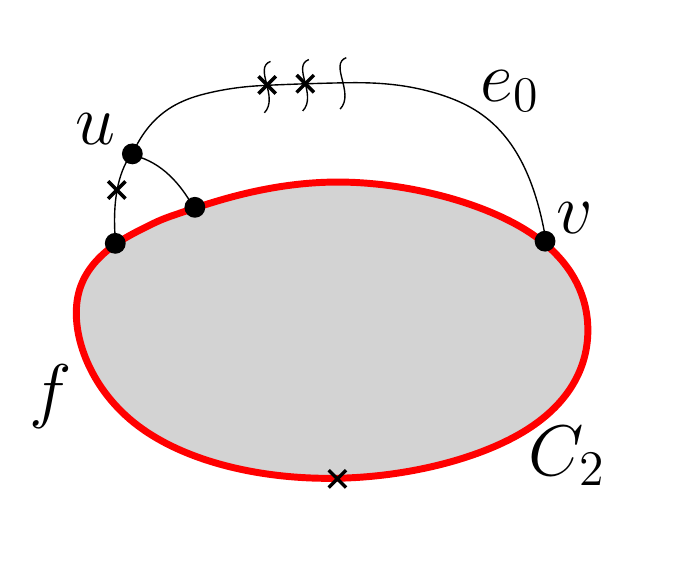}}
				\hfil
				\caption{Illustration for the proof of Theorem~\ref{th:fixed-embedding-min-bend}, when $m_f=1$, $flex(e_0)\ge 3$, and $x_{e_0}=0$.
				}\label{fi:flexible_mf=1_def}
			\end{figure}
		
			If $x_{e_0} \geq 1$ then, by definition, $\Pi$ contains an edge of a demanding cycle or a flexible edge. In the former case, $\Pi$ contains the contour path $P$ of a demanding 3-extrovert cycle $C$ (see \cref{fi:flexible_mf=1_g}). Since~$C$ must have a bend in $H$, we place this bend along~$P$. Indeed, every 3-extrovert cycle $C'$ that is an ancestor of $C$ and that shares edges with $C$ either contains $P$ or it contains $e_0$ (see \cref{flexible_mf=1__PeitherC}). In both cases Condition~$(ii)$ of Theorem~\ref{th:RN03} is satisfied for $C'$ if we place one bend along $P$ and three bends along $e_0$.
			Also, the bend on $P$ and a bend inserted in an edge of $C_o(G) \setminus e_0$ guarantee Condition~$(ii)$ for $C_1$ (see \cref{fi:flexible_mf=1_h}) and the bend on $P$ satisfies Condition~$(iii)$ for $C_2$ and $C_3$ (see \cref{fi:flexible_mf=1_i}).
			If $\Pi$ does not contain an edge of a demanding cycle but it contains a flexible edge, we can place a (free) bend along $\Pi$ and one bend along an edge of $C_o(G) \setminus e_0$ to satisfy Conditions~$(ii)$--$(iii)$ for $C_1$, $C_2$, and $C_3$. Therefore, since $x_{e_0} \geq 1$ implies $\flex(e_0) - x_{e_0} \leq 2$ and since $\flex(f) = 3$ (Case~2.b.i of \cref{de:flex-f}), \cref{eq:fixed-embedding-cost} holds in this case.
			
			\begin{figure}[htb]
				\centering
				\subfloat[]{\label{fi:flexible_mf=1_g}\includegraphics[width=0.32\columnwidth]{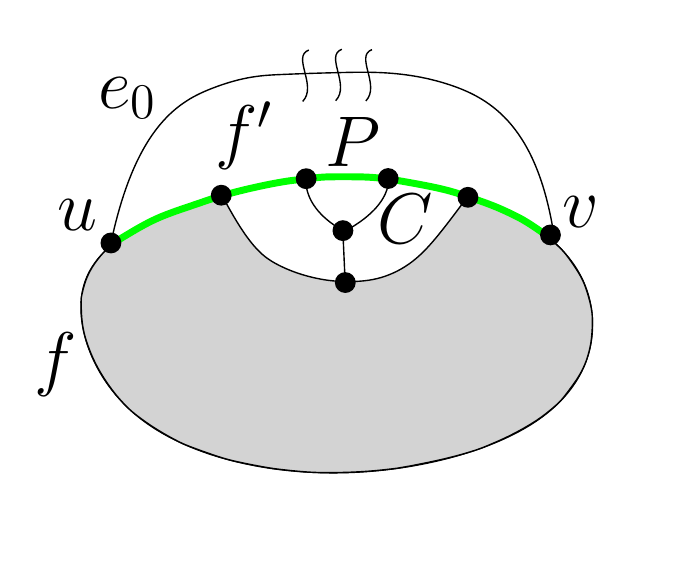}}
				\hfil
				\subfloat[]{\label{fi:flexible_mf=1_h}\includegraphics[width=0.32\columnwidth]{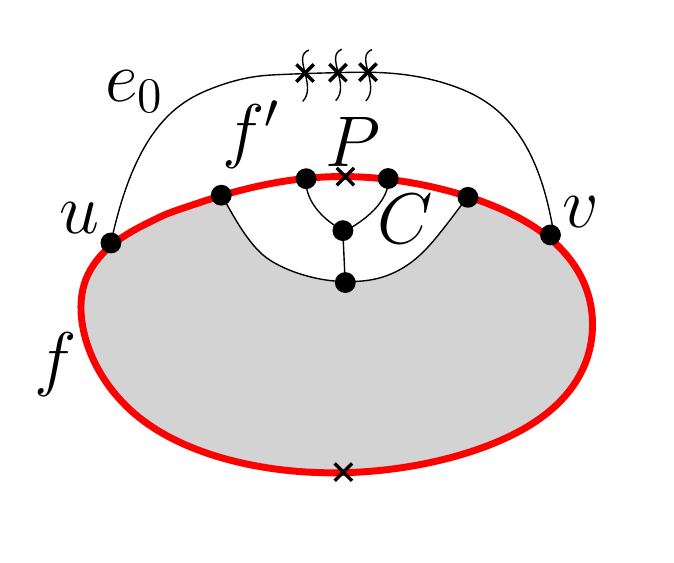}}
				\hfil
				\subfloat[]{\label{fi:flexible_mf=1_i}\includegraphics[width=0.32\columnwidth]{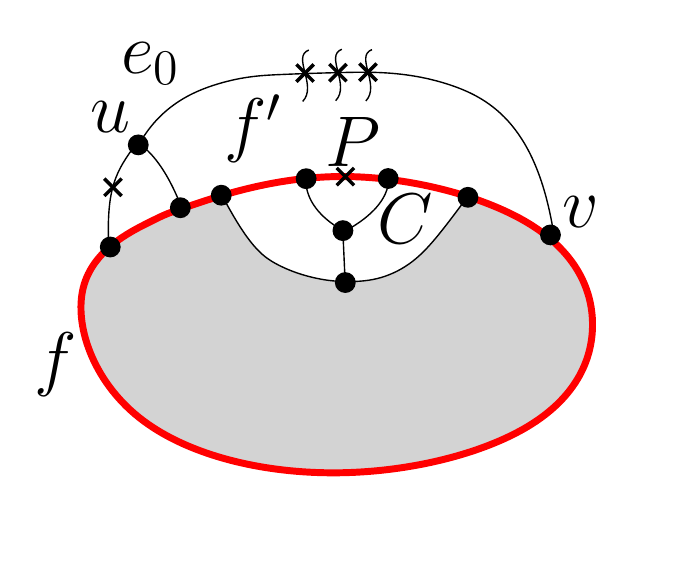}}
				\hfil
				\caption{Illustration for the proof of Theorem~\ref{th:fixed-embedding-min-bend} when $m_f=1$, $flex(e_0)\ge 3$, and $x_{e_0}=1$.
			}\label{fi:flexible_mf=1_ghi}
			\end{figure}
			
			\begin{figure}[htb]
				\centering
				\subfloat[]{\label{fi:flexible_mf=1_PeitherC_a}\includegraphics[width=0.33\columnwidth]{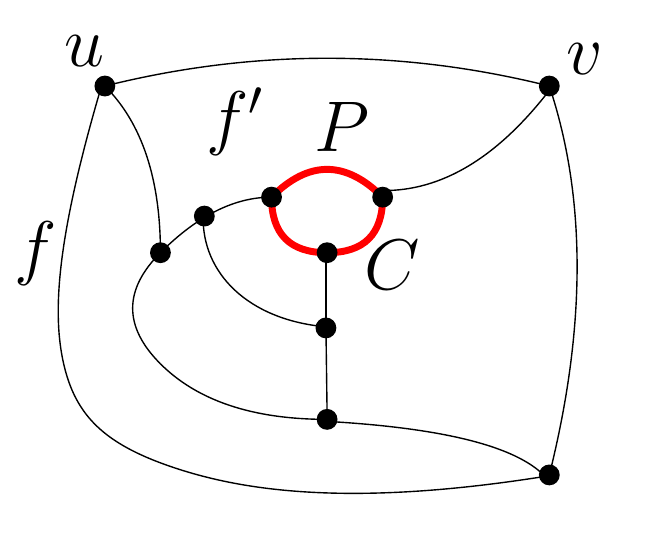}}
				\hfil
				\subfloat[]{\label{fi:flexible_mf=1_PeitherC_b}\includegraphics[width=0.33\columnwidth]{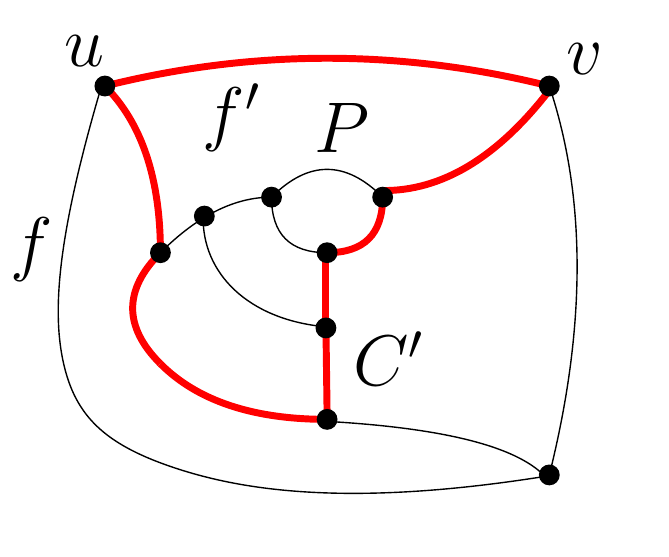}}
				\hfil
				\subfloat[]{\label{fi:flexible_mf=1_PeitherC_c}\includegraphics[width=0.33\columnwidth]{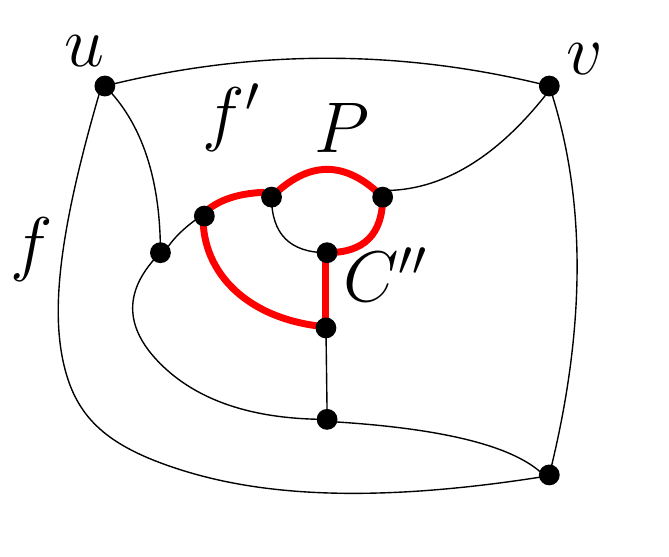}}
				\hfil
				\caption{(a)~A 3-extrovert demanding cycle $C$ sharing a contour path $P$ with face $f'$. (b)~An ancestor $C'$ of $C$ that does not contain $P$ and that contains $e=(u,v)$. (c)~An ancestor $C''$ of $C$ that does not contain $e$ and that contains $e=(u,v)$.
			}\label{flexible_mf=1__PeitherC}
			\end{figure}

			If $x_{e_0}=0$ (see \cref{fi:flexible_mf=1_d}), we place two bends on~$e_0$ and the remaining two bends appear on the (inflexible) edges of $C_o(G) \setminus e_0$, which satisfies Condition~$(ii)$ for $C_1$
			(see \cref{fi:flexible_mf=1_e}).
			Since the two bends on the edges of $C_o(G) \setminus e_0$ are placed on two distinct edges, they satisfy also Condition~$(iii)$ for $C_2$ and $C_3$ (see \cref{fi:flexible_mf=1_f}).
			By definition, when $x_{e_0}=0$, $\flex(f)=2$ (Case~2.b.ii of \cref{de:flex-f}) and \cref{eq:fixed-embedding-cost} holds.

			The case of $\flex(e) = 4$ is handled with the same argument, distinguishing between the subcases $x_e = 0$, $x_e = 1$, or $x_e \geq 2$. Namely, we place 2, 3, or 4 bends along $e$ depending on whether $x_e = 0$, $x_e = 1$, or $x_e \geq 2$, respectively.
		\end{itemize}

	\smallskip\paragraph{Case $m_f = 2$.}  Refer to Case~3 of \cref{de:flex-f}. Let $e_0$ and $e_1$ be the two flexible edges of $f$. We distinguish the following subcases:
		\begin{itemize}
			\item{$e_0$ and $e_1$ not ajacent}. Refer to Case~3.a of \cref{de:flex-f}. Suppose, without loss of generality, that $\flex(e_0) \geq \flex(e_1)$. We have two cases.
			\begin{itemize}
				
				\item[$(i)$] $\flex(e_0) \geq 3$ and $\flex(e_1) = 1$.
				Let $f' \neq f$ be the other face incident to $e_0=(u,v)$ and let $\Pi$ be the path from $u$ to $v$ incident to $f'$ and not containing $e_0$.
				Condition~$(iii)$ for the four 3-extrovert cycles created when inserting bends in both $e_0$ and $e_1$ is always satisfied by the bends themselves (see \cref{fi:flexible_mf=2_notadjacent_a}).
				Condition~$(ii)$ for the 2-extrovert cycle created when inserting a bend in $e_1$ is satisfied if we insert at least two bends in $e_0$ (see \cref{fi:flexible_mf=2_notadjacent_b}).
				Regarding Condition~$(ii)$ for the 2-extrovert cycle created by inserting bends in $e_0$, we have two subcases depending on the value of~$x_{e_0}$.
				If $x_{e_0} \geq 1$ (see \cref{fi:flexible_mf=2_notadjacent_c}, where $\Pi$ contains a flexible edge), then such a condition is satisfied if we insert a bend in $e_1$ and a bend in $\Pi$ 
				(see \cref{fi:flexible_mf=2_notadjacent_d}). Otherwise, if $x_{e_0} = 0$, in order to satisfy Condition~$(ii)$ for the 2-extrovert cycle created by inserting bends in $e_0$, we have to insert a costly bend in $C_o(G) \setminus e_0$ (see \cref{fi:flexible_mf=2_notadjacent_e}).
				Hence, if $x_{e_0} \geq 1$ we insert three bends along $e_0$, one bend along $e_1$, and one bend along $\Pi$. Since in this case $\flex(f)=4$ (Case~3.a.i.$\alpha$ of \cref{de:flex-f}), \cref{eq:fixed-embedding-cost} holds.
				Otherwise, if $x_{e_0} = 0$, we insert two bends along $e_0$, one bend along $e_1$, and one bend along $C_o(G) \setminus e_0$. Again, since in this case $\flex(f)=3$  (Case~3.a.i.$\beta$ of \cref{de:flex-f}), we have that \cref{eq:fixed-embedding-cost} holds.
				
				\item[$(ii)$] $\flex(e_0) < 3$ or $\flex(e_1) > 1$. In this case we place $\flex(e_0)$ bends along $e_0$ and (up to) $\flex(e_1)$ bends along $e_1$.
				The remaining $4-\flex(e_0)-\flex(e_1)$ bends, if any, are placed along $C_o(G) \setminus \{e_0,e_1\}$. It can be checked that Conditions~$(ii)$ and $(iii)$ hold for the two 2-extrovert cycles and the four 3-extrovert cycles created by the bends along $e_0$ and $e_1$ (see \cref{fi:flexible_mf=2_notadjacent_f}). Since in this case, by definition, we have $\flex(f) = \min\{4,\flex(e_0)+\flex(e_1)\}$ (Case~3.a.ii of \cref{de:flex-f})), \cref{eq:fixed-embedding-cost} holds.
				
			\end{itemize}
			\begin{figure}[htb]
				\centering
				\subfloat[]{\label{fi:flexible_mf=2_notadjacent_a}\includegraphics[width=0.32\columnwidth]{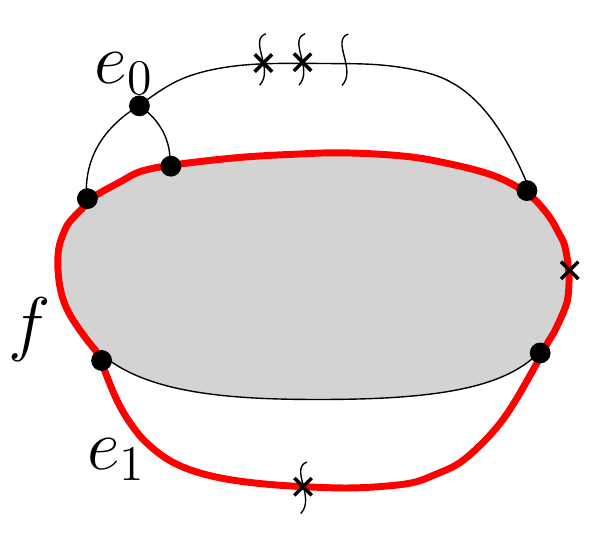}}
				\subfloat[]{\label{fi:flexible_mf=2_notadjacent_b}\includegraphics[width=0.32\columnwidth]{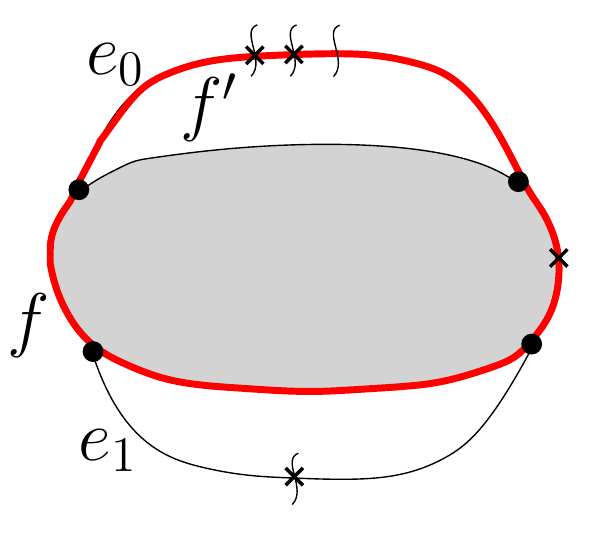}}
				\subfloat[]{\label{fi:flexible_mf=2_notadjacent_c}\includegraphics[width=0.32\columnwidth]{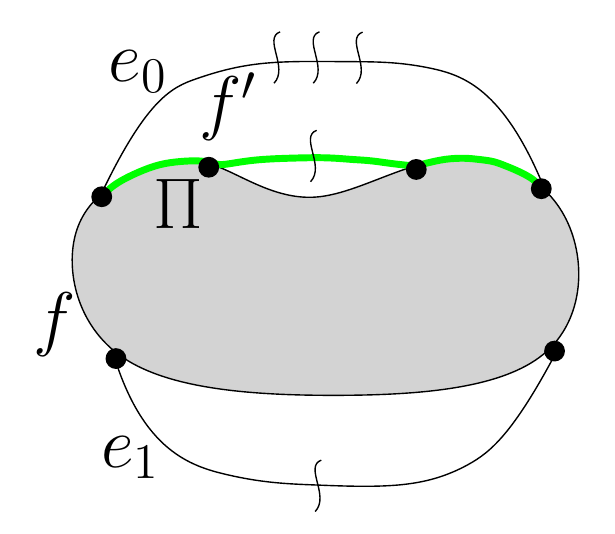}}
				\\
				\subfloat[]{\label{fi:flexible_mf=2_notadjacent_d}\includegraphics[width=0.32\columnwidth]{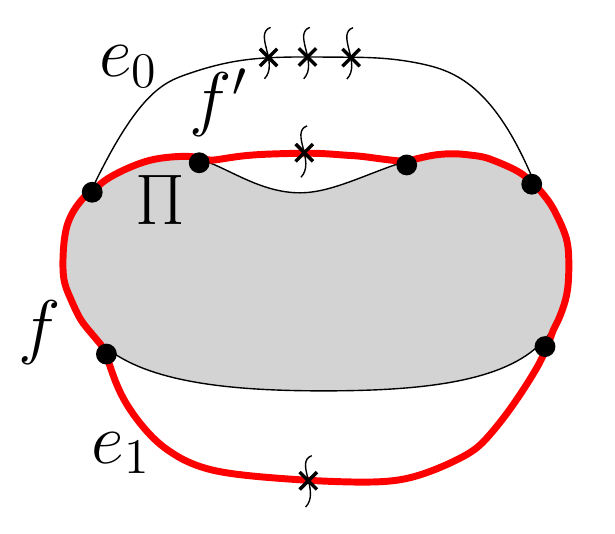}}
				\subfloat[]{\label{fi:flexible_mf=2_notadjacent_e}\includegraphics[width=0.32\columnwidth]{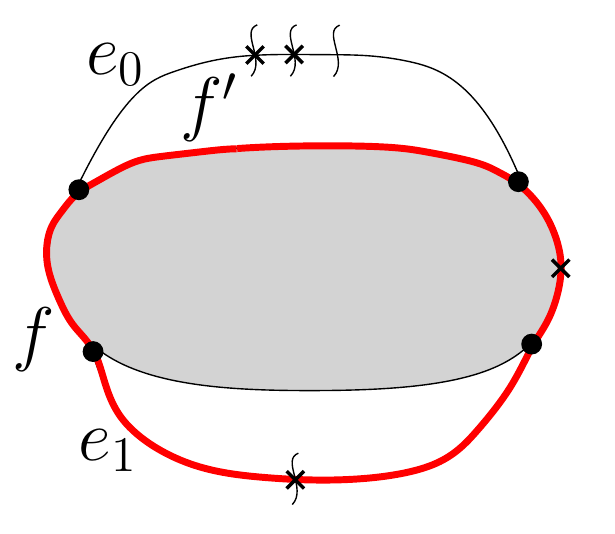}}
				\subfloat[]{\label{fi:flexible_mf=2_notadjacent_f}\includegraphics[width=0.32\columnwidth]{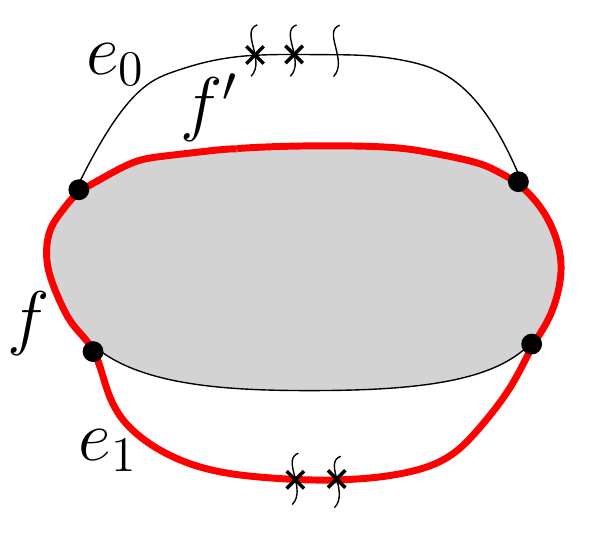}}
				\caption{Illustration for the proof of Theorem~\ref{th:fixed-embedding-min-bend} when $m_f=2$ and the two flexible edges $e_0$ and $e_1$ of $f$ are not adjacent.}\label{flexible_mf=2_notadjacent}
					
		
			\end{figure}
			\item{$e_0$ and $e_1$ share a vertex $v$}. Refer to Case~3.b. of \cref{de:flex-f}. Let $e_2$ be the edge incident to $v$ different from $e_0$ and $e_1$. Observe that inserting bends both along $e_0$ and along $e_1$ creates two 2-extrovert cycles and three 3-extrovert cycles (refer to \cref{flexible_mf=2_adjacent_5cycles}), where for two of the three 3-extrovert cycles Condition~$(iii)$ is satisfied by the bends themselves and only for one 3-extrovert cycle, that we call $C^*$, Condition~$(iii)$ needs to be satisfied.
			We have three subcases.
			\begin{itemize}
				
				\item[$(i)$] Suppose $\flex(e_0) + \flex(e_1) \leq 3$. Assume first that the external face is not a 3-cycle. In this case we place $\flex(e_0)$ bends along $e_0$ and $\flex(e_1)$ bends along $e_1$ and the remaining $4-\flex(e_0)-\flex(e_1)$ bends, if any, along distinct edges of $C_o(G) \setminus \{e_0,e_1\}$. By considerations analogous to those used in the previous cases, Condition~$(ii)$ is satisfied for the two 2-extrovert cycles introduced by the inserted bends and Condition~$(iii)$ is satisfied for $C^*$. Since in this case $\flex(f) = \flex(e_0)+\flex(e_1)$  (Case~3.b.i of \cref{de:flex-f}), \cref{eq:fixed-embedding-cost} holds.
				
				Now, consider the case when the external face is a 3-cycle. Suppose, without loss of generality, $1\leq \flex(e_0)\leq 2$ and $\flex(e_1)=1$.  In this case we place two bends in $e_0$ (one of which is costly if $\flex(e_0)=1$), one bend in $e_1$, and one bend in the other edge of the external face.  By considerations analogous to those used in the previous cases, Condition~$(ii)$ is satisfied for the two 2-extrovert cycles introduced by the inserted bends and Condition~$(iii)$ is satisfied for $C^*$. Since in this case $\flex(f) = \flex(e_0)+\flex(e_1)$  (Case~3.b.i of \cref{de:flex-f}), \cref{eq:fixed-embedding-cost} holds.

				\item[$(ii)$] Suppose $\flex(e_0) \geq 2$ and $\flex(e_1) \geq 2$. If $x_{e_0}+x_{e_1}-2\cdot\flex(e_2) \geq 1$, then we place two bends along $e_0$ and two bends along $e_1$.
				In fact, Condition~$(ii)$ is satisfied for the two 2-extrovert cycles. Since $x_{e_0}+x_{e_1}-2\cdot\flex(e_2) \geq 1$, there is either a flexible edge or a demanding 3-extrovert cycle on the portion of $C^*$ that is not incident to $f$. Analogously to the cases discussed above, we can use this flexible edge or demanding 3-extrovert cycle to satisfy Condition~$(iii)$ for $C^*$ (see \cref{fi:flexible_mf=2_adjacent_a}). Since in this case $\flex(f) = 4$ (Case~3.b.ii.$\alpha$ of \cref{de:flex-f}), we have that \cref{eq:fixed-embedding-cost} holds.
				
				Otherwise, if $x_{e_0}+x_{e_1}-2\cdot\flex(e_2) = 0$, then there is no flexible edge nor demanding 3-extrovert cycle on the portion of $C^*$ that is not incident to $f$.
				Hence, we distribute three among $e_0$ and $e_1$ and we place one costly bend along the portion of $C^*$ incident to $f$ to satisfy Condition~$(ii)$ for the 2-extrovert cycles and Condition~$(iii)$ for $C^*$ (see \cref{fi:flexible_mf=2_adjacent_b}). Since in this case $\flex(f)=3$ (Case~3.b.ii.$\beta$ of \cref{de:flex-f}), \cref{eq:fixed-embedding-cost} holds.
				
				\item[$(iii)$] Suppose $\flex(e_0) \geq 3$ and $\flex(e_1) = 1$. Let $f' \neq f$ be the other face incident to $e_0$ and let $f'' \neq f$ be the other face incident to $e_1$.
				
				If $x_{e_0}-\flex(e_2) \geq 1$, then we place three bends along $e_0$ and one bend along $e_1$, thus satisfying Condition~$(i)$ for the external face $f$ and Condition~$(ii)$ for the 2-extrovert cycle introduced by the bend along $e_1$. Regarding Condition~$(ii)$ for the 2-extrovert cycle introduced by the bends along $e_0$ and Condition~$(iii)$ for $C^*$, we observe that since $x_{e_0}-\flex(e_2) \geq 1$, then there is either a flexible edge or a demanding 3-extrovert cycle on $C_o(G) \setminus \{e_0,e_1\}$. Analogously to the previous cases, we can use this flexible edge or demanding 3-extrovert cycle to satisfy the above mentioned conditions (see \cref{fi:flexible_mf=2_adjacent_c}). Since in this case $\flex(f) = 4$ (Case~3.b.iii.$\alpha$ of \cref{de:flex-f}), we have that \cref{eq:fixed-embedding-cost} holds.
				
				If we have that: (a) $x_{e_0}-\flex(e_2) = 0$, (b) $x_{e_1} - \flex(e_2) \geq 1$, and (c) $\flex(e_2) \geq 1$, then we place three bends along $e_0$, one bend along $e_1$, and one bend along $e_2$.
				In fact, Condition~$(i)$ is satisfied for the external face $f$ and Condition~$(ii)$ is satisfied for the 2-extrovert cycle introduced by the bend along $e_1$. The bend along $e_2$ and the bend along $e_1$ satisfy Condition~$(ii)$ for the 2-extrovert cycle introduced by the bends along $e_0$ (see \cref{fi:flexible_mf=2_adjacent_d}).
				Regarding Condition~$(iii)$ for $C^*$, we observe that since $x_{e_1}-\flex(e_2) \geq 1$, then there is either a flexible edge or a demanding 3-extrovert cycle on the portion of $C^*$ that is incident to $f''$. Analogously to the previous cases, we can use this flexible edge or demanding 3-extrovert cycle to satisfy Condition~$(iii)$ (see \cref{fi:flexible_mf=2_adjacent_e}).
				Since in this case $\flex(f) = 4$ (Case~3.b.iii.$\beta$ of \cref{de:flex-f}), \cref{eq:fixed-embedding-cost} holds.
				
				Otherwise, we have $x_{e_0}-\flex(e_2) = 0$ and either $x_{e_1} - \flex(e_2) = 0$ or $\flex(e_2) = 0$. In this case we place two bends along $e_0$, one bend along $e_1$ and one costly bend along $C_o(G) \setminus \{e_0,e_1\}$. In fact, Condition~$(i)$ is satisfied for the external face $f$ and Condition~$(ii)$ is satisfied for the 2-extrovert cycle introduced by the bend along $e_1$. The bend along $e_1$ and the bend along $C_o(G) \setminus \{e_0,e_1\}$ satisfy Condition~$(ii)$ for the 2-extrovert cycle introduced by the bends along $e_0$.
				Also the bend along $C_o(G) \setminus \{e_0,e_1\}$ satisfies Condition~$(iii)$ for $C^*$ (see \cref{fi:flexible_mf=2_adjacent_f}).
				Since in this case $\flex(f)=3$ (Case~3.b.iii.$\gamma$ of \cref{de:flex-f}), \cref{eq:fixed-embedding-cost} holds.
				
			\end{itemize}
		\end{itemize}
		\begin{figure}[htb]
			\centering
			{\label{fi:flexible_mf=2_adjacent_5cycles}\includegraphics[width=\columnwidth]{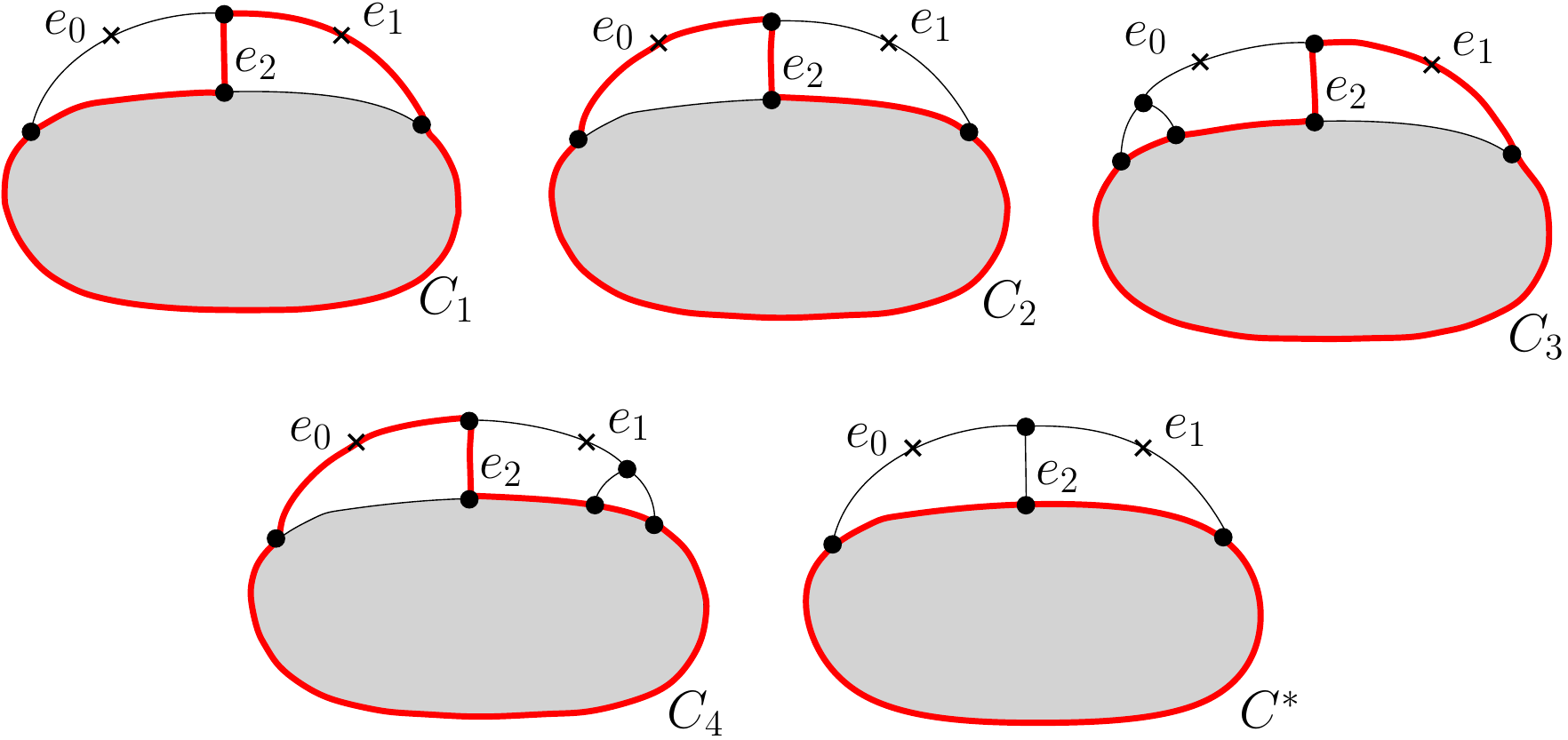}}
			\caption{Two adjacent edges $e_0$ and $e_1$ of the external face, two 2-extrovert cycles ($C_1$ and $C_2$), and three 3-extrovert cycles ($C_3$, $C_4$, and $C^*$) created by placing bends in $e_0$ and $e_1$. Notice that Condition~$(iii)$ for $C_3$ and $C_4$ is satisfied by the bends in $e_0$ and $e_1$.}\label{flexible_mf=2_adjacent_5cycles}
		\end{figure}
		\begin{figure}[htb]
			\centering
			\subfloat[]{\label{fi:flexible_mf=2_adjacent_a}\includegraphics[width=0.32\columnwidth]{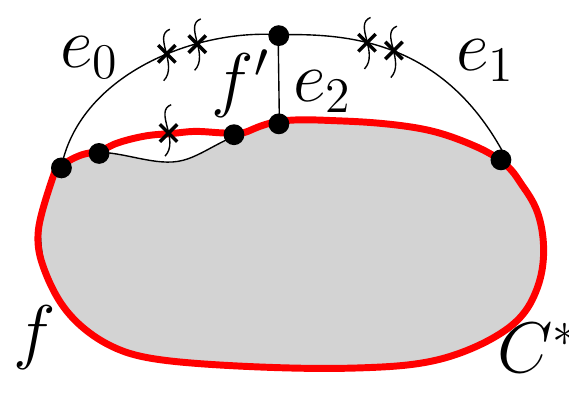}}
			\subfloat[]{\label{fi:flexible_mf=2_adjacent_b}\includegraphics[width=0.32\columnwidth]{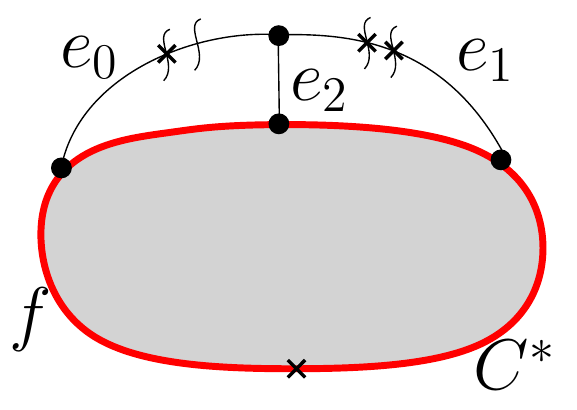}}
			\subfloat[]{\label{fi:flexible_mf=2_adjacent_c}\includegraphics[width=0.32\columnwidth]{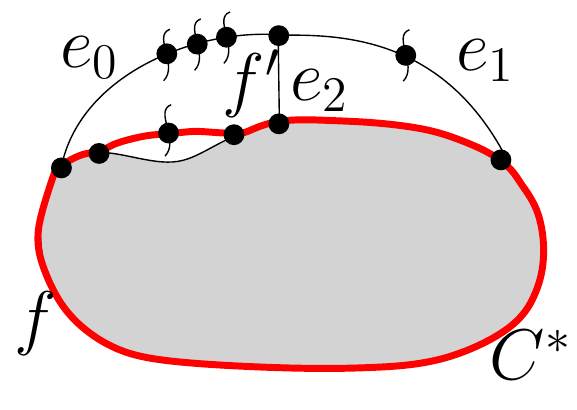}}\\
			\subfloat[]{\label{fi:flexible_mf=2_adjacent_d}\includegraphics[width=0.32\columnwidth]{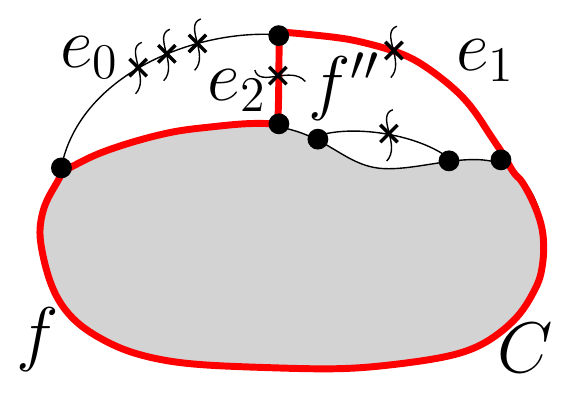}}
			\subfloat[]{\label{fi:flexible_mf=2_adjacent_e}\includegraphics[width=0.32\columnwidth]{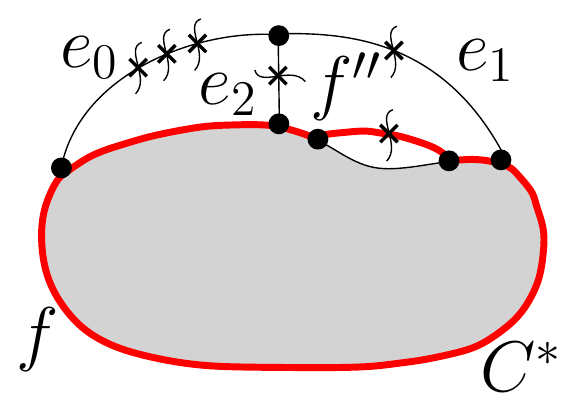}}
			\subfloat[]{\label{fi:flexible_mf=2_adjacent_f}\includegraphics[width=0.32\columnwidth]{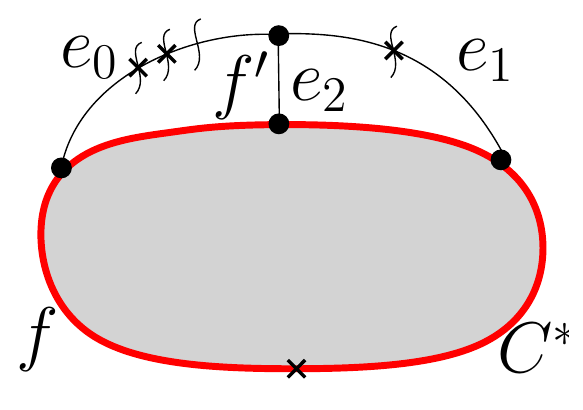}}
			\caption{Illustration of the proof of Theorem~\ref{th:fixed-embedding-min-bend} when $m_f=2$ and the two flexible edges $e_0$ and $e_1$ of $f$ are adjacent.} \label{flexible_mf=2_adjacent}
				
		\end{figure}
	
	\smallskip\paragraph{Case $m_f \geq 3$.} Refer to Case~4 of \cref{de:flex-f}. Suppose that the external face is not a 3-cycle. In this case we are able to distribute up to four bends to the flexible edges of $f$ in such a way that at most two bends are placed along the same flexible edge. Also, if the sum of the flexibilities of the edges incident to $f$ is three, we introduce one costly bend along an inflexible edge incident to $f$. This guarantees that Condition~$(i)$ is satisfied for $f$ and that Conditions~$(ii)$ and~$(iii)$ are satisfied for the 2-extrovert and 3-extrovert cycles, respectively, introduced by these bends. In this case, $\flex(f) = \min\{4,\sum_{e \in C_f}\flex(e)\}$ and, hence, \cref{eq:fixed-embedding-cost} holds.
	
	Now, suppose that the external face is a 3-cycle. If there exists at least one edge $e_0\in C_o(G)$ such that $\flex(e_0)\ge 2$, we place two bends along $e_0$ and one bend for each other (flexible) edge of $C_o(G)$. Else, we place two bends along an edge of $C_o(G)$ (one of which is costly) and one bend along every other edge of $C_o(G)$.  In this case, $\flex(f) = \min\{4,\sum_{e \in C_f}\flex(e)\}$ and, hence,  \cref{eq:fixed-embedding-cost} holds.
	
	Observe that, in all cases analyzed above, each inflexible edge receives at most one bend, except when the external face is a 3-cycle with all inflexible edges; in this case one of these inflexible edges is bent twice to satisfy Condition~$(i)$ of Theorem~\ref{th:RN03}. Also, each flexible edge has a number of bends that does not exceed its flexibility except when the external face is a 3-cycle with at least a flexible edge and all the external edges have flexibility at most one; in this case, one of the flexible edges of the external face is bent twice to satisfy Condition~$(i)$ of Theorem~\ref{th:RN03}.
	
	Concerning the computational cost of constructing $H$, suppose that $D(G)$ and $D_f(G)$ are given. Since the case analysis described above can be easily performed in $O(n)$ time, we can construct $\rect{G}$ by suitably subdividing the edges of $G$ with degree-2 vertices (which will represent the bends of $H$) in $O(n)$ time. A rectilinear representation $\rect{H}$ of $\rect{G}$ is computed in $O(n)$ time by applying the \textsf{NoBendAlg} of~\cite{DBLP:journals/jgaa/RahmanNN03}, and the desired representation $H$ is the inverse of $\rect{H}$.         	
\end{proof}

An example of application of Theorem~\ref{th:RN03} is given in \cref{fi:rahman_colouration-c}, which depicts a bend-minimum orthogonal representation $H$ of the plane graph $G$ in \cref{fi:rahman_colouration-a} having $f$ as external face. We have $D(G) = \{C_1, C_2, C_3, C_7, C_8,C_9\}$, $D_{f}(G) = \{C_7, C_8\}$, and $\flex(f)=1$. Thus, $c(G)=c(H)=|D(G)| + 4 - \min\{4, |D_{f}(G)| + \flex(f) \}=6+4-3=7$.

\begin{figure}[tb]
	\centering
	{\includegraphics[width=\columnwidth]{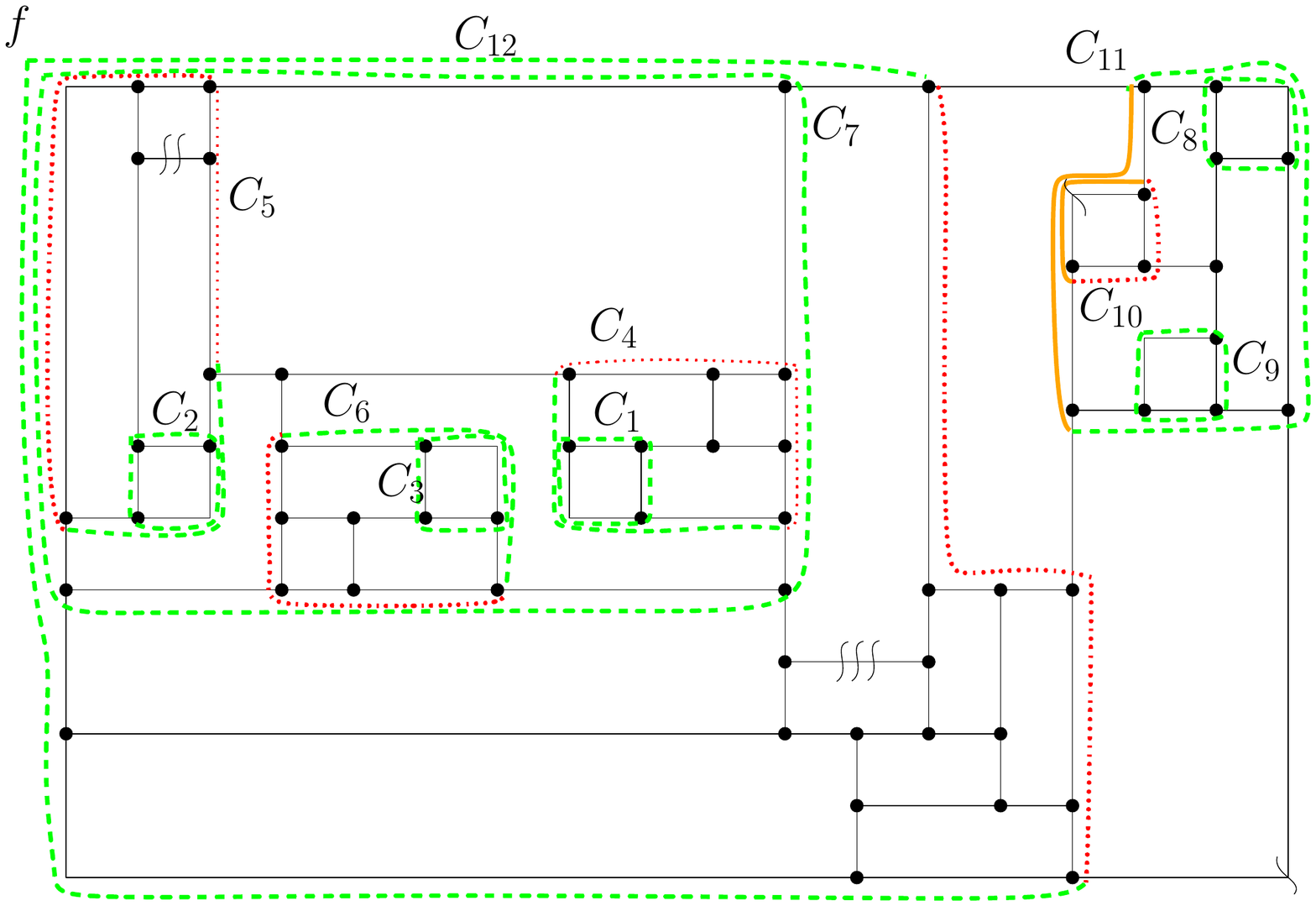}}
	\caption{A bend-minimum orthogonal representation of the plane graph in \cref{fi:rahman_colouration-a}.}\label{fi:rahman_colouration-c}
\end{figure}

\subsection{The Variable Embedding Setting with Flexible Edges.}\label{sse:variable-embedding}

Now we consider the case in which the external face of $G$ can be freely chosen and we want to efficiently compute the cost of a cost-minimum orthogonal representation of $G$. A high-level description of our algorithm consists of the following steps:
$(i)$ We suitably choose a planar embedding of $G$ called ``reference embedding''.
$(ii)$ We use the reference embedding to construct in $O(n)$ time the \texttt{Bend-Counter} data structure.
For any possible external face $f$ of $G$, the \texttt{Bend-Counter} returns in $O(1)$ time the cost of a cost-minimum orthogonal representation of $G$ with $f$ as external face.

A \emph{reference embedding} is an embedding of $G$ such that no 3-extrovert cycle has an edge incident to the external face. We denote by $G_f$ the plane graph corresponding to a reference embedding whose external face is~$f$.
Roughly speaking, the \texttt{Bend-Counter} stores information that makes it possible to compute how the values in the formula of Theorem~\ref{th:fixed-embedding-min-bend} change when we change the external face of $G$. In fact, when choosing a different external face for a triconnected cubic graph, some demanding 3-extrovert cycles are preserved, some may disappear, and some new ones may appear. A 3-extrovert cycle disappears when it becomes a 3-introvert cycle in the new embedding; a 3-extrovert cycle appears when a 3-introvert cycle is turned inside-out in the new embedding.

The main components of the \texttt{Bend-Counter} are: A tree, which we call ``3-extrovert tree'', and an array of pointers to the nodes of the 3-extrovert tree, which we call ``face array''.
Let $G_f$ be a triconnected plane 3-graph with a reference embedding.
The \emph{3-extrovert tree} $T_f$ of $G_f$ is rooted at the cycle $C_f$ of the external face $f$; every other node of $T_f$ is a 3-extrovert cycle of $G_f$; if $C'$ and $C$ are two 3-extrovert cycles such that $C$ is a child-cycle of $C'$, then $C'$ is the parent of $C$ in $T_f$. The children of the root of $T_f$ are the \emph{maximal} 3-extrovert cycles of $G_f$, i.e., those that are not child-cycles of any other 3-extrovert cycle of~$G_f$.
The \emph{face array} $A_f$ has an entry for every face $f'$ of $G_f$, which points to the lowest node $C'$ of $T_f$ such that $G_f(C')$ contains $f'$.

%
Since Theorem~\ref{th:fixed-embedding-min-bend} only considers non-trivial 3-extrovert cycles, in the remainder of this section we only consider non-trivial 3-extrovert and 3-introvert cycles, and we omit the term ``non-trivial''.

\subsubsection{Computing Reference Embeddings and 3-Extrovert Trees.}\label{sse:reference-embedding}
%

One could think that, for any given embedding, the 3-extrovert tree can be obtained by connecting the roots of the genealogical trees of the maximal 3-extrovert cycles to a common node. Unfortunately, this is not always the case and we need to identify a reference embedding to guarantee the existence of a 3-extrovert tree. Consider for example the planar embedding of \cref{fi:reference-embedding-a}. Since $C_2$ is a child-cycle of both $C_5$ and $C_6$, if we connected the roots of the two genealogical trees $T_{C_5}$ and $T_{C_6}$ to a common node, we would obtain a cyclic graph.

\begin{figure}[t]
	\centering
	\subfloat[]{\label{fi:reference-embedding-a}\includegraphics[width=0.5\columnwidth]{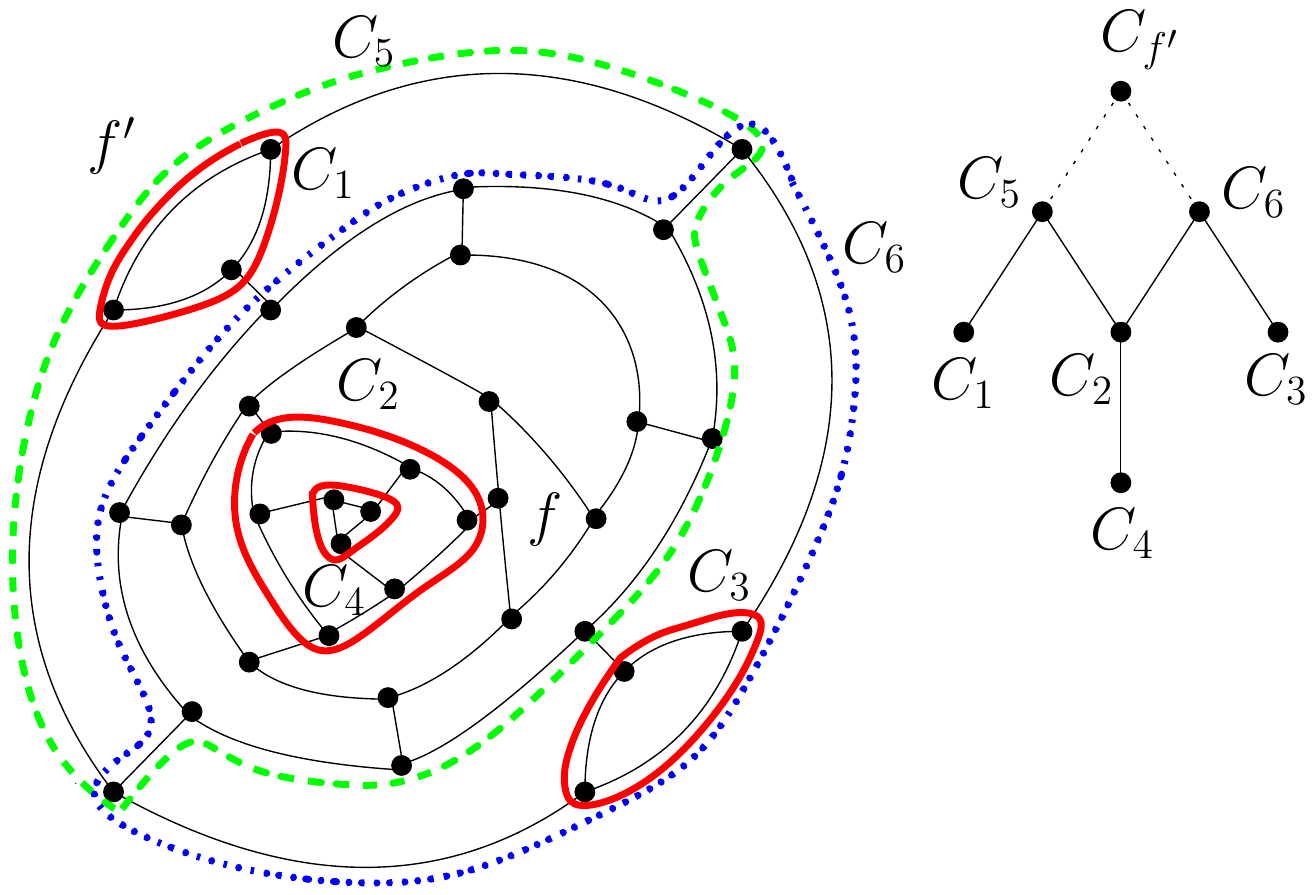}}
	\hfill
	\subfloat[]{\label{fi:reference-embedding-b}\includegraphics[width=0.5\columnwidth]{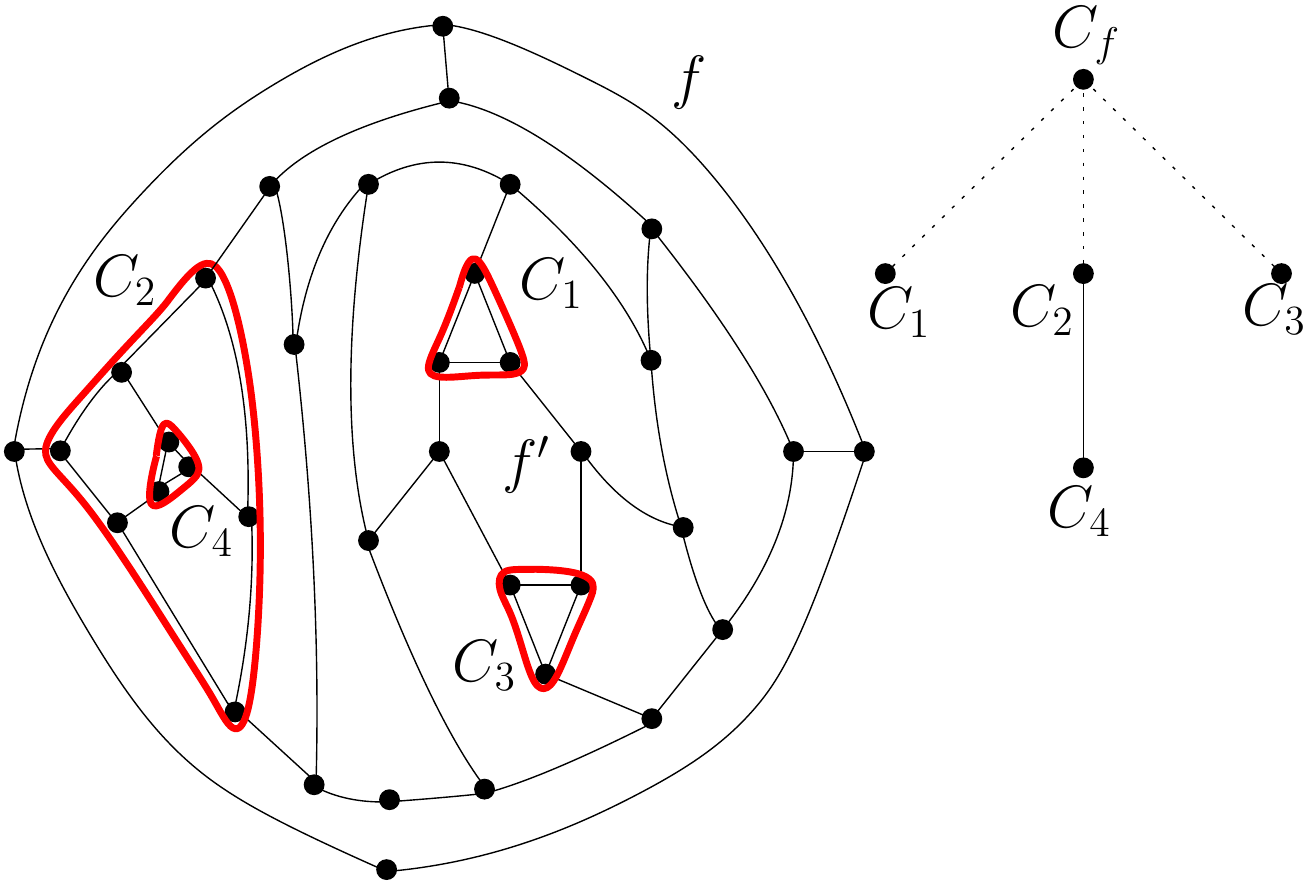}}
	\caption{(a) An embedding for which the inclusion relationships between the 3-extrovert cycles cannot be described by a rooted tree. (b) A reference embedding of the same graph and a rooted tree that describes the inclusion relationships between the 3-extrovert cycles.}\label{fi:refernce-embedding}
\end{figure}

Let $C$ be a 3-extrovert (3-introvert) cycle of $G_f$. The three faces of $G_f$ that are incident to the external (internal) legs of $C$ are called \emph{leg faces} of $C$. For example, face $f'$ in \cref{fi:reference-embedding-a} is a leg face of four 3-extrovert cycles, namely $C_1$, $C_3$, $C_5$, and $C_6$. In \cref{fi:reference-embedding-b} face $f'$ is a leg face of the two 3-extrovert cycles $C_1$ and $C_3$. Note that, since $G_f$ is triconnected and cubic, any two leg faces of a 3-extrovert (3-introvert) cycle only share a leg.

Since the reference embedding is such that the external face $f$ is not incident to any 3-extrovert cycle we have that $f$ is not a leg face of any 3-extrovert cycle. For example, the embedding of \cref{fi:reference-embedding-b} is a reference embedding, because the external face $f$ is not a leg face of any 3-extrovert cycle. Conversely, \cref{fi:reference-embedding-a} shows a different embedding of the same graph where the external face $f'$ is a leg face of some 3-extrovert cycles; thus, the plane embedding of \cref{fi:reference-embedding-a} is not a reference embedding.
\begin{restatable}{lemma}{leRefEmbedding}\label{le:ref-embedding}
	A cubic triconnected planar graph 
	always admits a reference embedding, which can be computed in $O(n)$ time.
\end{restatable}
\begin{proof}
	Let $G_f$ be a cubic triconnected plane graph with $f$ as its external face. If $G_{f}$ has no 3-extrovert cycle, then the embedding of $G_{f}$ is already a reference embedding. Otherwise, let $C_1$ be any 3-extrovert cycle of $G_{f}$ and let $T_{C_1}$ be its genealogical tree. Let $C_2$ be a leaf of $T_{C_1}$ (possibly coincident with $C_1$). Clearly, $G_{f}(C_2)$ does not contain any leg of a 3-extrovert cycle of $G_{f}$, which implies that every internal face of $G_{f}(C_2)$ is not a leg face. Hence, for any internal face $f'$ of $G_{f}(C_2)$, the embedding of $G_{f'}$ is a reference embedding.
	Since the set of 3-extrovert cycles of a plane graph can be computed in $O(n)$ time~\cite{DBLP:journals/jgaa/RahmanNN99}, the reference embedding can also be computed in $O(n)$ time.
\end{proof}

The next lemma extends \cref{le:independent-child-cycles}.

\begin{restatable}{lemma}{leIndependentChildCyclesRefEmbedding}\label{le:independent-child-cycles-ref-embedding}
	If the planar embedding of $G_f$ is a reference embedding, the 3-extrovert child-cycles of cycle $C_f$ are independent.
\end{restatable}
\begin{proof}
	Let $C_1$ and $C_2$ be two 3-extrovert child-cycles of $C_f$. By definition of child-cycles, neither $C_1$ is a descendant of $C_2$ nor $C_2$ is a descendant of $C_1$. Also, since $f$ is not a leg face, $C_1$ and $C_2$ do not contain any edge of $f$. Suppose by contradiction that $C_1$ and $C_2$ are dependent. Since $\Delta(G_f) \leq 3$, $C_1$ and $C_2$ have at least an edge in common. It follows that two of the three legs of $C_1$ are edges of $C_2$ and vice versa (see for example \cref{fi:pasticca}). Hence, the external cycle of $G_f(C_1) \cup G_f(C_2)$ is a 2-extrovert cycle, which is impossible since $G$ is triconnected.
\end{proof}

In the reference embedding depicted in \cref{fi:reference-embedding-b} the 3-extrovert child-cycles of $C_f$ are $C_1$, $C_2$, and $C_3$ and they are independent. If the embedding is not a reference embedding this property is not guaranteed: The embedding depicted in \cref{fi:reference-embedding-a} is not a reference embedding, the 3-extrovert child-cycles of $C_f$ are $C_5$ and $C_6$ and they are not independent.

By \cref{le:independent-child-cycles,le:independent-child-cycles-ref-embedding}, if the embedding of $G_f$ is a reference embedding, the 3-extrovert tree $T_f$ can be defined. Also, by \cref{le:ref-embedding}, a reference embedding and a 3-extrovert tree always exist.

\begin{restatable}{lemma}{leDemandingTreeLineartime}\label{le:demanding-tree-lineartime}
	Let $G_f$ be a cubic 3-connected plane graph with a reference embedding. The 3-extrovert tree $T_f$ of $G_f$ can be computed in $O(n)$ time.
\end{restatable}
\begin{proof}
	Rahman et al.~\cite{DBLP:journals/jgaa/RahmanNN99} prove that the set of genealogical trees $T_C$ for any child-cycle $C$ of $C_f$ can be computed in $O(n)$ time.
	$T_f$ is obtained by connecting all such $T_C$ to a unique root, which still takes $O(n)$ time.
\end{proof}

\begin{figure}[h]
	\centering
	\includegraphics[width=0.5\columnwidth]{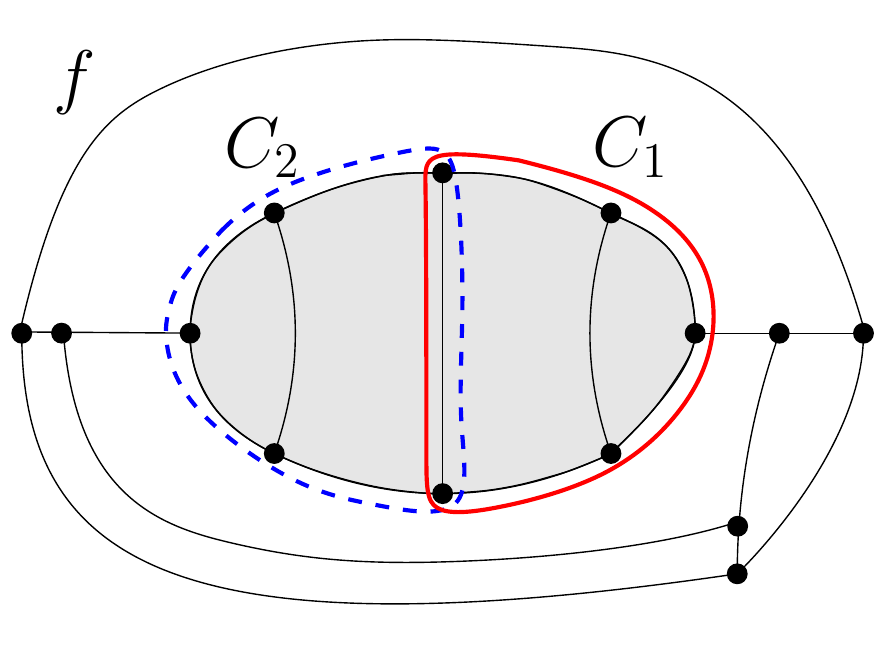}
	\caption{Illustration of the proof of \cref{le:independent-child-cycles-ref-embedding}. The external cycle of the graph $G_f(C_1) \cup G_f(C_2)$ (shaded) is a 2-extrovert cycle. }\label{fi:pasticca}
\end{figure}	

In the following we give some properties of the 3-extrovert and 3-introvert cycles of a graph $G_f$ whose embedding is a reference embedding. Also, we give some definitions and properties that will help to understand how the set of demanding 3-extrovert cycles changes when we change the external face of $G_f$. Since $|D_{f}(G_f)|=0$, by Theorem~\ref{th:fixed-embedding-min-bend} we immediately have the following.

\begin{corollary}\label{co:min-bend-ref-embedding}
Let $G_f$ be a plane triconnected cubic graph whose embedding is a reference embedding. Then $c(G_f) = |D(G_f)| + 4 -\min\{4,\flex(f)\}$.
\end{corollary}

For a graph $G_f$ whose embedding is a reference embedding we can efficiently compute the coloring of the contour paths of its 3-extrovert cycles and its demanding 3-extrovert cycles.

\begin{restatable}{lemma}{leThreeExtrovertColoringLinear}\label{le:3-extrovert-coloring-linear}
	The red-green-orange coloring of the contour paths of the 3-extrovert cycles of $G_f$ can be computed in $O(n)$ time.
\end{restatable}

\begin{proof}
	We color the contour paths through a postorder visit of the 3-extrovert tree~$T_f$ of~$G_f$. Every contour path $P$ of each 3-extrovert cycle is equipped with a counter $\flex(P)$ that reports the number of flexible edges in $P$ (all values of $\flex(P)$ can be computed in $O(n)$ time through a bottom-up visit of $T_f$).
	
	
	The algorithm consists of two steps. In the first step we assign an orange or green color to some contour paths. The color of a contour path $P$ may remain undefined at the end of this step. In the second step we assign the colors to the undefined contour paths. We now describe the two steps and then discuss the time complexity.
	
	\begin{description}
		\item\textsf{First Step}: Let $C$ be a non-root node of $T_f$, $P$ be a contour path of $C$, and $f'$ be the leg face of $C$ incident to $P$. Let $C_1,...,C_k$ be the child-cycles of $C$ in $T_f$ having $f'$ as a leg face. We assume that the contour paths of the cycles $C_1,...,C_k$ are already colored. Denote by $P_i$ the contour path of cycle $C_i$, $1 \leq i \leq k$, incident to $f'$.
		\begin{itemize}
			\item If $\flex(P)>0$ (i.e., $P$ contains a flexible edge) $P$ is colored orange (see Case~2(a) of the \textsc{3-Extrovert Coloring Rule}).
			\item If $\flex(P)=0$ and if it exists a path $P_i$, $1 \le i \le k$, such that $P_i$ is green, we color $P$ of green (see Case~2(b) of the \textsc{3-Extrovert Coloring Rule}).
			\item Otherwise the color of $P$ remains undefined.
		\end{itemize}
		
		\item\textsf{Second Step}: of the 3-extrovert cycle $C$ are colored according to the first step.
		Let $C$ be a non-root node of $T_f$. If the color of all the paths of $C$ is undefined, we color all of them as green (see Case~1 of the \textsc{3-Extrovert Coloring Rule}). Otherwise, we color the undefined colored paths as red (see Case~2(c) of the \textsc{3-Extrovert Coloring Rule}).
		
	\end{description}
	
	The correctness of the algorithm above is easily established by observing that whenever a contour path is given the red, green, or orange color this is done in accordance with the \textsc{3-Extrovert Coloring Rule}. As for the time complexity, observe that the algorithm considers each contour path a constant number of times and that the number of contour paths is $O(n)$.
\end{proof}

\begin{restatable}{lemma}{leDemandingExtrovertLineartime}\label{le:demanding-3-extrovert-lineartime}
	The demanding 3-extrovert cycles of $G_f$ can be computed in $O(n)$ time.
\end{restatable}

\begin{proof}
	By \cref{le:demanding-tree-lineartime}, $T_f$ is computed in $O(n)$ time. By \cref{le:3-extrovert-coloring-linear} the red-green-orange coloring of the contour paths of the 3-extrovert cycles of $G_f$ is computed in $O(n)$ time. Let $C$ be a 3-extrovert cycle of $T_f$. If $C$ has a contour path that is not green it is not demanding. Suppose that $C$ has three green contour paths. Let $C^*$ be a child-cycle of $C$. A contour path $P^*$ of $C^*$ is contained in a contour path $P$ of $C$ if the leg face of $C^*$ incident to $P^*$ is equal to the leg face of $C$ incident to $P$. Let $n_C$ be the number of child-cycles of $C$ in $P$. Every 3-extrovert cycle has three contour paths, hence we can check if $C$ has a contour path containing a green contour path of one of its child-cycles in $O(n_C)$ time. By definition $C$ is a demanding if and only if it has no contour path containing a green contour path of one of its child-cycles. We can perform this operation for all the 3-extrovert cycles of $T_f$ in overall $O(n)$ time.
\end{proof}

The next lemma proves that in a reference embedding we can establish a one-to-one correspondence between the set of (non-trivial) 3-extrovert cycles and the set of (non-trivial) 3-introvert cycles: every 3-extrovert cycle can be associated with a 3-introvert cycle with the same set of legs, and vice-versa.

\begin{restatable}{lemma}{leThreeExtroThreeIntro}\label{le:3-extro-3-intro}
Let $G_f$ be a plane triconnected cubic graph whose embedding is a reference embedding and let $\mathcal{E}$ and $\mathcal{I}$ be the sets of 3-extrovert and 3-introvert cycles of $G_f$, respectively. There is a one-to-one correspondence $\phi: \mathcal{E} \rightarrow \mathcal{I}$ such that $C$ and $\phi(C)$ have the same legs, for every $C \in \mathcal{E}$. 	
\end{restatable}

\begin{proof}
	Consider a (non-trivial) 3-extrovert cycle $C \in \mathcal{E}$. For $i = 1,2,3$, let $e_i=(u_i,v_i)$ be the legs of $C$, where $u_i$ belongs to $C$.
	Finally, let $f_{ij}$ be the leg face of $C$ that contains the legs $e_i$ and $e_j$ $(1 \leq i<j \leq 3)$.
	Observe that, vertices $v_1$, $v_2$, $v_3$ are distinct, as otherwise $C$ would be trivial or $G$ would be not triconnected. Hence we can consider the cycle $\phi(C)$ formed by the union of the three paths from $v_i$ to $v_j$ along $f_{ij}$, not passing for any leg of $C$. \cref{fi:reference-extro-intro} shows a plane graph $G_f$ with a reference embedding, a 3-extrovert cycle $C$ and its correspondent 3-introvert cycle $\phi(C)$.
	Note that $\phi(C)$ is a simple cycle because any two leg faces of $C$ only share a leg.
	Since the embedding of $G_f$ is a reference embedding, any leg of $C$ does not belong to the external face $f$, and therefore every $f_{ij}$ is an internal face. It follows that $\phi(C)$ is a 3-introvert cycle with legs $e_i=(u_i,v_i)$ $(i = 1,2,3)$.
	
	\begin{figure}[htb]
		\centering
		\subfloat[]{\label{fi:reference-extro-intro}\includegraphics[width=0.48\columnwidth]{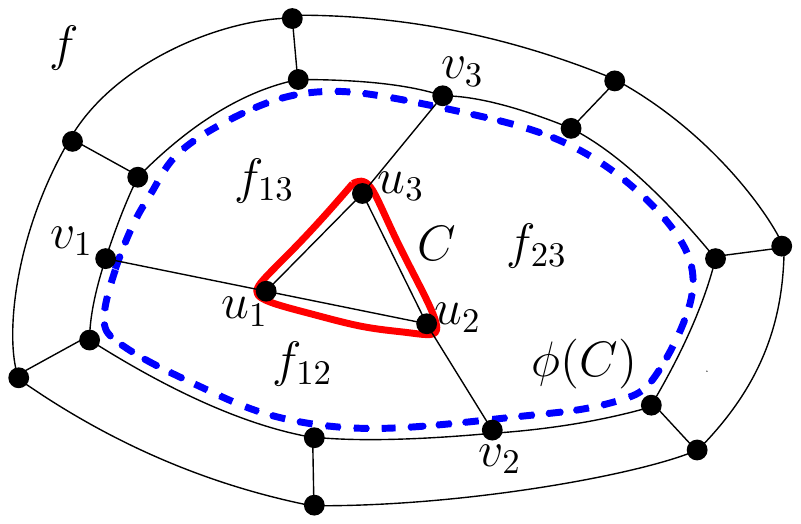}}
		\hfil
		\subfloat[]{\label{fi:3-introvert-cycles-intersecting}\includegraphics[width=0.48\columnwidth]{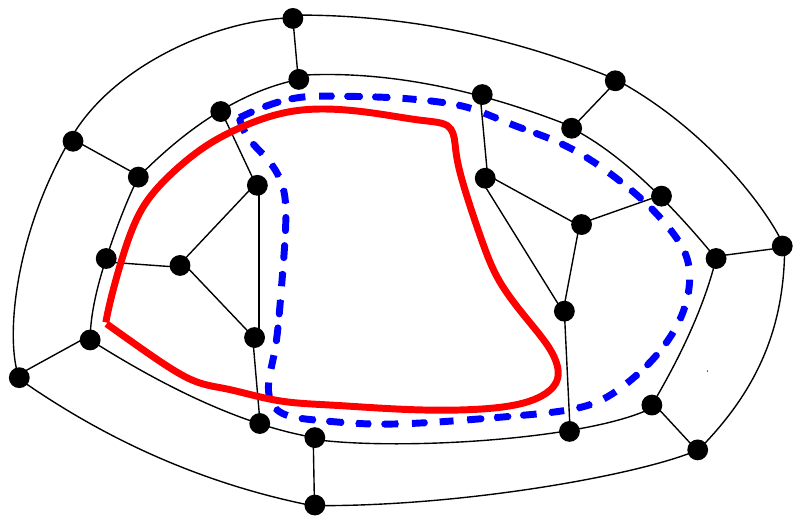}}
		\caption{(a) A plane graph $G_f$ with a reference embedding, a 3-extrovert cycle $C$ and its correspondent 3-introvert cycle $\phi(C)$. (a) Two 3-introvert extrovert cycles that are not edge-disjoint.}
	\end{figure}

	Vice versa, let $C'$ be a 3-introvert cycle of $G_f$. Consider the plane subgraph $G' \subseteq G_f(C')$ obtained by removing $C'$ and its legs. The cycle $C_o(G')$ is a simple cycle and we can set $\phi^{-1}(C')=C_o(G')$. Clearly, $\phi^{-1}(C')$ is a 3-extrovert cycle whose legs coincide with those of $C'$. 	
\end{proof}

Note that, if $C$ and $C'$ are two distinct 3-extrovert cycles of $G_f$, the corresponding 3-introvert cycles $\phi(C)$ and $\phi(C')$ are not necessarily edge-disjoint (see \cref{fi:3-introvert-cycles-intersecting}).

\subsubsection{Demanding 3-introvert cycles.}\label{sse:demanding-3-introvert}
Let $G_f$ be a plane 3-graph whose embedding is a reference embedding.
When we change the planar embedding of $G_f$, by choosing any external face $f' \neq f$, some 3-introvert cycles of $G_f$ may become demanding 3-extrovert cycles in $G_{f'}$. We call such 3-introvert cycles of $G_f$ the \emph{demanding 3-introvert cycles} of $G_f$.
To efficiently identify the set of demanding 3-introvert cycles of $G_f$, we give a characterization of this set (see \cref{le:demanding-3-introvert-charact}) that uses a suitable red-green-orange coloring of the contour paths of the 3-introvert cycles. Let $C$ be any 3-extrovert cycle of $G_f$, let $\phi(C)$ be the 3-introvert cycle associated with $C$ according to \cref{le:3-extro-3-intro}, and let $C'$ be the parent node of $C$ in $T_f$. We assume that the three contour paths of $C$ are colored red, green, or orange according to the \textsc{3-Extrovert Coloring Rule} of \cref{sse:fixed-embedding}. The coloring of the contour paths of $\phi(C)$ depends on the coloring of the siblings of $C$ in $T_f$; in the case that $C'$ is not the root $C_f$ of $T_f$, the coloring of $\phi(C)$ also depends of the coloring of $\phi(C')$; note that if $C'=C_f$, cycle $\phi(C')$ is not defined. See \cref{fi:introvert-colouration-1} for an illustration.

\medskip\noindent \textsc{3-Introvert Coloring Rule}: The three contour paths of $\phi(C)$ are colored according to the following two cases.
\begin{enumerate}
\item $\phi(C)$ has no contour path that contains either a flexible edge or a green contour path of a sibling of $C$ in $T_f$ or a green contour path of $\phi(C')$ (if $\phi(C')$ is defined); in this case all three contour paths of $\phi(C)$ are colored green.
\item Otherwise, let $P$ be a contour path of $\phi(C)$.
\begin{itemize}
	\item [(a)] If $P$ contains a flexible edge then $P$ is colored orange.
	\item [(b)] If $P$ does not contain a flexible edge and it contains either a green contour path of a sibling of $C$ in $T_f$ or a green contour path of $\phi(C')$ (if $\phi(C')$ is defined), then $P$ is colored green.
	\item [(c)] Otherwise $P$ is colored red.
\end{itemize}
\end{enumerate}
\medskip

\cref{fi:introvert-colouration-1} illustrates a 3-extrovert tree (\cref{fi:introvert-colouration-1-a}) and the different cases of the \textsc{3-Introvert Coloring Rule}. Refer to \cref{fi:introvert-colouration-1-b} and focus on the 3-extrovert cycle $C_1$ and the 3-introvert cycle $\phi(C_1)$, the leg faces of which are shaded: the contour paths of $\phi(C_1)$ do no contain a flexible edge or a green contour path of a sibling of $C_1$, which are $C_2$ and $C_3$; the parent node of $C_1$ is $C_f$, that is not associated to a 3-introvert cycle. Consequently, $\phi(C_1)$ has three three green contour paths according to Case~1 of the \textsc{3-Introvert Coloring Rule}. Also in the case of cycle $C_2$ in \cref{fi:introvert-colouration-1-c} the parent node is $C_f$; one of the leg faces of  $C_2$ is adjacent to a green path of a sibling of $C_2$ (namely a contour path of $C_3$) and $\phi(C_2)$ does not contain any flexible edge. Hence, two contour paths of $\phi(C_2)$ are red and one is green according to Case~2(b) and Case~2(c) of the \textsc{3-Introvert Coloring Rule}. The case of cycle $C_3$ \cref{fi:introvert-colouration-1-d} is similar. The parent node of cycle $C_4$ in \cref{fi:introvert-colouration-2-a} is $C_1$: $\phi(C_4)$ has a contour path containing a flexible edge; $C_4$ does not have any sibling; there is no contour path of $\phi(C_4)$ containing a green contour path of $\phi(C_1)$. Therefore, $\phi(C_4)$ has one orange contour path and two red contour paths according to Case~2(a) and Case~2(c) of the \textsc{3-Introvert Coloring Rule}. Finally, focus on $\phi(C_5)$ in \cref{fi:introvert-colouration-2-b}: the contour paths of $\phi(C_5)$ do not contain any flexible edge; $C_5$ does not have any sibling; a contour path of $\phi(C_5)$ contains a green contour path of the 3-introvert cycle $\phi(C_3)$, which is associated to the parent $C_3$ of $C_5$ in $T_f$. Consequently, two contour paths of $\phi(C_5)$ are red ad one is green according to  according to Case~2(b) and Case~2(c) of the \textsc{3-Introvert Coloring Rule}.

\begin{figure}[!ht]
	\centering
	\hfil
	\subfloat[]{\label{fi:introvert-colouration-1-a}\includegraphics[width=0.49\columnwidth]{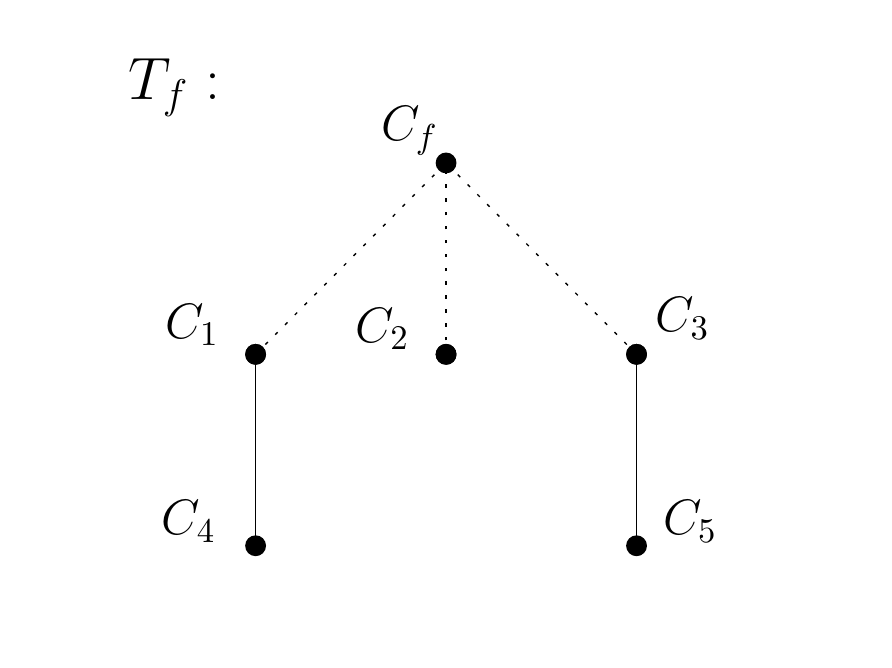}} 
	\hfil
	\subfloat[]{\label{fi:introvert-colouration-1-b}\includegraphics[width=0.47\columnwidth]{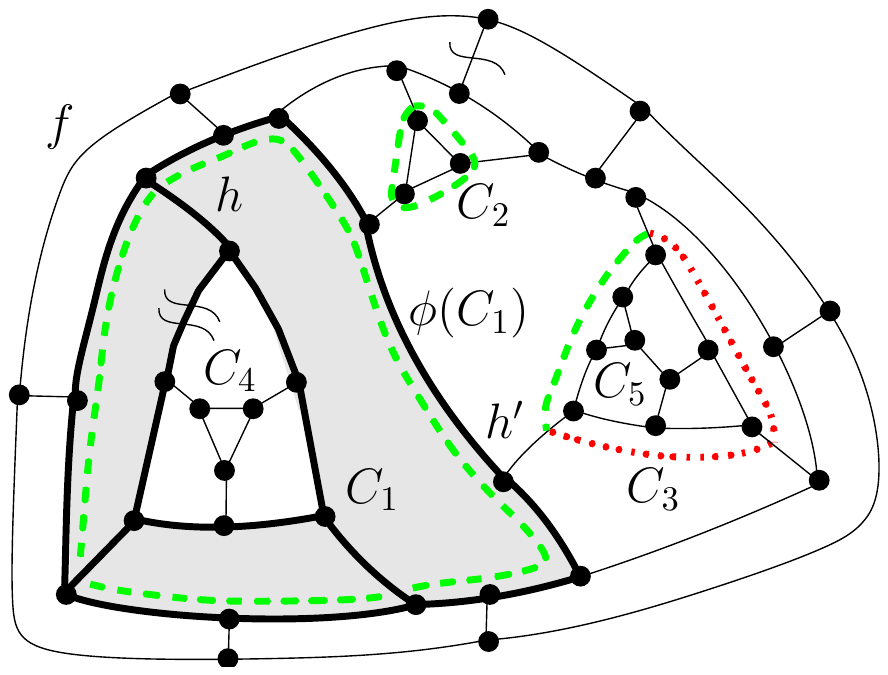}} 
	\hfil
	\\
	\hfil
	\subfloat[]{\label{fi:introvert-colouration-1-c}\includegraphics[width=0.47\columnwidth]{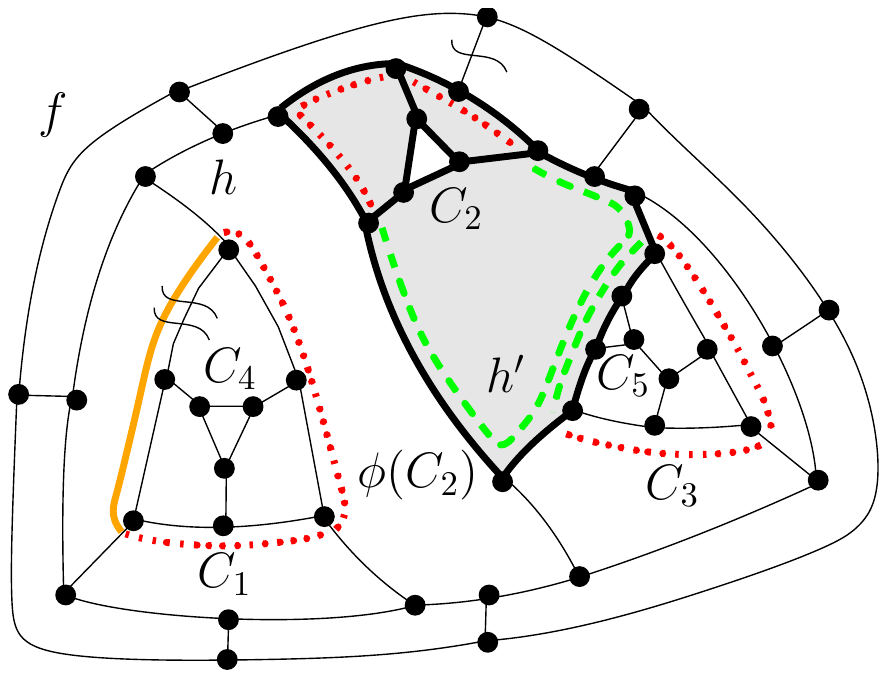}} 
	\hfil
	\subfloat[]{\label{fi:introvert-colouration-1-d}\includegraphics[width=0.47\columnwidth]{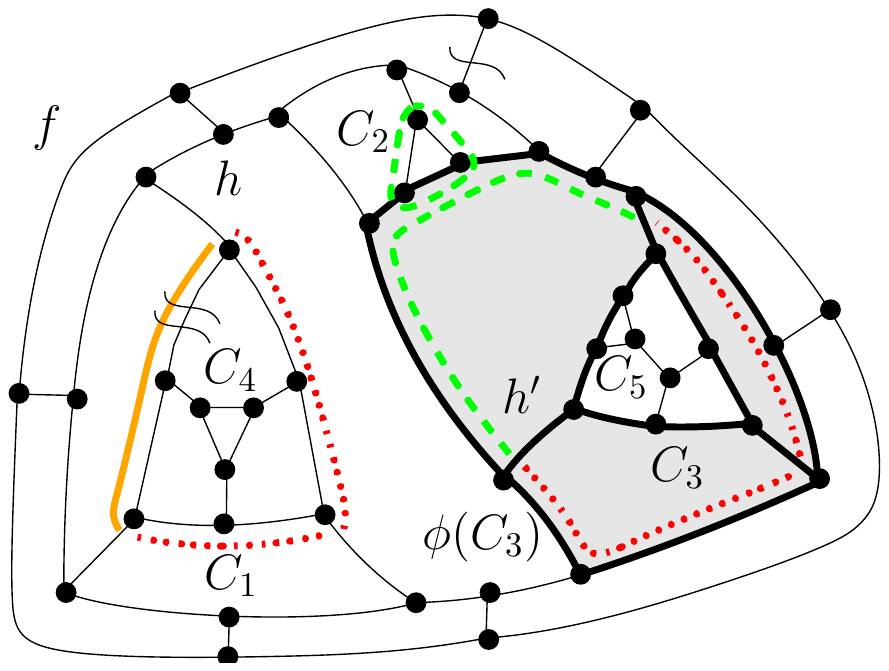}} 
	\hfil
	\\
	\hfil
	\subfloat[]{\label{fi:introvert-colouration-2-a}\includegraphics[width=0.47\columnwidth]{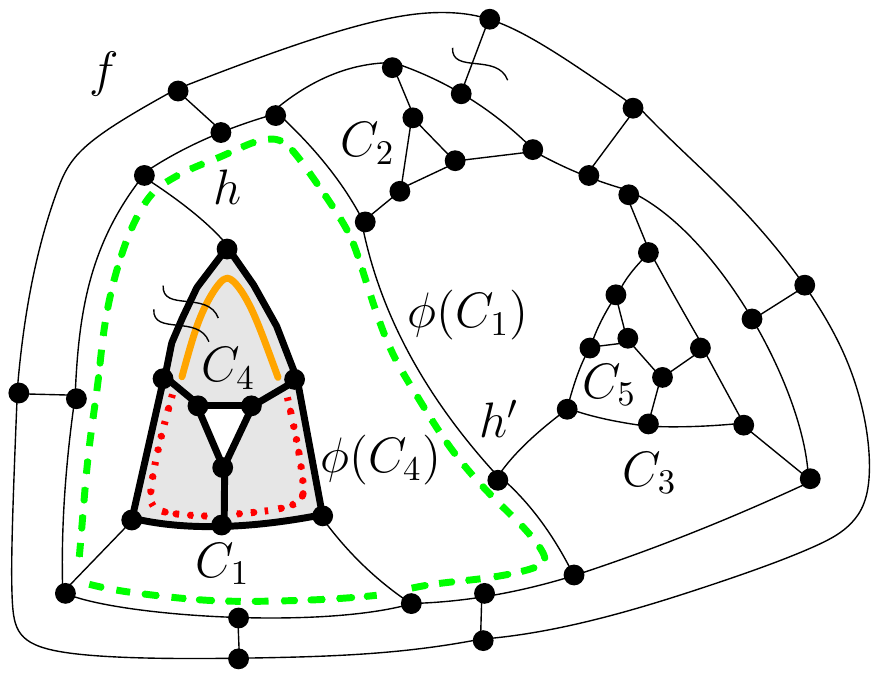}} 
	\hfil
	\subfloat[]{\label{fi:introvert-colouration-2-b}\includegraphics[width=0.47\columnwidth]{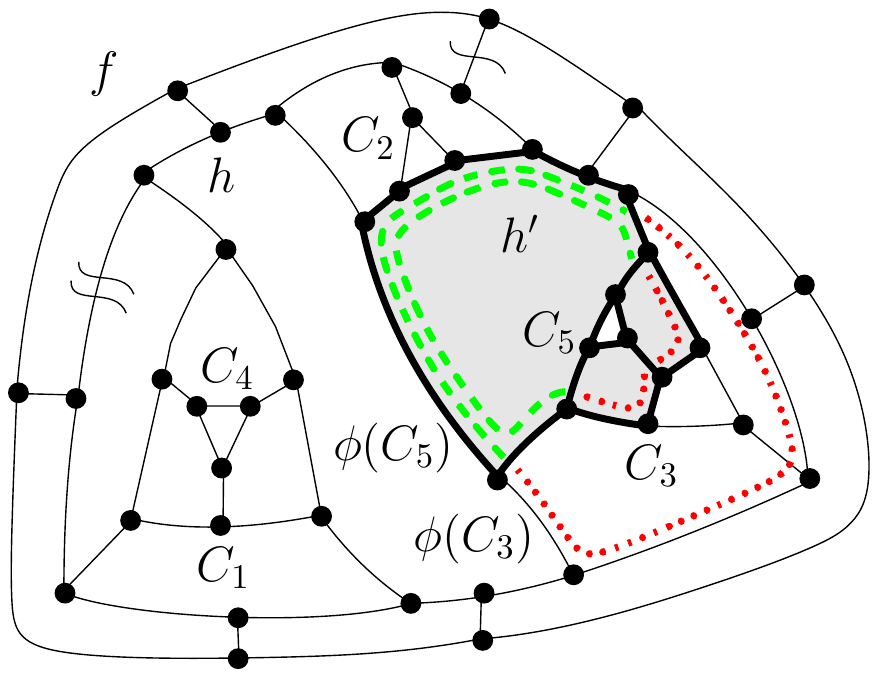}} 
	\hfil
	\caption{Red-green-orange coloring of the contour paths of 3-introvert cycles. (a) A 3-extrovert tree $T_f$ of a reference embedding $G_f$.
		(b--f) The red-green-orange coloring of the contour paths of $\phi(C_1)$, $\phi(C_2)$, $\phi(C_3)$, $\phi(C_4)$, $\phi(C_5)$, respectively. The contour paths of the 3-extrovert cycles $C_1$, $C_2$, and $C_3$ are also shown.
	}\label{fi:introvert-colouration-1}
\end{figure}

The following properties are used in the proof of several lemmas in this section.
\begin{property}\label{pr:intro-siblings}
	Let $\phi(C)$ be the 3-introvert cycle associated with a 3-extrovert cycle $C$ of $G_f$. Let $P$ be a contour path of $\phi(C)$ and let $f'$ be the leg face of $\phi(C)$ incident to $P$. $P$ contains a contour path of a sibling $C''$ of $C$ if and only if
	$C''$ has a contour path incident to $f'$.
\end{property}
\begin{property}\label{pr:spartizione-faccia}
	Let $\phi(C)$ be the 3-introvert cycle associated with a 3-extrovert cycle $C$ of $G_f$. Let $f'$ be a leg face of $\phi(C)$. The contour path $P$ of $\phi(C)$ incident to $f'$ contains all the edges of $f'$ not contained in $C$ and that are not legs of $C$ and $\phi(C)$.
\end{property}

\cref{le:introvert-coloring-algo,le:demanding-3-introvert-algo} show that we can efficiently color the 3-introvert cycles and compute the demanding 3-introvert cycles of $G_f$. The proof of \cref{le:demanding-3-introvert-algo} exploits the characterization of \cref{le:demanding-3-introvert-charact}, which in turn uses \cref{le:extrovert-extrovert-coloring} and \cref{le:extrovert-introvert-coloring}.

\begin{restatable}{lemma}{leIntrovertColoringAlgo}\label{le:introvert-coloring-algo}
	The red-green-orange coloring of the contour paths of the 3-introvert cycles of $G_f$ that satisfies the \textsc{3-Introvert Coloring Rule} can be computed in $O(n)$ time.
\end{restatable}

\begin{proof}
	Color the contour paths of the $3$-extrovert cycles of $G_f$ according to the \textsc{3-Extrovert Coloring Rule} by applying the algorithm in the proof of \cref{le:3-extrovert-coloring-linear}. Each contour path of every 3-extrovert or 3-introvert cycle of $G_f$ is equipped with a pointer to the corresponding leg face (this set of pointers can be constructed in $O(n)$ time through a visit of the graph). Also, every contour path $P$ of each 3-extrovert cycle is equipped with a counter $flex(P)$ that reports the number of flexible edges in $P$ (all values of $flex(P)$ can be computed in $O(n)$ time through a bottom-up visit of $T_f$). Every face $f'$ of $G_f$ has two counters, $flex(f')$ and $c_{green}(f')$. Counter $flex(f')$ reports the number of flexible edges in $f'$ (this set of counters can be constructed in $O(n)$ time through a visit of the graph) and $c_{green}(f')$ is initialized to zero.
	
	We perform a preorder visit of the 3-extrovert tree $T_f$ of $G_f$. Let $C_1, C_2, \dots, C_k$  be the child-cycles of a 3-extrovert cycle $C$ of $T_f$. We color the contour paths of $\phi(C_1), \phi(C_2), \dots, \phi(C_k)$ by performing three algorithmic steps. Note that, since $T_f$ is visited in preorder, we assume that the contour paths of $\phi(C)$ are already colored, unless $C=C_f$, in which case there is no 3-introvert cycle associated with the root of~$T_f$.
	
	In the first step, the coloring algorithm assigns values to counter $c_{green}(f')$; in the second step, it assigns the orange color to every contour path that has some flexible edges and it assigns the green color to some other paths.  This coloring is done based on the values of $flex(f')$ and $c_{green}(f')$. At the end of the second step, the color for some contour paths of $\phi(C_1), \phi(C_2), \dots, \phi(C_k)$ may remain undefined, i.e. it is not yet decided whether they will be red or green.  In the third step, the undecided contour paths are turned into either green or red paths. We now give a detailed description of the three steps and then discuss the time complexity.
	
	\begin{description}
		
		\item\textsf{First Step}: Let $S=\{C_1, C_2, \dots, C_k, \phi(C)\}$ be the set of the 3-extrovert cycles $C_1,...,C_k$ and $\phi(C)$ (if $\phi(C)$ is defined). For each contour path $P$ of a cycle $C'\in S$ incident to a leg face $f'$, we if $P$ is green, $c_{green}(f')$ is increased by one unit.
		
		\item\textsf{Second Step}:	Let $P$ be the contour path of $\phi(C_i)$ ($1 \leq i \leq k$), let $f'$ be the leg face of $\phi(C_i)$ incident to $P$. In the second step, we use the values of $flex(f')$, $flex(P)$, and of $c_{green}(f')$ to decide whether $P$ is colored orange, green, or it remains undecided as follows. Let $q$ be the number of legs of $C_i$ that are flexible  and incident to $f'$ ($0\le q\le 2$).
		\begin{itemize}
			\item If $flex(f') - flex(P)- q >0$ there always exists a flexible edge $e$ along $f'$ such that $e$ is neither a leg of $C_i$ nor an edge of the contour path of $C_i$ incident to $f'$. By \cref{pr:spartizione-faccia} $P$ includes a flexible edge and hence it is colored orange (see Case~2(a) of the \textsc{3-Introvert Coloring Rule}).
			\item If $flex(f') - flex(P)- q =0$, by \cref{pr:spartizione-faccia} no edges of $P$ are flexible. In this case $P$ is not orange; in this second step it will be colored either green or it will remain undecided depending on the value of $c_{green}$ (see below).
		\end{itemize}
		
		If $P$  is not colored orange, we decide whether $P$ is colored green or it remains undecided as follows.
		
		\begin{itemize}
			\item If $c_{green}>1$ there always exists a cycle $C'\in S$ that has a green contour path incident to $f'$ (see \cref{pr:intro-siblings}). In this case $P$ is colored green ((see Case~2(b) of the \textsc{3-Introvert Coloring Rule}).
			\item If $c_{green}=1$ and $C_i$ does not have a green contour path incident to $f'$, by \cref{pr:intro-siblings} $P$ includes a green contour path of a cycle $C'\in S$. In this case $P$ is colored green (see Case~2(b) of the \textsc{3-Introvert Coloring Rule}).
			\item If $c_{green}=0$ or if $c_{green}=1$ and $C_i$ has a green contour path incident to $f'$, $P$ will be either green or red depending on the colors of the other two contour paths of $\phi(C_i)$ (see Case~1 or Case~2(c) of the \textsc{3-Introvert Coloring Rule}). In this case the color of $P$ remains undecided.
		\end{itemize}
		
		\item\textsf{Third Step}:	The third step is executed after all contour paths of $\phi(C_i)$ ($1 \leq i \leq k$) have been processed by Step~1 and Step~2. This step considers the 3-introvert cycles having at least an undecided contour path. Let $\phi(C_i)$ be one such cycle. Two cases are possible.
		
		\begin{itemize}
			\item Every contour paths of $\phi(C_i)$ is undecided. This means that $\phi(C_i)$ does not have a flexible edge and it does not share any green edges with any cycle $C'\in S$. In this case the three contour paths of $\phi(C_i)$ are recolored green (see Case~1 of the \textsc{3-Introvert Coloring Rule}.
			\item If $\phi(C_i)$ has a some contour paths that have been colored either orange or green in the previous steps of the algorithm, all undecided contour paths of $\phi(C_i)$ are colored red (see Case~2(c) of the \textsc{3-Introvert Coloring Rule}).
		\end{itemize}
		At the end of the third step counters $c_{green}(f')$ are reset to $0$ for every leg face $f'$ of a cycle in~$S$.
	\end{description}
	
	The correctness of the above described coloring algorithm is easily established by observing that whenever a contour path is given the red, green, or orange color, this is done in accordance with the \textsc{3-Introvert Coloring Rule}. As for the time complexity, observe that the algorithm has $O(1)$-time complexity per contour path and that the number of contour paths is $O(n)$.
\end{proof}

We are now ready to use the coloring of the 3-introvert cycles to identify those that are demanding. We start with the following simple observation (a similar observation can be found in~\cite{DBLP:journals/ieicet/RahmanEN05}).

\begin{property}\label{pr:extrovert-introvert}
	Let $C$ be a 3-extrovert (3-introvert) cycle of $G_f$ and let $f'$ be any face of $G_f$. If $f'$ is a face of $G_f(C)$, then $C$ becomes a 3-introvert (3-extrovert) cycle in $G_{f'}$. If $f'$ is not a face of $G_f(C)$, then $C$ remains a 3-extrovert (3-introvert) cycle in $G_{f'}$.
\end{property}
Consider for example faces $h$ and $h'$ of graph $G_f$ in \cref{fi:introvert-colouration-1-b} and focus on the 3-introvert cycle $\phi(C_1)$. By choosing $h$ as a new external face we obtain the plane graph $G_h$ shown in \cref{fi:introvert-colouration-3-a}. Note that $\phi(C_1)$ has become a 3-extrovert cycle of $G_h$ because $h$ is a face of $G_f(\phi(C_1))$. By choosing $h'$ as new external face we obtain the plane graph $G_{h'}$ shown in \cref{fi:introvert-colouration-3-b}. In this case $\phi(C_1)$ is a 3-introvert cycle, since $h'$ is not a face of $G_f(\phi(C_1))$ and it is a face of $G_h(\phi(C_1))$.

The next two lemmas describe useful properties of the coloring of the contour paths in the variable embedding setting.

\begin{lemma}\label{le:extrovert-extrovert-coloring}
	Let $C$ be a 3-extrovert cycle of $G_f$ and let $f'$ be any face of $G_f$ such that $C$ is a 3-extrovert cycle also in $G_{f'}$. The coloring of any contour path of $C$ obtained by applying the \textsc{3-Extrovert Coloring Rule} to $C$ is the same in $G_f$ and in $G_{f'}$. Also, if $C$ is a demanding 3-extrovert cycle of $G_f$, $C$ is also a demanding 3-extrovert cycle of $G_{f'}$.
\end{lemma}
\begin{proof}
	By \cref{pr:extrovert-introvert} and by the fact that $C$ is 3-extrovert both in $G_f$ and in $G_{f'}$, we have that $f'$ is not in $G_f(C)$. Consider any 3-extrovert cycle $C'$ contained into $G_f(C)$. Since $f'$ is not in $G_f(C)$, by \cref{pr:extrovert-introvert} we have that $C'$ is also a 3-extrovert cycle in $G_{f'}$. It follows that the genealogical tree $T_{C}$ in $G_{f'}$ is the same as 
	the genealogical tree $T_{C'}$ in $G_f$.
	Hence, the application of the \textsc{3-Extrovert Coloring Rule} to $C'$ gives the same result is in $G_{f'}$ as in $G_{f}$. Finally, since also the coloring of the contour paths of 3-extrovert cycles in $T_{C}$ is not changed, if $C$ is demanding in $G_f$ it is also demanding in $G_{f'}$.
\end{proof}

\begin{restatable}{lemma}{leExtrovertIntrovertColoring}\label{le:extrovert-introvert-coloring}
	Let $\phi(C)$ be a 3-introvert cycle of $G_f$ that becomes a 3-extrovert cycle in $G_{f'}$.
	The coloring of each contour path of $\phi(C)$ resulting by the application of the \textsc{3-Introvert Coloring Rule} to $\phi(C)$ in $G_f$ is the same as the one resulting by the application of the \textsc{3-Extrovert Coloring Rule} to $\phi(C)$ in $G_{f'}$.
\end{restatable}
\begin{proof}
	Let $c(P_1)$, $c(P_2)$, and $c(P_3)$ be the red-green-orange coloring of the three contour paths $P_1$, $P_2$, and $P_3$ of $\phi(C)$ defined according to the \textsc{3-Introvert Coloring Rule} applied to $\phi(C)$ in $G_{f}$. We prove that the coloring of $P_1$, $P_2$, and $P_3$ defined according to the \textsc{3-Extrovert Coloring Rule} applied to $\phi(C)$ in $G_{f'}$ remains $c(P_1)$, $c(P_2)$, and $c(P_3)$, respectively.
	
	First, suppose that $c(P_i)$, for any $1 \leq i \leq 3$, is orange in $G_f$. In this case $P_i$ contains a flexible edge in $G_f$ according to Case~2(a) of the \textsc{3-Introvert Coloring Rule}. The contour path $P_i$ contains a flexible edge also in $G_{f'}$ and $c(P_i)$ is orange in $G_{f'}$ according to Case~2(a) of the \textsc{3-Extrovert Coloring Rule}.
	
	Second, suppose that $c(P_i)$, for any $1 \leq i \leq 3$, is red or green in $G_f$.
	The proof of the statement is by induction on the depth $h$ of $C$ in the 3-extrovert tree $T_f$ of $G_f$. The base case is when $h=1$, that is when $C$ is a child-cycle of $C_f$. Since $\phi(C)$ is a 3-extrovert cycle of $G_{f'}$, it has a genealogical tree $T_{\phi(C)}$ in $G_{f'}$. The child-cycles of $\phi(C)$ in $T_{\phi(C)}$ are exactly the siblings of $C$ in $T_f$. By \cref{pr:intro-siblings} a leg face of the 3-extrovert cycle $\phi(C)$ in $G_{f'}$ is incident to a contour path of one of its child-cycles in $T_{\phi(C)}$ if and only if a leg face of the 3-introvert cycle $\phi(C)$ in $G_f$ is incident to a contour path of one of the siblings of $C$ in $T_f$. By \cref{le:extrovert-extrovert-coloring} the coloring of the contour paths of the siblings of $C$ in $T_f$ is maintained when they become child-cycles of $\phi(C)$ in $T_{\phi(C)}$. Hence, the three contour paths $P_1$, $P_2$, and $P_3$ of the 3-extrovert cycle $\phi(C)$ of $G_{f'}$ have the coloring $c(P_1)$, $c(P_2)$, and $c(P_3)$, respectively.
	
	Suppose that the statement is true for all 3-extrovert cycles at depth $h < k$ of $T_f$. Let $C$ be a 3-extrovert cycle at depth $k$ in $T_f$ and let $C'$ be the parent of $C$ in $T_f$. Since each face contained into $G_f(\phi(C))$ is also contained into $G_f(\phi(C'))$,
	we have that $f'$ is also contained into $\phi(C')$.
	By \cref{pr:extrovert-introvert} $\phi(C)$ and $\phi(C')$ are 3-extrovert cycles of $G_{f'}$. Also, $\phi(C')$ and the siblings of $C$ in $T_f$ are exactly the child-cycles of $\phi(C)$ in $T_{\phi(C)}$. In particular, by induction we have that the coloring of the contour paths of the 3-introvert cycle $\phi(C')$ in $G_f$ is the same as the coloring of the 3-extrovert cycle $\phi(C')$ in $G_{f'}$. Analogously to the base case, by \cref{pr:intro-siblings} a leg face $f''$ of the 3-extrovert cycle $\phi(C)$ in $G_{f'}$ is incident to a contour path of one of its child-cycles in $T_{\phi(C)}$ if and only if the leg face $f''$ of the 3-introvert cycle $\phi(C)$ in $G_f$ is incident to a contour path of one of the siblings of $C$ in $T_f$ or to a contour path of $\phi(C')$. It follows that the three contour paths $P_1$, $P_2$, and $P_3$ of the 3-extrovert cycle $\phi(C)$ of $G_{f'}$ have the coloring $c(P_1)$, $c(P_2)$, and $c(P_3)$, respectively.
\end{proof}

The following lemma characterizes the demanding 3-introvert cycles.

\begin{restatable}{lemma}{leDemandingThreeIntrovertCharact}\label{le:demanding-3-introvert-charact}
	Let $C$ be a 3-extrovert cycle of $G_f$ and let $\phi(C)$ be the corresponding 3-introvert cycle. Cycle $\phi(C)$ is demanding if and only if the following two conditions hold: $(i)$ $\phi(C)$ has no contour path that contains either a flexible edge or a green contour path of a sibling of $C$ in $T_f$; $(ii)$ either the parent of $C$ is the root of $T_f$ or the parent of $C$ is a 3-extrovert cycle $C'$ and $\phi(C)$ has no contour path that contains a green contour path of $\phi(C')$.
\end{restatable}

\begin{proof}
	Suppose that Conditions~$(i)$ and~$(ii)$ of the statement are verified for cycle $\phi(C)$.
	We want to show that $\phi(C)$ is a demanding 3-extrovert cycle for every choice of the external face $f'$ such that $\phi(C)$ is a 3-extrovert cycle of $G_{f'}$. By \cref{le:extrovert-extrovert-coloring} it suffices to prove the existence of a face $f'$ such that $\phi(C)$ is a demanding 3-extrovert cycle of $G_{f'}$.
	We choose~$f'$ to be any leg face of $\phi(C)$ and prove that $\phi(C)$ is a demanding 3-extrovert cycle of~$G_{f'}$.
	Observe that, if $C$ is a child-cycle of a 3-extrovert cycle $C'$ in $T_f$, the faces of $G_f(\phi(C))$ are also faces of $G_f(\phi(C'))$.
	Hence, face $f'$ is also internal to $\phi(C')$.
	By \cref{pr:extrovert-introvert} we have that: $\phi(C)$ is a 3-extrovert cycle in $G_{f'}$; the siblings of $C$ in $T_f$ are still 3-extrovert cycles in $G_{f'}$; and $\phi(C')$, provided $C'$ exists, is a 3-extrovert cycle of $G_{f'}$.
	Also, the siblings of $C$ in $T_f$ and $\phi(C')$, if it is defined, are exactly the child-cycles of $\phi(C)$ in the genealogical tree $T_{\phi(C)}$.
	Condition~$(i)$ and \cref{le:extrovert-extrovert-coloring} imply that $\phi(C)$ does not have a leg face incident to a green contour path of one of its child-cycles in $T_{\phi(C)}$ that was a sibling of $C$ in $T_f$.
	By \cref{le:extrovert-introvert-coloring} we have that the coloring of the contour paths of the 3-introvert cycle $\phi(C')$ in $T_f$, provided that $C'$ exists, is exactly the coloring of the corresponding contour paths of the 3-extrovert cycle $\phi(C')$ in $G_{f'}$. Therefore, Condition~$(ii)$ implies that $\phi(C)$ does not have a contour path containing a green contour path of the child-cycle $\phi(C')$ in $T_{\phi(C)}$ that was the 3-introvert cycle associated to the parent $C'$ of $C$ in $T_f$. Hence, $\phi(C)$ is a demanding 3-extrovert cycle of $G_{f'}$.
	
	Now, suppose that $\phi(C)$ is a demanding 3-extrovert cycle of $G_{f'}$ for some face $f'$ of $G_f$. We prove that Conditions~$(i)$ and~$(ii)$ are verified. Observe that, the hypothesis that $\phi(C)$ is a demanding 3-extrovert cycle implies, by \cref{le:extrovert-introvert-coloring} and by the definition of demanding 3-extrovert cycle, that $\phi(C)$ does not have flexible edges.
	Consider the genealogical tree $T_{\phi(C)}$ rooted at $\phi(C)$ in $G_{f'}$. The child-cycles of $\phi(C)$ in $T_{\phi(C)}$ are exactly the siblings of $C$ in $T_f$ and $\phi(C')$, provided that $C$ has a 3-extrovert parent cycle $C'$ in the 3-extrovert tree $T_f$ associated with~$G_f$. Since $\phi(C)$ is a demanding 3-extrovert cycle of $G_{f'}$, it does not share edges with any green contour path of any of its child-cycles.
	It follows that Condition~$(i)$ is verified for $C$ and its siblings in $T_f$. Moreover, by \cref{le:extrovert-introvert-coloring}, Condition $(ii)$ is verified for $C$ and $\phi(C')$ in $T_f$.
\end{proof}

Refer to the 3-introvert cycle $\phi(C_1)$ of $G_f$ depicted in \cref{fi:introvert-colouration-1-a}. As already observed, by \cref{pr:extrovert-introvert} $\phi(C_1)$ is a 3-extrovert cycle of $G_{h}$. See \cref{fi:introvert-colouration-3-a}. According to \cref{le:demanding-3-introvert-charact} $\phi(C_1)$ is a demanding 3-introvert cycle in $G_f$ and, according to the definition of 3-extrovert cycle, it is a demanding 3-extrovert cycle of $G_{h}$.
Now, refer to the 3-introvert cycle $\phi(C_3)$ of $G_f$ depicted in \cref{fi:introvert-colouration-1-d}. By \cref{pr:extrovert-introvert} $\phi(C_3)$ is a 3-extrovert cycle of $G_{h'}$, as shown in \cref{fi:introvert-colouration-3-b}. According to \cref{le:demanding-3-introvert-charact} $\phi(C_3)$ is not a demanding 3-introvert cycle of $G_f$. In fact, it violates Condition~$(ii)$ of the lemma, since one of its contour paths contains a green contour path of $C_2$, which is a sibling of $C_3$. Notice that $\phi(C_3)$  is not a demanding 3-extrovert cycle of $G_{h'}$.

\begin{figure}[htb]
	\centering
	\subfloat[]{\label{fi:introvert-colouration-3-a}\includegraphics[width=0.49\columnwidth]{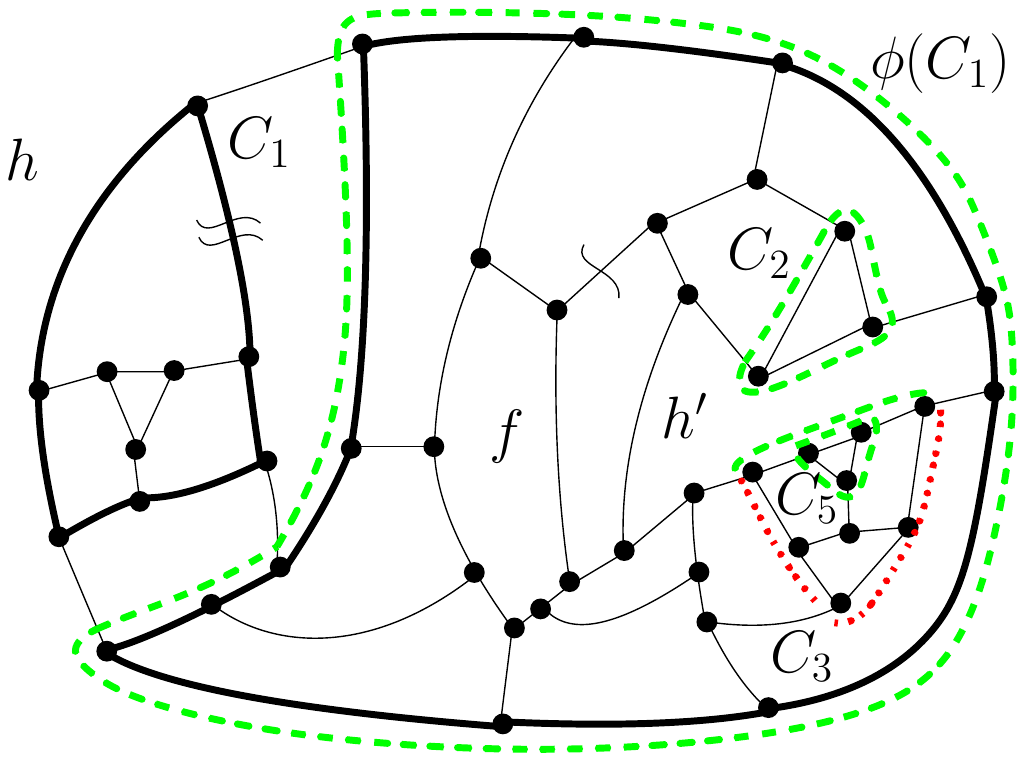}}
	\hfil
	\subfloat[]{\label{fi:introvert-colouration-3-b}\includegraphics[width=0.49\columnwidth]{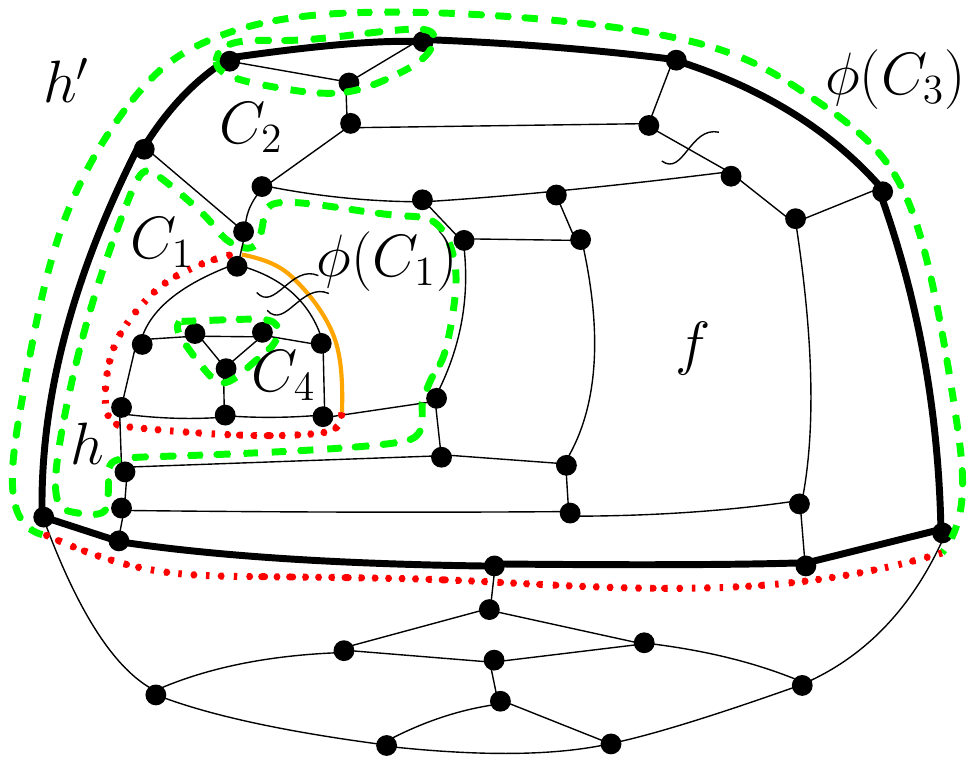}}
	\caption{(a)~The 3-introvert cycle $\phi(C_1)$ of the graph $G_f$ depicted in \cref{fi:introvert-colouration-1-a} is a demanding 3-introvert cycle of $G_f$ according to \cref{le:demanding-3-introvert-charact} and it is a demanding 3-extrovert cycle of $G_h$. (b)~The 3-introvert cycle $\phi(C_3)$ of the graph $G_f$ depicted in \cref{fi:introvert-colouration-2-a} is not a demanding 3-introvert cycle of $G_f$ according to \cref{le:demanding-3-introvert-charact} and it is not a demanding 3-extrovert cycle of $G_{h'}$.
	}\label{fi:introvert-colouration-3}
\end{figure}

\begin{restatable}{lemma}{leDemandingTreeIntrovertAlgo}\label{le:demanding-3-introvert-algo}
	Let $G_f$ be a reference embedding and $T_f$ be the corresponding 3-extrovert tree. The demanding 3-introvert cycles of $G_f$ can be computed in $O(n)$ time.
\end{restatable}

\begin{proof}
	We use the same approach described in the proof of \cref{le:introvert-coloring-algo}. Namely, we associate every face $f'$ with a counter that reports the number of the green contour paths incident to $f'$ and we perform a pre-order visit of $T_f$. The counters are initialized, incremented, and reset exactly as in \cref{le:introvert-coloring-algo}. Let $C$ be any 3-extrovert cycle of $G_f$ and let $\phi(C)$ be the corresponding 3-introvert cycle.
	Analogously to the proof of \cref{le:introvert-coloring-algo} we verify whether Condition (i) and Condition (ii) of \cref{le:demanding-3-introvert-charact} hold by checking the counters associated with the leg faces of $\phi(C)$.
\end{proof}

We conclude this section with two lemmas, namely \cref{le:fagiolo-nero-extrovert,le:fagiolo-bianco},
which will be used in \cref{sse:bend-counter} to construct the {\tt Bend-Counter}.

\begin{restatable}{lemma}{leFagioloNeroExtrovert}\label{le:fagiolo-nero-extrovert}
	Let $\phi(C_0), \phi(C_1), \dots, \phi(C_{k-1})$ be a set of demanding 3-introvert cycles of $G_f$ that share a leg face $f'$ such that $k>1$. We have that: $(a)$ Face $f'$ is the leg face of at most one demanding 3-extrovert cycle of $G_f$; $(b)$ any edge of $f'$ is contained in all the demanding (3-introvert or 3-extrovert) cycles having $f'$ as leg face with the possible exception of one of them.
\end{restatable}

\begin{proof}
	$(a)$ Let $C$ be a 3-extrovert cycle such that $f'$ is a leg face of $C$. Refer to \cref{fagiolonero2}. If there exists a $\phi(C_i)$ ($0 \leq i \leq k$))  such that $C_i$ is not a descendant of $C$ in $T_f$ (see, e.g., $\phi(C_0)$ in \cref{fi:fagiolonero-c}), we have that $\phi(C_i)$ and cycle $C$ share all edges of the
	contour path of $C$ incident to $f'$ and, by \cref{le:demanding-3-introvert-charact}, $C$ cannot be demanding. Otherwise, let $S$ denote the set of all 3-extrovert cycles of $T_f$ that have $f'$ as a leg face and such that $C_0, \ldots C_{k-1}$  are among their descendants (see, e.g., \cref{fi:fagiolonero-d}, where set $S$ consists of a cycle $C$ and $k=2$). The 3-extrovert cycles of $S$ all share edges of $f'$ and if one of them is demanding no other can be demanding by definition of demanding 3-extrovert cycle. It follows that there can be at most one demanding 3-extrovert cycle having $f'$ as a leg face in $G_f$.
	
	$(b)$ Note that the cycles $C_1,...,C_{k-1}$ are independent or else at least a cycle in $\phi(C_0), \phi(C_1), \dots, \phi(C_{k-1})$ would not satisfy the conditions of \cref{le:demanding-3-introvert-charact}. 
	Hence, by \cref{pr:spartizione-faccia}: If $e$ is an edge or a leg of a cycle $C_i$ then it is contained in all the demanding cycles having $f'$ as leg face except $\phi(C_i)$; if $e$ is an edge of $f'$ not contained in $C$ then it is contaided in all the cycles  $\phi(C_0),...,\phi(C_{k-1})$; else, it is contained in all the demanding cycles.
\end{proof}

\begin{figure}[htb]
	\centering
	\subfloat[]{\label{fi:fagiolonero-c}\includegraphics[width=0.48\columnwidth]{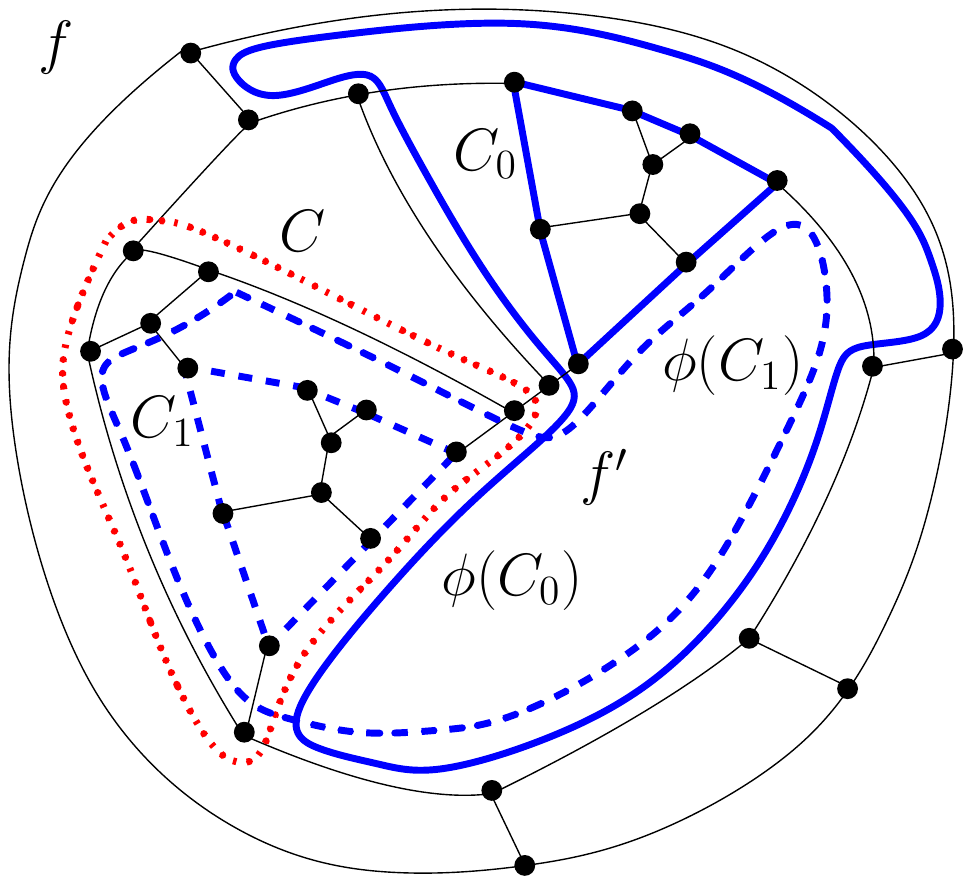}}
	\hfil
	\subfloat[]{\label{fi:fagiolonero-d}\includegraphics[width=0.48\columnwidth]{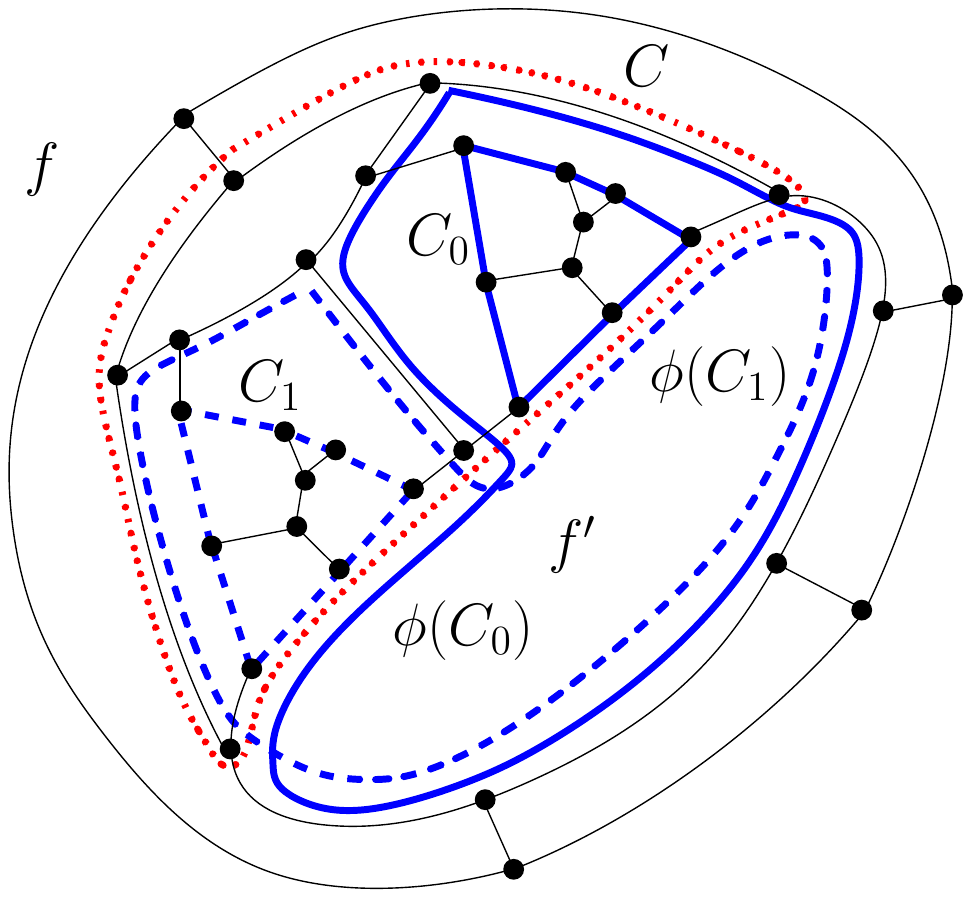}}
	\caption{Illustration for the proof of \cref{le:fagiolo-nero-extrovert}
		.}\label{fagiolonero2}
\end{figure}

\begin{restatable}{lemma}{leFagioloBianco}\label{le:fagiolo-bianco}
	Let $C_0, C_1, \dots C_{k-1}$ be a set of demanding 3-extrovert cycles of $G_f$ that share a leg face $f'$ such that $k>0$. We have that: $(a)$ Face $f'$ is the leg face of at most one demanding 3-introvert cycle of $G_f$; $(b)$~any two demanding cycles having $f'$ as leg face are edge disjoint.
\end{restatable}

\begin{proof}
	$(a)$ Let $C$ be any 3-extrovert cycle having $f'$ as its leg face ($C$ could be one of the demanding 3-extrovert cycles or some other, not demanding, 3-extrovert cycle). Refer to \cref{fagiolobianco}. Two cases are possible for $C$: Either the subtree of $T_f$ rooted at $C$ contains all cycles $C_0, C_2, \ldots C_{k-1}$ or there exists some $C_i$ ($0 \leq i \leq k-1$) such that $C_i$  is not a descendant of $C$ in $T_f$. In the latter case $\phi(C)$ has a contour path that shares edges with a green contour path of $C_i$; therefore $\phi(C)$ is not a demanding 3-introvert cycle because it violates Condition~(i) of \cref{le:demanding-3-introvert-charact}. See, e.g., \cref{fi:fagiolobianco-a}.
	In the former case, let $S$ denote the set of all 3-extrovert cycles of $T_f$ that have $f'$ as a leg face and such that $C_0, C_1, \ldots C_{k-1}$  are among their descendants. The 3-introvert cycles associated with the elements of $S$ all share edges of $f'$ and if one of them is demanding no others can be demanding since otherwise Condition~(ii) of \cref{le:demanding-3-introvert-charact} would be violated. It follows that there can be at most one demanding 3-introvert cycle having $f'$ as a leg face in $G_f$.
	
	$(b)$ If such a demanding 3-introvert exists, it is the 3-introvert cycle associated with an ancestor of $C_0, C_1, \ldots C_{k-1}$ and hence it cannot share vertices with anyone of them. See, e.g., \cref{fi:fagiolobianco-b}.
\end{proof}

\begin{figure}[htb]
	\centering
	\subfloat[]{\label{fi:fagiolobianco-a}\includegraphics[width=0.48\columnwidth]{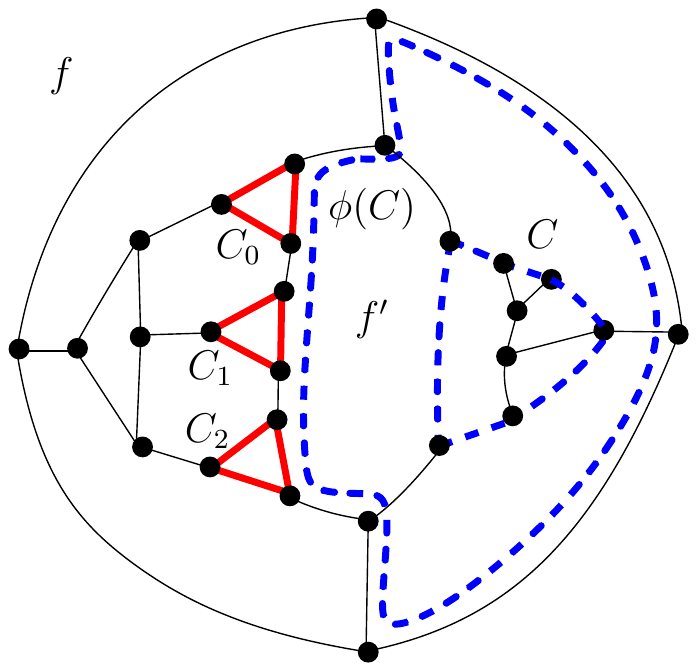}}
	\hfil
	\subfloat[]{\label{fi:fagiolobianco-b}\includegraphics[width=0.48\columnwidth]{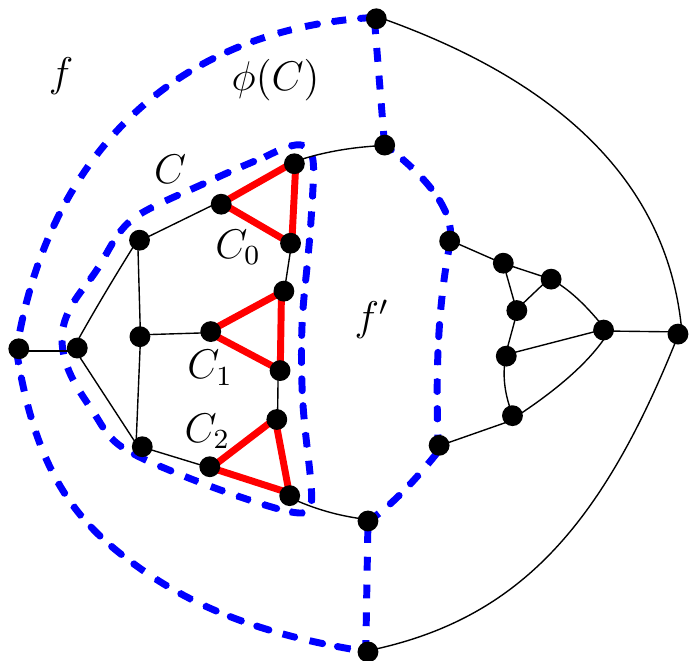}}
	\caption{Illustration for the proof of \cref{le:fagiolo-bianco}\label{fagiolobianco}}
\end{figure}

\subsubsection{Constructing the Bend-Counter}\label{sse:bend-counter}

As mentioned at the beginning of \cref{sse:variable-embedding}, a \texttt{Bend-Counter} of $G$ is constructed in $O(n)$ time starting from a reference embedding $G_f$ and consists of a 3-extrovert tree $T_f$, the number $|D(G_f)|$ of demanding 3-extrovert cycles of $T_f$, and of a face array~$A_f$. The elements of $T_f$ and of $A_f$ are enriched with information that allows us to compute in $O(1)$ time the cost
of a cost-minimum orthogonal representation of $G_{f'}$, for any external face $f'$, with at most one bend per edge.
We first describe the information that equips the nodes of $T_f$ and then describe the information that equips the elements of $A_f$.

Let $C_f$ be the root of the 3-extrovert tree $T_f$ of $G_f$.
Let $C$ be any non-root node of $T_f$ and let $\phi(C)$ be the 3-introvert cycle corresponding to $C$ (see \cref{le:3-extro-3-intro}). 
The \texttt{Bend-Counter} data structure equips node $C$ of $T_f$ with the following information:
\begin{itemize}
	\item A label $d(C) \in \{\texttt{true},\texttt{false}\}$ such that $d(C)=\texttt{true}$ (resp. $d(C)=\texttt{false}$) if $C$ is demanding (resp. not demanding) in $G_f$.
	\item A label $d(\phi(C)) \in \{\texttt{true},\texttt{false}\}$ such that $d(\phi(C))=\texttt{true}$ (resp. $d(\phi(C))=\texttt{false}$) if $\phi(C)$ is demanding (resp. not demanding) in $G_f$.
	\item The number $\extr(C)$ of demanding 3-extrovert cycles contained in the path connecting $C$ to the root $C_f$ in $T_f$.
	\item The number $\intr(C)$ of demanding 3-introvert cycles associated with 3-extrovert cycles contained in the path connecting $C$ to the root $C_f$ in $T_f$.
\end{itemize}

\begin{lemma}\label{le:addinformation-in-linear-time}
	The values $d(C)$, $d(\phi(C))$, $\extr(C)$, and $\intr(C)$ for all non-root nodes $C$ of $T_f$ can be computed in $O(n)$ time.
\end{lemma}
\begin{proof}
All 3-extrovert cycles of $G_f$ are represented by the nodes of $T_f$ and by using \cref{le:demanding-3-extrovert-lineartime} we can compute the demanding ones in $O(n)$ time. It follows that the set of labels $d(C)$ of the non-root nodes of $T_f$ can be computed in $O(n)$ time. Also, by \cref{le:demanding-3-introvert-algo}, the set of labels $d(\phi(C))$ can be computed in $O(n)$ time. From the labels $d(C)$ and $d(\phi(C)$, the indexes $\extr(C)$ and $\intr(C)$ can be easily computed in $O(n)$ time through a preorder
visit of $T_f$.      	
\end{proof}

For each node $C$ of $T_f$ (including the root), denote by $F_C$ the set of faces of $G_f$ that belong to $G_f(C)$ and that do not belong to $G_f(C')$ for any child-cycle $C'$ of $C$ in $T_f$. The sets $F_C$ (for all nodes $C$ of $T_f$) form a partition of the set of faces of $G_f$ distinct from $f$, i.e., each face $f' \neq f$ belongs to exactly one $F_C$.
We recall that the face array $A_f$ has one entry for every face of $G_f$.
Each entry $f'$ of $A_f$ is equipped with the following information.

\begin{itemize}
	\item A pointer $\tau(f')=C$ that maps each $f'$ to the node $C$ of $T_f$ such that $f' \in F_C$. For the external face $\tau(f)$ is null.

    \item The number $\delta_e(f')$ of demanding 3-extrovert cycles of $G_f$ having $f'$ as leg face.

    \item The number $\delta_i(f')$ of demanding 3-introvert cycles of $G_f$ having $f'$ as leg face.


    \item The number $m_{f'}$ of flexible edges incident to $f'$ and the sum $s_{f'}$ of their flexibilities. Also, if $m_{f'} = 1$ a pointer $p_1(f')$ to the unique flexible edge of $f'$. If $m_{f'} = 2$ the pointers $p_1(f')$ and $p_2(f')$ to the two flexible edges of $f'$.

\end{itemize}

%

\begin{restatable}{lemma}{leExtDemandingCounter}\label{le:ext-demanding-counter}
	 The values associated with each entry $f'$ of $A_f$ can be computed in overall $O(n)$ time.
\end{restatable}

\begin{proof}
	Concerning the computation of $\tau(f')$, we initialize $\tau(f')$ with a null value for all the faces $f'$ of $G_f$. Then, we recursively remove leaves from $T_f$. Let $C$ be the current non-root leaf of $T_f$ and let $C'$ the parent of $C$. We set $\tau(f')=C$ for all the faces that are inside $C$ for which $\tau(f')$ is null. We set $\tau(f'')=C'$ for the three leg faces of $C$. Then we collapse $C$ in $G_f$ into a degree-three vertex and we remove the leaf $C$ from $T_f$. When $C=C_f$, that is, $C$ is the only leaf of $T_f$ and it is also its root, then each face $f'''$ of $G_f$ for which $\tau(f''')$ is null is set $\tau(f''')=C$.
	
	Concerning the computation of the values $\delta_e(f')$ and $\delta_i(f')$, for each contour path of each 3-extrovert and 3-introvert cycle of $G_f$ we assume to have a pointer to the
	corresponding leg face (this set of pointers can be constructed with an $O(n)$ time preprocessing). We first apply the technique of \cref{le:addinformation-in-linear-time} to compute the values $d(C)$ and $d(\phi(C))$ for every node $C$ of $T_f$. We then initialize to zero the values $\delta_e(f')$ and $\delta_i(f')$ for every face $f'$. We visit $T_f$ and for each node $C$ and for each leg face $f'$ of $C$ we increment $\delta_e(f')$ by one if $d(C)=\texttt{true}$ and we increment $\delta_i(f')$ by one if $d(\phi(C))=\texttt{true}$. Since every 3-extrovert cycle represented in $T_f$ has three leg faces and since there are $O(n)$
	3-extrovert cycles in $G_f$, the values $\delta_e(f')$ and $\delta_i(f')$ can be computed in $O(n)$ time.
	
	
	Finally, we describe how to compute $m_{f'}$ and $s_{f'}$. For each face $f'$, we inizialize  the values of $m_{f'} = s_{f'} = 0$ and the pointes $fe_0(f')$ and $fe_1(f')$ to null.
	We list all edges of $G_f$ and if the current edge $e$ is flexible, we increment $m_{f''}$ and $m_{f'''}$ for the two faces $f''$ and $f'''$ incident to $e$. Also, we sum the flexibiliy $flex(e)$ to $s_{f''}$ and $s_{f'''}$.
	In addition, each time we set $m_{f'} = 1$ for some face $f'$, we also set $fe_0(f')$ to point to the current edge $e$; each time we set $m_{f'} = 2$, we also set $fe_1(f')=e$. If instead we set $m_{f'}$ to a value greater than $2$, we reset $fe_0(f')$ and $fe_1(f')$ to null.
	
	All the operations described above can be performed in $O(n)$ time.
\end{proof}

\begin{figure}[tb]
	\centering
	\subfloat[]{\label{fi:bend-counter-a}\includegraphics[width=0.48\columnwidth]{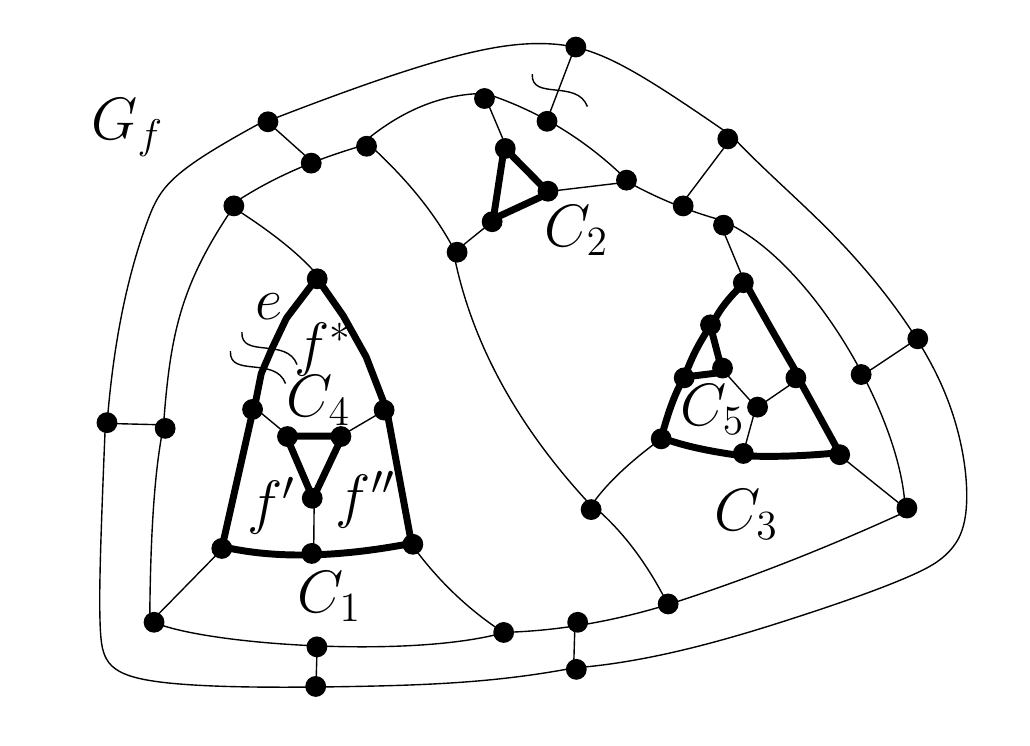}}
	\hfil
	\subfloat[]{\label{fi:bend-counter-c}\includegraphics[width=0.48\columnwidth]{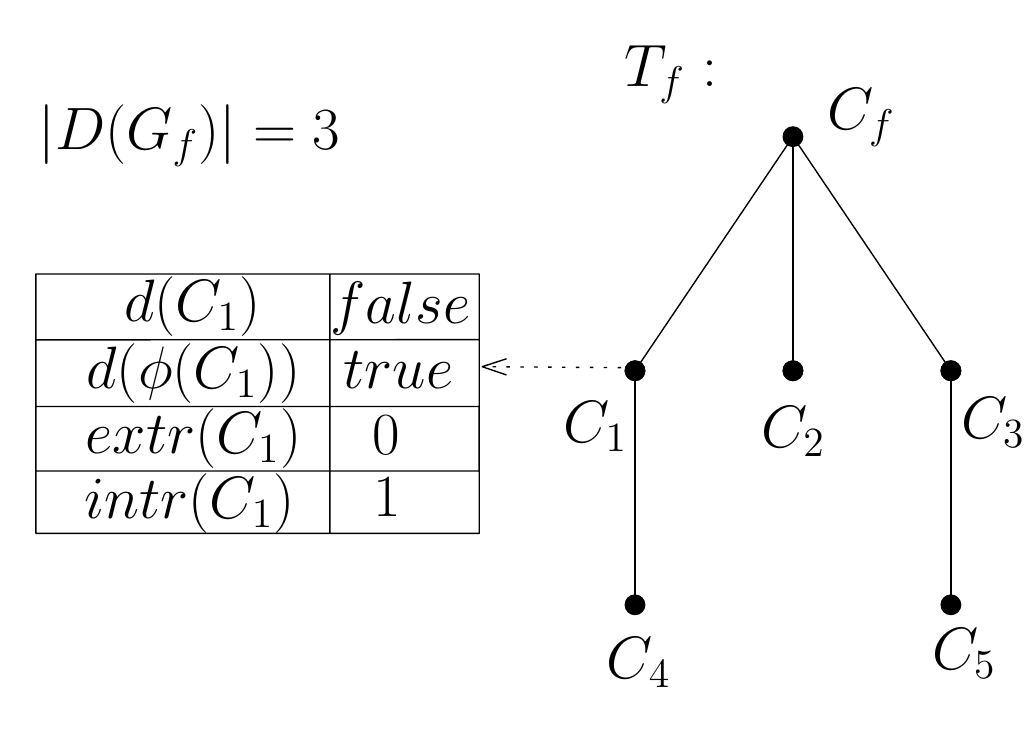}}
	\hfil
	\subfloat[]{\label{fi:bend-counter-d}\includegraphics[width=0.48\columnwidth]{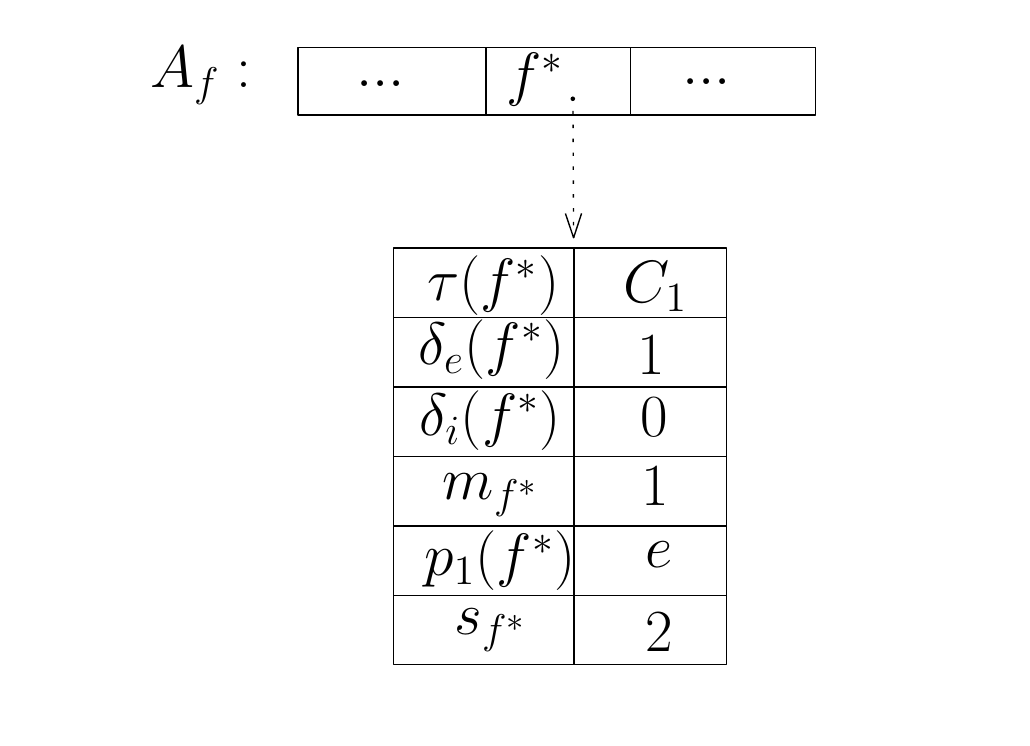}}
	\hfil
	\subfloat[]{\label{fi:bend-counter-b}\includegraphics[width=0.48\columnwidth]{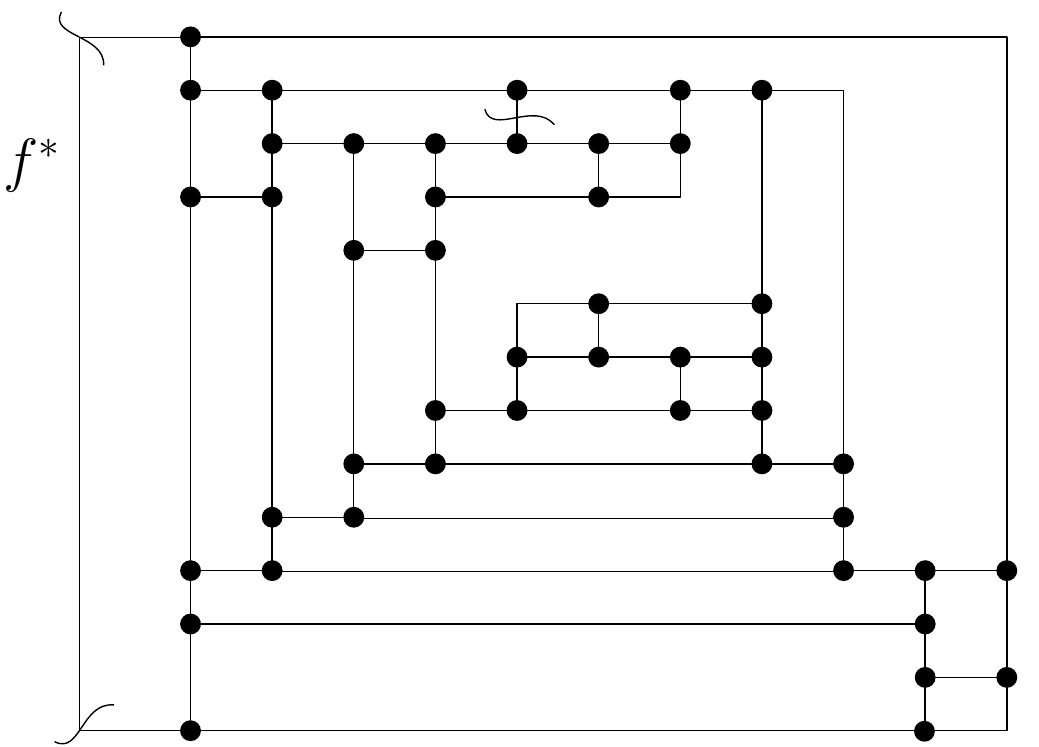}}
	\hfil
	\caption{(a) A plane graph $G_f$; (b) The 3-extrovert tree $T_f$ and the value of $|D(G_f)|$; (c) The face array $A_f$; (d) A minimum-bend orthogonal representation of $G_{f^*}$. Note that $\flex(f)=2$ and $|D_{f^*}(G_{f^*})|=1$. We have $|D(G_{f^*})|=|D(G_{f})|-\extr(C_1)+\intr(C_1)+|D_{f^*}(G_{f^*})|-\delta_e(f^*)=4$. Hence, $c(G_{f^*})=|D(G_{f^*})|+4-\min\{4,|D_{f^*}(G_{f^*})|+\flex(f^*)\}=5$.}\label{fi:bend-counter}
\end{figure}

\begin{lemma}\label{le:bend-counter_computation}
	The \texttt{Bend-Counter} of $G$ can be constructed in $O(n)$ time.
\end{lemma}
\begin{proof}
A reference embedding of $G_f$ of $G$ can be computed in $O(n)$ time by \cref{le:ref-embedding}. Also, by \cref{le:demanding-tree-lineartime} the 3-extrovert tree $T_f$ can be computed in $O(n)$ time.
For every non-root node of $T_f$, the statament is a consequence of \cref{le:demanding-tree-lineartime,le:addinformation-in-linear-time,le:ext-demanding-counter}.
%
Finally, we can compute $|D(G_f)|$ by performing a traversal of $T_f$ and by counting each node $C$ such that $d(C)=\texttt{true}$.
\end{proof}

\begin{lemma}\label{le:demanding-external-algo}
	Given the \texttt{Bend-Counter} of $G$, for any choice of the external face $f^*$ of $G$ the number $|D_{f^*}(G_{f^*})|$ of edge-disjoint demanding 3-extrovert cycles of $G_{f^*}$ incident to $f^*$ can be computed in $O(1)$ time.
\end{lemma}
\begin{proof}
	If $f^* = f$, we have that $|D_{f^*}(G_{f^*})|=0$ by the fact that $G_f$ is a reference embedding. Assume that $f^* \neq f$. Let $C^*=\tau(f^*)$ be the lowest node of $T_f$ such that $f^*$ belongs to $G_f(C^*)$. A demanding 3-extrovert cycle of $G_{f^*}$ that has $f^*$ as a leg face may be either a demanding 3-extrovert cycle $C'$ or a demanding 3-introvert cycle $\phi(C'')$ of $G_f$. Observe that $C'$ ($\phi(C'')$, resp.) has $f^*$ as a leg face also in $G_f$ and, since $f^*$ belongs to $G_f(C^*)$, $C'$ ($\phi(C'')$, resp.) also belongs $G_f(C^*)$.
	
	Consider a demanding 3-extrovert cycle $C'$ of $G_f$ that has $f^*$ as a leg face. We have that cycle $C'$ is a demanding 3-extrovert cycle also of $G_{f^*}$, because of \cref{pr:extrovert-introvert} and of the fact that $f^*$ is not a face of $G_f(C')$.
	Conversely, consider a demanding 3-introvert cycle $\phi(C'')$ of $G_f$ that has $f^*$ as a leg face. We have that $\phi(C'')$ always is a demanding 3-extrovert cycle of $G_{f^*}$, because of \cref{pr:extrovert-introvert} and of the fact that $f^*$ is a face of $G_f(\phi(C''))$. The number $|D_{f^*}(G_{f^*})|$ is the sum of the above cycles, that is, the sum of $\delta_e(f^*)$ and $\delta_i(f^*)$, minus those that are not edge disjoint. We distinguish between two cases.
	
	\begin{itemize}
		
		\item [(a)] $\delta_i(f^*) \le 1$. In this case, whatever is the value of $\delta_e(f^*) \geq 0$, \cref{le:fagiolo-bianco} guarantees that all such demanding cycles are edge disjoint. Hence, the number of non-intersecting demanding 3-extrovert cycles in $G_{f^*}$ incident to $f^*$ is $|D_{f^*}(G_{f^*})| = \delta_i(f^*) + \delta_e(f^*)$.
		
		\item [(b)] $\delta_i(f^*) > 1$. In this case, by \cref{le:fagiolo-nero-extrovert} we have that $\delta_e(f^*) \leq 1$. Also, by the same lemma we have that any two demanding cycles having $f^*$ as a leg face share at least one edge. It follows that $|D_{f^*}(G_{f^*})| = 0$.
	\end{itemize}
	
	Since values $\tau(f^*)$, $\delta_i(f^*)$, and $\delta_e(f^*)$, can be accessed in $O(1)$ time, the computation of $|D_{f^*}(G_{f^*})|$ can be performed in $O(1)$ time.
\end{proof}

\begin{lemma}\label{le:demanding-algo}
	Given the \texttt{Bend-Counter} of $G$, for any choice of the external face $f^*$ of $G$ the number $|D(G_{f^*})|$ of the non-intersecting demanding 3-extrovert cycles of $G_{f^*}$ can be computed in $O(1)$ time.
\end{lemma}
\begin{proof}
	If $f^* = f$ the value of $|D(G_{f})|$ is provided by the \texttt{Bend-Counter} of $G$. Otherwise, let $f^*$ be a face $f^* \neq f$. Our general strategy is that of computing $|D(G_{f^*})|$ by difference with respect to $|D(G_f)|$. In fact, when moving the external face from $f$ to $f^*$, some demanding 3-extrovert cycles appear and some others disappear.
	
	Let $C^*=\tau(f^*)$ be the lowest node of $T_f$ such that $f^*$ belongs to $G_f(C^*)$.
	By \cref{pr:extrovert-introvert}, we have that for every node $C$ in the path $\Pi_{C^*}$ from the root $C_f$ of $T_f$ to $C^*$, $C$ becomes a 3-introvert cycle of $G_{f^*}$ while $\phi(C)$ becomes a 3-extrovert cycle of $G_{f^*}$. Also, by definition of demanding 3-introvert cycle, if $\phi(C)$ is demanding in $G_f$ it becomes a demanding 3-extrovert cycle in $G_{f^*}$.
	
	Further, some other cycles of $G(C^*)$ may become edge disjoint 3-extrovert cycles incident to $f^*$ in $G_{f^*}$. These are exactly the $|D_{f^*}(G_{f^*})|$ independent demanding 3-extrovert cycles addressed in \cref{le:demanding-external-algo}. Observe that if $|D_{f^*}(G_{f^*})| > 0$, then it contains also the term $\delta_e(f^*)$, that is already accounted for in $|D(G_f)|$.
	
	In conclusion, in order to compute $|D(G_{f^*})|$, we have to: $(i)$ subtract from $|D(G_f)|$ the number $extr(C)$ of demanding 3-extrovert cycles in the path $\Pi_C$ that become 3-introvert in $G_{f'}$; $(ii)$ add the number $intr(C)$ of demanding 3-introvert cycles in the path $\Pi_C$ that become 3-extrovert in $G_{f^*}$; $(iii)$ add the number $|D_{f^*}(G_{f^*})|$ of independent demanding 3-extrovert cycles incident to $f^*$ in $G_{f^*}$, and $(iv)$ if $|D_{f^*}(G_{f^*})|>0$ subtract $\delta_e(f^*)$.
	
	Since $\intr(f^*)$, $\extr(f^*)$, $\delta_e(f^*)$ can be accessed in $O(1)$ time and $|D_{f^*}(G_{f^*})|$ by \cref{le:demanding-external-algo} can be computed in $O(1)$ time, we have that also $|D(G_{f^*})|$ can be computed in $O(1)$ time.
\end{proof}

\begin{lemma}\label{le:flex-f-algo}
	Given the \texttt{Bend-Counter} of $G$, for any face $f^*$ of $G$ the value of $\flex(f^*)$ can be computed in $O(1)$ time.
\end{lemma}
\begin{proof}
	The face array $A_f$ of the \texttt{Bend-Counter} provides the values of $m_{f^*}$, $s_{f^*}$, and the pointers to the unique flexible edge (pointer $p_1(f^*)$)  or to the unique pair of flexible edges (pointers $p_1(f^*)$ and $p_2(f^*)$) of $f^*$. The only missing values to compute $flex(f^*)$ are the following: (i) if $m_{f^*}=1$ and $p_1(f^*)=e$, denote by $f' \neq f^*$ the other face incident to $e$; we need the sum $x_e$ of the flexibilities of the edges distinct from $e$ of the face $f'$ plus the number of demanding 3-extrovert cycles incident to $f'$; (ii) if $m_{f^*}=2$, $p_1(f^*)=e_0$, and $p_2(f^*)=e_1$, denote by $f'$ and $f''$ the faces different from $f^*$ incident to $e_0$ and $e_1$, respectively; we need the analogous sums $x_{e_0}$ and $x_{e_1}$.
	
	In order to compute such values, first observe that the sum of the flexibilities of the edges incident to a face $f'$ distinct from a specific edge $e$ is $s_{f'} - \flex(e)$. Second, observe that the number of demanding 3-extrovert cycles incident to the face $f'$ when $f'$ is the external face is also the number of demanding 3-extrovert cycles incident to the face $f'$ when the external face is a face $f^*$ adjacent to $f'$ through a flexible edge $e$. In fact, when we move the external face from $f'$ to $f^*$, since no demanding 3-extrovert (3-introvert) cycle contains $e$ (recall that $e$ is flexible), we have that the sets of demanding 3-extrovert (3-introvert) cycles containing the faces $f'$ and $f^*$ are not changed and, therefore, by \cref{pr:extrovert-introvert}, the number of demanding 3-extrovert cycles incident to $f'$ is the same.
	The $O(1)$ time complexity of the statement descend from the fact that all the values needed to compute $\flex(f^*)$ are either provided by the face array $A_f$ (this is the case of $m_{f^*}$, $s_{f^*}$, $p_1(f^*)$, $p_2(f^*)$, $s_{f'}$, and $s_{f''}$) or they can be computed in $O(1)$ time (this is the case of $|D_{f'}(G_{f'})|$ and $|D_{f''}(G_{f''})|$ that can be computed as described in the proof of \cref{le:demanding-external-algo}).
\end{proof}

\begin{lemma}\label{le:min-bend-constant-time}
	Given the \texttt{Bend-Counter} of $G$, for every choice of the external face $f^*$, the cost $c(G_{f^*})$ of a cost-minimum orthogonal representation of $G_{f^*}$ can be computed in $O(1)$ time.
\end{lemma}
\begin{proof}
By Theorem~\ref{th:fixed-embedding-min-bend} we have that $c(G_{f^*}) = |D(G_{f^*})| + 4 - \min\{4, |D_{f^*}(G_{f^*})| + \flex(f^*) \}$. The values of $|D(G_{f^*})|$, $|D_{f^*}(G_{f^*})|$, and $\flex(f^*)$ can be computed in $O(1)$ time by \cref{le:demanding-external-algo},  \cref{le:demanding-algo}, and \cref{le:flex-f-algo} respectively.
%
\end{proof}



Observe that \cref{le:bend-counter_computation,le:demanding-external-algo} imply Theorem~\ref{th:bend-counter}.
\cref{fi:bend-counter} schematically illustrates the \texttt{Bend-Counter} and its use. \cref{fi:bend-counter-a} shows the plane graph $G_f$.
The demanding 3-extrovert cycles of $G_f$ are $C_2$, $C_4$, and $C_5$ and the only demanding 3-introvert cycle of $G_f$ is $\phi(C_1)$ (see also \cref{fi:introvert-colouration-1}).

We conclude with a technical lemma that will be of use in the next section.

\begin{restatable}{lemma}{leBendCounterUpdate}\label{le:bend-counter-update}
	Let $e$ be a flexible edge of $G$. If $\flex(e)$ is changed to any value in $\{1,2,3,4\}$, the \texttt{Bend-Counter} of $G$ can be updated in $O(1)$ time.
\end{restatable}

\begin{proof}
	Since $e$ is still flexible, changing its flexibility does not affect the red-green-orange coloring of the contour paths of the 3-extrovert and 3-introvert cycles of $G$. Hence, the set of demanding 3-extrovert (3-introvert) cycles are not modified and the tree $T_f$ of the \texttt{Bend-Counter} is not changed. Let $f'$ and $f''$ be the two faces incident to $e$. The only values that are affected by the change of flexibility of $e$ are the sums $s_{f'}$ and $s_{f''}$ of the flexibilities of the edges incident to $f'$ and $f''$ respectively, that can be updated in $O(1)$ time.
\end{proof}


\section{The Labeling Algorithm: Proofs of Theorems~\ref{th:key-result-2} and~\ref{th:1-connected-labeling}}\label{se:labeling}


We give details about the second ingredient of the proof of Theorem~\ref{th:main}, i.e., the labeling algorithm. 
\cref{sse:labeling-biconnected} proves Theorem~\ref{th:key-result-2} for biconnected graphs; it uses the procedures of \cref{sse:shape-cost-qps} to efficiently compute the representative shapes for the different nodes of the SPQR-tree of the graph. \cref{sse:labeling-1-connected} proves Theorem~\ref{th:1-connected-labeling} for 1-connected graphs.     

\subsection{The Shape-Cost Sets.}\label{sse:shape-cost-qps}


We assume first that $G$ is a biconnected graph distinct from $K_4$. We arbitrarily choose an initial edge $e^*$ of $G$, root $T$ at the Q-node $\rho$ corresponding to $e^*$, and perform a pre-processing bottom-up visit of $T$. For each visited node $\mu \neq \rho$ of $T$ we compute and store at $\mu$ the {\em shape-cost set of $\mu$}, that is the set $\{b_{e^*}^{\sigma_1}(\mu), b_{e^*}^{\sigma_2}(\mu), \ldots, b_{e^*}^{\sigma_h}(\mu)\}$, where $\sigma_i$ ($1 \leq i \leq h$) is one of the representative shapes defined by Theorem~\ref{th:shapes} and $b_{e^*}^{\sigma_i}(\mu)$ is the number of bends of an orthogonal representation $H_{\mu}$ of $G_\mu$ such that: $(i)$ every edge of $H_{\mu}$ has at most one bend; $(ii)$ $H_{\mu}$ has shape $\sigma_i$; $(iii)$ $H_{\mu}$ is bend-minimum among all orthogonal representations of $G_{\mu}$ that satisfy~$(i)$--$(ii)$. 
The shape-cost set of each node $\mu$ is efficiently computed by accessing the shape-cost sets of its children. 
In some cases, if we know that a specific representative shape is never used in a bend-minimum representation of the graph, we conventionally set its cost to infinity.

After this pre-processing visit, the labeling algorithm considers all possible ways of re-rooting $T$ at an edge $e \neq e^*$. At each re-rooting we have that for some nodes of $T$ the set of its children remains unchanged, while for some other nodes the former parent node becomes a new child node and a former child node becomes the new parent node. The shape-cost sets of those nodes of $T$ whose children have changed must be updated. The next lemmas describe how to efficiently compute the shape-cost set of every node $\mu$ of $T$ in the pre-processing visit and how to efficiently update it when re-rooting $T$. 
For each possible edge $e$ of $G$, \cref{le:b(e)} describes how to compute $b(e)$ at the end of the visit of $T$ rooted at the Q-node corresponding to $e$.  In all the statements, $n_\mu$ denotes the number of children of~$\mu$. 

\begin{restatable}{lemma}{leShapeCostSetInnerQ}\label{le:shape-cost-set-inner-Q}
	If $\mu$ is an leaf Q-node, the shape-cost set of $\mu$ is computed in $O(1)$ time in the pre-processing visit and updated in $O(1)$ time when $T$ is re-rooted at any other edge.
\end{restatable}
\begin{proof}
	The procedure to compute and update the shape-cost set of $\mu$ is the same. Let $e$ the edge associated with the root of $T$ (possibly $e=e^*$). 
	Since we look for a representation with at most one bend per edge, the shape-cost set of a Q-node has two entries:
	$b_{e}^{\zerob}(\mu)=0$ and $b_{e}^{\oneb}(\mu)=1$. 
\end{proof}

\begin{restatable}{lemma}{leShapeCostSetP}\label{le:shape-cost-set-P}
	If $\mu$ is a P-node, the shape-cost set of $\mu$ is computed in $O(1)$ time in the pre-processing visit and updated in $O(1)$ time when $T$ is re-rooted at any other edge.
\end{restatable}
\begin{proof}
	As for the Q-nodes, the procedure to compute and update the shape-cost set of $\mu$ is the same. Let $e$ the edge associated with the current root of $T$ (possibly $e=e^*$).
	If $\mu$ is a P-node, $\mu$ has two children $\nu_1$ and $\nu_2$ in $T$, each being either a Q-node or an S-node (see of \cref{le:spqr-tree-3-graph}). For simplicity, we say here that the \zeroB-shaped (resp. \oneB-shaped) representation of the edge associated with a Q-node is 0-spiral (resp. 1-spiral).
	Suppose first that $\mu$ is an inner node. A bend-minimum \D-shape representation for $G_\mu$ is obtained by combining a 0-spiral representation stored at $\nu_1$ with a 2-spiral representation stored at $\nu_2$, or vice versa.
	Hence, $b_{e}^{\d}(\mu)=\min\{b_{e}^0(\nu_1) + b_{e}^2(\nu_2), b_{e}^2(\nu_1) + b_{e}^0(\nu_2) \}$. Similarly,
	a bend-minimum \X-shape representation for $G_\mu$ is obtained by combining a 1-spiral representation stored at $\nu_1$ with a 1-spiral representation stored at $\nu_2$. Hence $b_{e}^{\x}(\mu)= b_{e}^1(\nu_1) + b_{e}^1(\nu_2)$. It follows that both $b_{e}^{\d}(\mu)$ and $b_{e}^{\x}(\mu)$ are computed in $O(1)$ time.
	
	Suppose now that $\mu$ is the root child. In this case, due to Property~\textsf{O2} of Theorem~\ref{th:shapes}, we have to compute $b_{e}^{\c}(\mu)$ and $b_{e}^{\l}(\mu)$. The value $b_{e}^{\c}(\mu)$ is obtained by merging a 4-spiral representation of $G_{\nu_1}$ with a 2-spiral representation of $G_{\nu_2}$, or vice versa. The value $b_{e}^{\l}(\mu)$ is obtained by merging a 3-spiral representation of $G_{\nu_1}$ with a 1-spiral representation of $G_{\nu_2}$, or vice versa. Hence we have: $b_{e}^{\c}(\mu)=\min\{b_{e}^4(\nu_1) + b_{e}^2(\nu_2), b_{e}^2(\nu_1) + b_{e}^4(\nu_2) \}$ and $b_{e}^{\l}(\mu)=\min\{b_{e}^3(\nu_1) + b_{e}^1(\nu_2), b_{e}^1(\nu_1) + b_{e}^3(\nu_2) \}$. Therefore, both $b_{e}^{\c}(\mu)$ and $b_{e}^{\l}(\mu)$ are computed in $O(1)$ time.  
\end{proof}

\begin{restatable}{lemma}{leShapeCostSetS}\label{le:shape-cost-set-S}
	If $\mu$ is an S-node, the shape-cost set of $\mu$ is computed in $O(n_\mu)$ time in the pre-processing visit and updated in $O(1)$ time when $T$ is re-rooted at any other edge.
\end{restatable}
\begin{proof}
	Denote by $e'=(u,v)$ the reference edge of $\mu$. 
	If $\mu$ is an S-node, $\skel(\mu)$ is a simple cycle. More precisely, if $\mu$ is an inner node then $e'$ is a virtual edge, else $e'$ corresponds to the real edge associated with the root of $T$. 
	Each virtual edge of $\skel(\mu) \setminus e'$ corresponds to a child of $\mu$ in $T$, which is either a P-node or an R-node. 
	
	Consider first the pre-processing visit. Denote by $n^{\d}_{e^*}$ the number of virtual edges of $\skel(\mu) \setminus e'$ for which the cheapest representation of the corresponding component is $\D$-shaped, and denote by $n^{\zerob}_{e^*}$ the number of real egdes of $\skel(\mu) \setminus e'$ (i.e., the number of Q-children of $\mu$). Also, let $n^a_{e^*}$ denote the number of alias vertices associated with the poles of $\mu$, i.e., $n^a_{e^*}=0$ if both $u$ and $v$ have degree one in $G_\mu$ (which is always the case when $\mu$ is an inner node), $n^a_{e^*}=1$ if exactly one of $u$ and $v$ has degree two in $G_\mu$, and $n^a_{e^*}=2$ if both $u$ and $v$ have degree two in $G_\mu$.   
	In~\cite{DBLP:conf/gd/DidimoLP18} it is proved that for each possible value $k \in \{0, 1, 2, 3, 4 \}$ of spirality, a bend-minimum representation of $H_\mu$ with spirality $k$ can be obtained by selecting the cheapest representation for each P- and R-child of $\mu$ (in case of ties the $\D$-shaped representation is chosen) and by adding a number of extra bends equals to $B_{e^*}(k)_{extra} = \max\{0, k - n^{\d}_{e^*} - n^{\zerob}_{e^*} - n^a_{e^*} + 1\}$. Note that, for $k \in \{0,1\}$, $B_{e^*}(k)_{extra}=0$, because $n^{\zerob}_{e^*} \geq 1$. Therefore, $b_{e^*}^k(\mu)$ is set equal to the sum of the costs of the cheapest representations for the P- and R-children of $\mu$ plus $B_{e^*}(k)_{extra}$, which implies that $b_{e^*}^k$ can be computed in $O(n_\mu)$ time. Also, in the pre-processing visit, we store at $\mu$ the values $b_{e^*}^k(\mu)$, $n^{\d}_{e^*}$, $n^{\zerob}_{e^*}$, and $n^a_{e^*}$, which will be used at each re-rooting of $T$.
	
	Note however that if $n^{\zerob}_{e^*} < B_{e^*}(k)_{extra}$ it is not possible to arbitrarily insert each of the extra bends on a distinct real edge of $\skel(\mu) \setminus e'$ in $H_\mu$, so to guarantee at most one bend per edge; in this case, the value of $b_{e^*}^k(\mu)$ that is stored at $\mu$ is set as described above, but the value shown to the parent of $\mu$ during this specific visit is set to $\infty$, which rules out using spirality $k$ for an orthogonal representation of $G_\mu$.
	
	Suppose now that $T$ is re-rooted at an edge $e \neq e^*$. Denote by $\pi_{e^*}$ (resp. $\pi_{e}$) the parent of $\mu$ in $T$ rooted at $e^*$ (resp. at $e$). If $\pi_{e^*}=\pi_{e}$, the shape-cost set is unchanged, i.e., $b_{e}^{k}(\mu)= b_{e^*}^{k}(\mu)$ for $k=0,\dots,4$. Otherwise, we exploit the values $n^{\d}_{e^*}$, $n^{\zerob}_{e^*}$, and $n^a_{e^*}$ stored at $\mu$. Observe that $\pi_{e^*}$ is a child of $\mu$ in $T$ rooted at $e$, while $\pi_{e}$ is a child of $\mu$ in $T$ rooted at $e^*$. Since we know whether $\pi_{e^*}$ and $\pi_{e}$ are Q-, P-, or R-nodes, we can easily compute from $n^{\d}_{e^*}$ and $n^{\zerob}_{e^*}$ the number $n^{\d}_{e}$ (the children of $\mu$ in $T$ rooted at $e$ that are not Q-nodes and for which the cheapest representation of the corresponding component is $\D$-shaped) and the number $n^{\zerob}_{e}$ (the Q-node children of $\mu$ in $T$ rooted at $e$). Also, we can compute in constant time the value $n^a_{e}$. 
	
	For each $k \in \{0, 1, 2, 3, 4 \}$, let $B_{e}(k)_{extra} = \max\{0, k - n^{\d}_{e} - n^{\zerob}_{e} - n^a_{e} + 1\}$. If $n^{\zerob}_{e} < B_{e}(k)_{extra}$ then $b_{e}^{k}(\mu) = \infty$ because a $k$-spiral representation of $G_\mu$ would require two bends on some edge. Otherwise, $b_{e}^{k}(\mu)$ can be obtained by subtracting from $b_{e^*}^{k}(\mu)$ the contribution of $\pi_{e}$, summing up the contribution of $\pi_{e^*}$, subtracting the extra bends $B_{e^*}(k)_{extra}$, and summing up the extra bends $B_e(k)_{extra}$. Concerning the time complexity, the values stored at $\mu$ can be accessed in $O(1)$ time. Also, although $\mu$ has $n_\mu$ children, only the shape-cost set of the child $\pi_{e^*}$ of $\mu$ is accessed in the procedure described above. Therefore, the update of the shape-cost set of $\mu$ is performed in $O(1)$ time.
\end{proof}

We now consider the case of an R-node $\mu$. Observe that if $\mu$ is an inner R-node, its parent node is necessarily an S-node $\mu'$. A consequence of the proof of \cref{le:shape-cost-set-S} is that if the number of bends of a bend-minimum $\D$-shape representation of $G_\mu$ is not greater than the number of bends of a bend-minimum $\X$-shape representation of $G_\mu$, then the $\D$-shape representation can always be preferred to the $\X$-shape representation to construct a bend-minimum representation $H_{\mu'}$ of $G_{\mu'}$, for any desired value of spirality of $H_{\mu'}$. Hence, when we compute the shape-cost set of an inner R-node (for which either a $\D$-shape or an $\X$-shape has to be considered, due to Property~\textsf{O2} of Theorem~\ref{th:shapes}), we will set to infinity the cost of the $\X$-shape representation of $G_\mu$ if this cost is not strictly less than the cost of the $\D$-shape representation. We prove the following.
	
\begin{restatable}{lemma}{leShapeCostSetR}\label{le:shape-cost-set-inner-R}
	If $\mu$ is an inner R-node, the shape-cost set of $\mu$ is computed in $O(n_\mu)$ time in the pre-processing visit and updated in $O(1)$ time when $T$ is re-rooted at any other edge.
\end{restatable}	
\begin{proof}
	Consider first the pre-processing visit (i.e., when the root $\rho$ of $T$ corresponds to edge $e^*$). Denote by $e'$ the reference edge of $\mu$. We distinguish between two cases: 
	In this case, we have to compute $b_{e^*}^{\d}(\mu)$ and $b_{e^*}^{\x}(\mu)$; as discussed above, if $b_{e^*}^{\d}(\mu) \leq b_{e^*}^{\x}(\mu)$, we can set $b_{e^*}^{\x}(\mu)=\infty$. 
	Let $f'$ and $f''$ be the two faces incident to $e'$ in $\skel(\mu)$. To compute $b_{e^*}^{\d}(\mu)$, we follow this strategy. $(i)$ For each face $f \in \{f', f''\}$, we first compute the cost $\beta^{\d}_f(\mu)$ of a cost-minimum $\D$-shape representation of $\skel(\mu) \setminus e'$ within a representation of $\skel(\mu)$ having $f$ as external face. 
	We set the real edges of $\skel(\mu)$ as inflexible and we set the flexibility of each virtual edge $e_\nu \neq e'$ of $\skel(\mu)$ to the maximum value of spirality that the S-component associated with $e_\nu$ can have with the same cost it has for spirality zero. We will compute $\beta^{\d}_f(\mu)$ without bending a flexible edge more than its flexibility. $(ii)$ For each virtual edge $e_\nu \neq e'$ of $\skel(\mu)$, corresponding to a child S-node $\nu$ of $\mu$, we sum to $\beta^{\d}_f(\mu)$ the cost of a bend-minimum representation of $G_\nu$ with spirality equal to the number of bends of $e_\nu$; note that, under the assumption that $e_\nu$ has no more than $\flex(e_\nu)$ bends, this cost equals $b_{e^*}^{0}(\nu)$, i.e., the cost of a bend-minimum representation of $G_\nu$ with spirality zero.     
	$(iii)$ Finally, we set $b_{e^*}^{\d}(\mu)=\min\{\beta^{\d}_{f'}(\mu) + \mathcal{S}_{e^*}(\mu), \beta^{\d}_{f''}(\mu) +\mathcal{S}_{e^*}(\mu)\}$, where $\mathcal{S}_{e^*}(\mu) = \sum_\nu b^0_{e^*}(\nu)$. 
	Clearly, the fact that each virtual edge $e_\nu$ with $k$ bends is replaced by a corresponding $k$-spiral S-component, guarantees that if the orthogonal representation of $\skel(\mu) \setminus e'$ is $\D$-shaped also the representation of $G_\mu$ is $\D$-shaped. The optimality of the $\D$-shaped representation of $G_\mu$ derives from the optimality of the representation of $\skel(\mu) \setminus e'$ and of the S-components that replace the virtual edges of $\skel(\mu)$.   

	We now show how to compute $\beta^{\d}_f(\mu)$. We exploit the \texttt{Bend-Counter} on $\skel(\mu)$, which is a triconnected cubic graph. Based on \cref{le:shape-cost-set-S}, the flexibility of each virtual edge $e_\nu \neq e'$ is set to $\flex(e_\nu) = \min\{4, n^{\d}_{e^*} + n^{\zerob}_{e^*} -1\}$,  where $n^{\d}_{e^*}$ and $n^{\zerob}_{e^*}$ are the values taken from the pair $\langle n^{\d}_{e^*}, n^{\zerob}_{e^*} \rangle$ of node $\nu$. By \cref{le:elbow}, $\flex(e_\nu) \geq 1$. The flexibility of the reference edge $e'$ of $\mu$ is set to $\flex(e')=2$.
	By Property~(a) of \cref{le:shapes}, the value of the formula of Theorem~\ref{th:fixed-embedding-min-bend}, computed on $\skel(\mu)$ with $f$ as external face, is equal to the cost of a cost-minimum $\D$-shape representation of $\skel(\mu) \setminus e'$. Also, Theorem~\ref{th:fixed-embedding-min-bend} and \cref{le:shapes} imply that in this representation each real (inflexible) edge is bent at most once and each inflexible edge is bent no more than its flexibility. By Theorem~\ref{th:bend-counter},  the \texttt{Bend-Counter} is constructed in $O(n_\mu)$ time and it computes in $O(1)$ time the value of the formula in Theorem~\ref{th:fixed-embedding-min-bend}. 
		
	To know whether $b_{e^*}^{\x}(\mu) < b_{e^*}^{\d}(\mu)$, we use the same approach as before, but this time we set $\flex(e')=3$ in the \texttt{Bend-Counter} of $\skel(\mu)$. This update can be done in $O(1)$ time by \cref{le:bend-counter-update}. Again by Theorem~\ref{th:bend-counter}, in $O(1)$ time we can compute the cost $c$ of a cost-minimum orthogonal representation of $\skel(\mu)$ by considering both $f'$ and $f''$ as possible external faces. If $c < \min\{\beta^{\d}_{f'}(\mu), \beta^{\d}_{f''}(\mu)\}$, then $e'$ has necessarily 3 bends in this representation (otherwise $\flex(e')=2$ would have been already sufficient to achieve the same cost), and therefore, due to Property~(b) of \cref{le:shapes}, we conclude that $b_{e^*}^{\x}(\mu) < b_{e^*}^{\d}(\mu)$ and set $b_{e^*}^{\x}(\mu)=c$. Else, we set $b_{e^*}^{\x}(\mu)=\infty$.
		
	At the end of the pre-processing phase described above, for each R-node $\mu$, both the \texttt{Bend-Counter} of $\skel(\mu)$ and the value $\mathcal{S}_{e^*}(\mu)$ are stored at $\mu$ in order to efficiently update the shape-cost set of $\mu$ when $T$ is re-rooted. 
	
	Suppose now to re-root $T$ at the Q-node of an edge $e \neq e^*$. As a consequence the reference edge of $\skel(\mu)$ changes from $e'$ to some other virtual edge $e''$ of $\skel(\mu)$. Denote by $\nu'$ the S-node of $T$ corresponding to $e'$ and by $\nu''$ the S-node of $T$ corresponding to $e''$. We compute $b^{\d}_e(\mu)$ in $O(1)$ time as follows: $(i)$ By \cref{le:bend-counter-update} we update in $O(1)$ time the flexibilities of $e'$ and $e''$, by setting $\flex(e'')=2$ and $\flex(e') = \min\{4, n^{\d}_{e} + n^{\zerob}_{e} -1\}$,  where $n^{\d}_{e}$ and $n^{\zerob}_{e}$ are the values taken from the pair $\langle n^{\d}_{e}, n^{\zerob}_{e} \rangle$ of the S-node $\nu'$. $(ii)$ By Theorem~\ref{th:bend-counter}, for each of the two possible faces $f$ incident to $e''$, we compute $\beta^{\d}_f(\mu)$ in $O(1)$ time. $(iii)$ We compute in $O(1)$ the value $\mathcal{S}_{e}(\mu) = \sum_\nu b^0_{e}(\nu)$ from $\mathcal{S}_{e^*}(\mu)$, by subtracting $b^0_{e^*}(\nu'')$ and by adding $b^0_{e^*}(\nu')$; notice that $b^0_{e}(\nu) = b^0_{e^*}(\nu)$ for any $\nu$ distinct from $\nu'$ and $\nu''$.    
	
	To know whether $b_{e}^{\x}(\mu) < b_{e}^{\d}(\mu)$ we follow the same approach as for the pre-processing visit, setting $\flex(e'')=3$. 
\end{proof}

We finally prove how to compute the label $b(e)$ depending on the type of the root child. 

\begin{restatable}{lemma}{leBE}\label{le:b(e)}
	Let $T$ be the SPQR-tree of $G$ rooted at the Q-node of an edge $e$ and let $\mu$ be the root child of $T$. The label $b(e)$ can be computed in $O(n_\mu)$ time in the pre-processing visit (i.e., when $e=e^*$) and in $O(1)$ time otherwise.
\end{restatable}
\begin{proof}
	By Theorem~\ref{th:shapes}, if $\mu$ is a P-node, $b(e) = \min\{b_{e}^{\c}(\mu), b_{e}^{\l}(\mu)+1\}$, where $b_{e}^{\c}(\mu)$ and $b_{e}^{\l}(\mu)$ are computed according to \cref{le:shape-cost-set-P} for the root child $\mu$. Indeed, $b_{e}^{\c}(\mu)$ and $b_{e}^{\l}(\mu)+1$ represent the cost of a bend-minimum representation of $G$ when $e$ is \zeroB-shaped and \oneB-shaped, respectively. Hence, computing $b(e)$ in the case of root child P-nodes takes $O(1)$ time.
	If $\mu$ is an S-node, we have the following cases based on the degree that the poles $u$ and $v$ of $\mu$ have in $G_\mu$:
	\begin{itemize}
	\item Both $u$ and $v$ have degree two in $G_\mu$. In this case $b(e) = \min\{b_{e}^{4}(\mu), b_{e}^{3}(\mu)+1\}$, where $b_{e}^{4}(\mu)$ and $b_{e}^{3}(\mu)$ are computed according to \cref{le:shape-cost-set-S} for the root child $\mu$. Again, $b_{e}^{4}(\mu)$ and $b_{e}^{3}(\mu)+1$ represent the cost of a bend-minimum representation of $G$ when $e$ is \zeroB-shaped and \oneB-shaped, respectively.
	
	\item Exactly one of $u$ and $v$ has degree two in $G_\mu$. In this case $b(e) = \min\{b_{e}^{3}(\mu), b_{e}^{2}(\mu)+1\}$.
	
	\item Both $u$ and $v$ have degree one in $G_\mu$. In this case $b(e) = \min\{b_{e}^{2}(\mu), b_{e}^{1}(\mu)+1\}$. 
	\end{itemize}
	
	\noindent In all cases the computation of $b(e)$ takes $O(1)$ time.
	
	\smallskip If $\mu$ is an R-node we distinguish the case of the pre-processing visit and of the successive re-rooting of $T$.
	In the pre-processing visit we build the \texttt{Bend-Counter} of $\skel(\mu)$, as for the inner R-nodes, but this time the reference edge $e'=e^*$ is a real (inflexible) edge. Also we compute $\mathcal{S}_{e^*}(\mu)$ summing up the values $b^{0}_{e^*}(\nu)$ of each S-node child $\nu$ of $\mu$.
	The value $b(e^*)$ is obtained by adding $\mathcal{S}_{e^*}(\mu)$ with the minimum value returned by the \texttt{Bend-Counter} for the two faces $f'$ and $f''$ incident to $e^*$. By Theorem~\ref{th:shapes}, this latter value coincides with the minimum between $\beta^{\c}_{e^*}(\mu)$ and $\beta^{\l}_{e^*}(\mu) + 1$, although these two values are not explicity computed. Since the \texttt{Bend-Counter} can be constructed in $O(n_\mu)$ time (\cref{le:bend-counter_computation}) and $\mathcal{S}_{e^*}(\mu)$ can also be computed in $O(n_\mu)$ time, the time needed to process node $\mu$ is $O(n_\mu)$.
	
	When $T$ is re-rooted at the Q-node of an edge $e \neq e^*$, we take advantage of the \texttt{Bend-Counter} and of $\mathcal{S}_{e^*}(\mu)$ stored at $\mu$. Namely, $b(e)$ is the sum of the minimum value returned by the \texttt{Bend-Counter} for the two faces incident to $e$ and of the value of $\mathcal{S}_{e}(\mu)$, obtained from $\mathcal{S}_{e^*}(\mu)$ as described for the inner R-nodes in the proof of \cref{le:shape-cost-set-inner-R}. This takes $O(1)$ time.
	
	We finally observe that if the external face of $\skel(\mu)$ is a 3-cycle with all inflexible edges, then the value of the formula in Theorem~\ref{th:fixed-embedding-min-bend} would correspond to a solution with two bends along a real edge. 
	In this case we set $b(e)=\infty$, which will avoid to select $e$ as the best reference edge for $T$. By Theorem~\ref{th:shapes} and by the fact that $G$ is not $K_4$, we know that there will be another bend-minimum representation with at most one bend per edge for a different embedding of the graph.
\end{proof}

\subsection{A Reusability Principle (Proof of Theorem~\ref{th:key-result-2})}\label{sse:labeling-biconnected}

Observe that, by naively applying the algorithms in the proofs of \cref{le:shape-cost-set-inner-Q,le:shape-cost-set-P,le:shape-cost-set-S,le:shape-cost-set-inner-R,le:b(e)} for all possible re-rootings of $T$, we would not have the $O(n)$-time algorithm that is needed to prove Theorem~\ref{th:key-result-2}. In fact, even if the computation on each node of $T$ took $O(1)$ time, we would perform such computation $O(n)$ times, that is, once for each possible edge to be labeled.
%
%
%
In order to achieve $O(n)$ time complexity for the labeling algorithm over all possible re-rootings, we exploit a reusability principle that is expressed in general terms because it can have applications also beyond the scope of this paper. The strategy of the proof of the next lemma is illustrated in \cref{fi:amortized}.


\begin{lemma}\label{le:amortized}
Let $G$ a biconnected planar graph with $n$ vertices, let $e$ be any edge of $G$, and let $T$ be the SPQR-tree of $G$ rooted at the Q-node corresponding to $e$. Let $\mathcal{A}$ be an algorithm that traverses bottom-up $T$ and, for each node $\mu$ of $T$, computes in $O(f(n))$ time some value $\mathcal{V}_e(\mu)$ based on the values $\mathcal{V}_e(\mu_i)$, $i = 1, \dots, k$ computed for the children $\mu_i$ of $\mu$. There exists an algorithm $\mathcal{A^+}$ that executes $\mathcal{A}$ for all possible re-rootings of $T$ at all its Q-nodes in $O(n \cdot f(n))$ time.
\end{lemma}
\begin{proof}
In the proof we will say that $T$ is rooted at $e$ to mean that the root of $T$ is the Q-node associated with $e$; also, we denote by $T_e$ the tree $T$ rooted at $e$.
We equip each edge $(\mu_i,\mu_j)$ of $T_e$ with two \emph{darts}: dart $\overrightarrow{\mu_i \mu_j}$ stores the value of $\mathcal{V}_e(\mu_i)$ when $\mu_j$ is the parent of $\mu_i$ in $T_e$; dart $\overrightarrow{\mu_j \mu_i}$ stores the value of $\mathcal{V}_e(\mu_j)$  when $\mu_i$ is the parent of $\mu_j$ in $T_e$.
We first perform a bottom-up (i.e., post-order) traversal of $T_{e_1}$ rooted at the Q-node corresponding to some edge $e_1$. During this traversal, for each pair of nodes $\mu_i$ and $\mu_j$ of $T_{e_1}$ such that $\mu_i$ is a child of $\mu_j$, the algorithm computes the value of $\mathcal{V}_{e_i}(\mu_i)$ and stores it in the dart $\overrightarrow{\mu_i \mu_j}$.
Then, we re-root the SPQR-tree at the Q-node corresponding to some edge $e_2 \neq e_1$ and we perform a post-order traversal of $T_{e_2}$. If the node $\mu_i$ is still a child of $\mu_j$ in $T_{e_2}$, we re-use, in the traversal of $T_{e_2}$, the value $\mathcal{V}_{e_1}(\mu_i)$ stored into dart $\overrightarrow{\mu_i \mu_j}$ without the need of visiting again the subtree rooted at $\mu_i$. Otherwise, if $\mu_i$ is the parent of $\mu_j$ in $T_{e_2}$, we launch a post-order traversal of the subtree rooted at $\mu_j$ and equip dart $\overrightarrow{\mu_j \mu_i}$ with the value $\mathcal{V}_{e_2}(\mu_j)$.
We proceed in this way until each edge $e$ of $G$ has been used as reference edge of $T$, calling each time a post-order traversal on $T_e$.

\begin{figure}[tb]
	\centering
	\subfloat[]{\label{fi:amortized-a}\includegraphics[page=1,width=0.46\columnwidth]{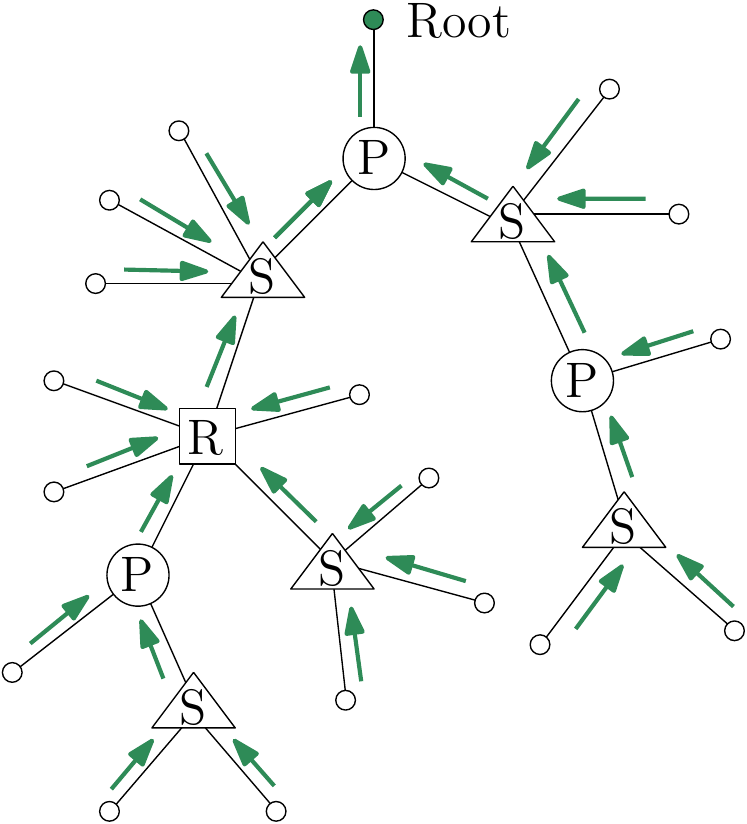}}
	\hfil
	\subfloat[]{\label{fi:amortized-b}\includegraphics[page=2,width=0.46\columnwidth]{amortized}}
	\hfil
	\subfloat[]{\label{fi:amortized-c}\includegraphics[page=3,width=0.46\columnwidth]{amortized}}
	\hfil
	\subfloat[]{\label{fi:amortized-d}\includegraphics[page=4,width=0.46\columnwidth]{amortized}}
	\hfil
	\caption{A figure showing the algorithm in the proof of \cref{le:amortized}. (a) The darts updated by the first bottom up traversal of the SPQR-tree. (b--c) The darts updated by two subsequent bottom up traversals of the rerooted SPQR-tree. (d) The darts updated at the end of all bottom up traversals.}\label{fi:amortized}
\end{figure}

We prove that the time needed to perform all the $O(n)$ post-order traversals is $n \cdot O(f(n))$. In fact, since at the end of each recursive call of the post-order traversal one dart is updated, and since a dart is never updated twice, the total number of recursive calls of the post-order traversal is equal to the number of the darts, which is~$O(n)$. Each recursive call on a node $\mu$, in addition to possibly calling the post-order traversal on some children of $\mu$, uses $O(f(n))$-time to compute $\mathcal{V}(\mu)$. Therefore, the overall time used by all the $O(n)$ post-order traversals is $O(f(n))\cdot O(n)=O(f(n) \cdot n)$.
\end{proof}

The proof of Theorem~\ref{th:key-result-2} is a consequence of \cref{le:shape-cost-set-inner-Q,le:shape-cost-set-P,le:shape-cost-set-S,le:shape-cost-set-inner-R,le:b(e)} and of \cref{le:amortized} where $\mathcal{V}_{e}(\mu) = \{b_e^{\sigma_1}(\mu), b_e^{\sigma_2}(\mu), \ldots, b_e^{\sigma_h}(\mu)\}$. In this case that $f(n)=O(1)$, and therefore $O(f(n) \cdot n) = O(n)$.

\subsection{Labeling 1-connected graphs (Proof of Theorem~\ref{th:1-connected-labeling})}\label{sse:labeling-1-connected}

Let $G$ be a 1-connected planar 3-graph and let $\cal T$ be the block-cut-vertex tree of $G$. We recall that $\cal T$ describes the decomposition of $G$ in terms of its blocks, i.e., of its biconnected components (see, e.g.,~\cite{dett-gd-99}). Each node of $\cal T$ is either a block or a cut-vertex, and cut-vertices are incident to the blocks they belong to. Call \emph{trivial blocks} those composed of a single edge and \emph{full blocks} the remaining. Since $\Delta(G) \leq 3$, full blocks are only adjacent to trivial blocks. Also, cut-vertices of degree three are adjacent to three trivial blocks. Suppose $\cal T$ is labeled at a block $B$. Let $B'$ be a leaf of $\cal T$ and let $c'$ its parent cut-vertex in $\cal T$. We associate a cost $b_B(B')$ with $B'$ defined as the label $b(c')$ that $c'$ has in $B'$. We recall that $b(c')$ is the minimum of the labels associated with the edges of $B'$ incident to $c'$. If $B'$ is an internal node of $\cal T$ distinct from the root, its cost $b_B(B')$ is the sum of the costs of its grandchildren plus the label $b(c')$ that its parent cut-vertex has in $B'$. For the root $B$, the label $b(B)$ is defined by the sum of the costs of its grandchildren plus the minimum value $b(e)$ over all edges $e$ of $B$.
Given the constant degree of the cut-vertex nodes of $\cal T$, for each possible root $B$ of $\cal T$, the cost of each block is computed in $O(1)$ time.

Similar to the biconnected case, a critical step with respect to the overall time complexity of the labeling algorithm is when we try each non-trivial block $B$ as a candidate root of $\mathcal{T}$. In fact, if for each block $B$ we summed up the contributions of all other blocks, we would use $O(n)$ time for each candidate root block $B$ and $O(n^2)$ time overall.
To achieve $O(n)$-time complexity, we use a strategy similar to the one described in the proof of \cref{le:amortized}. Namely, we compute such sums with a bottom-up traversal of $\mathcal{T}$ rooted at the current candidate root block $B$ and we store in the darts of $\mathcal{T}$ pointing towards the root $B$ the partial values of the sums of the subtrees. When we re-root $\mathcal{T}$ on some other candidate root $B'$, we launch another bottom-up traversal of $\mathcal{T}$, where the subtrees of $\mathcal{T}$ for which the partial value of the sum has been stored into the dart pointing towards the root $B'$ do not need to be explored again.
This concludes the proof of Theorem~\ref{th:1-connected-labeling}.

\section{Optimal~Orthogonal~Representations: Proofs of Theorems~\ref{th:gd2018-enhanced} and~\ref{th:gd2018-enhanced-v}}\label{se:thgd2018-enhanced}


Concerning the proof of Theorem~\ref{th:gd2018-enhanced}, we observe that the algorithm used to construct an $e$-constrained bend-minimum representation $H$ of $G$ that satisfies Properties~\textsf{O1--O4} of Theorem~\ref{th:shapes} performs a bottom-up visit of $T$ rooted at $e$, and computes the cost-shape set of each visited node $\mu$ 
as described in the proofs of \cref{le:shape-cost-set-inner-Q,le:shape-cost-set-P,le:shape-cost-set-S,le:shape-cost-set-inner-R,le:b(e)}, depending on whether $\mu$ is a Q-, P-, S-, or R-node. To prove Theorem~\ref{th:gd2018-enhanced} we use an algorithm that follows the same strategy as the one of Theorem~\ref{th:shapes}. In addition to the shape-cost set, this algorithm stores at $\mu$ an orthogonal representation $H_e^{\sigma}(\mu)$ of cost $b_e^{\sigma}(\mu)$ for each possible representative shape $\sigma$ of $\mu$. Such a representation is constructed in $O(n_\mu)$ time for S- and R-nodes, and in $O(1)$ time for P- and Q-nodes. 
Similarly to \cite{DBLP:conf/gd/DidimoLP18}, to achieve overall linear-time complexity, the representations for the representative shapes stored at $\mu$ are described incrementally, linking the desired representations in the set of the children of $\mu$ for each virtual edge of $\skel(\mu)$. Namely:
\begin{itemize}
	\item For a leaf Q-node $\mu$, the representations $H_e^{\zerob}(\mu)$ and $H_e^{\oneb}(\mu)$ are trivially constructed. 
	
	\item For a P-node $\mu$ the orthogonal representations stored at $\mu$ are constructed by simply merging the representations of the children of $\mu$ with the suitable values of spirality described in \cref{le:shape-cost-set-inner-Q}. 
	
	\item For an S-node $\mu$ and for each value $k$ of spirality for which $b_e^k(\mu) \neq \infty$, we suitably merge the cheapest representations of the children of $\mu$ associated with the virtual edges of $\skel(\mu)$ distinct from the reference edge, and we distribute the $B_{e}(k)_{extra}$ bends on the real edges of $\skel(\mu)$, so that each real edge takes at most one bend. According to the procedure described in the proof of \cref{le:shape-cost-set-S}, this is always possible if $b^k_e(\mu) \neq \infty$.   
	
	\item If $\mu$ is an R-node, we compute a cost-minimum orthogonal representation of $\skel(\mu)$ by means of Theorem~\ref{th:fixed-embedding-min-bend}, where the flexibilities of the edges are defined as in the proofs of \cref{le:shape-cost-set-inner-R,le:b(e)}. To obtain the final representation $H_e^\sigma(\mu)$, we replace each virtual edge $e_\nu$ (distinct from the reference edge) in the representation of $\skel(\mu)$  with the representation of the corresponding S-node $\nu$ whose spirality equals the number of bends of~$e_\nu$.
\end{itemize}

Therefore, the whole algorithm takes $O(n)$ time and the orthogonal representation of cost $b(e)$ associated with the root of $T$ is the desired bend-minimum orthogonal representation of $G$. This concludes the proof of Theorem~\ref{th:gd2018-enhanced}.


\medskip
About the proof of Theorem~\ref{th:gd2018-enhanced-v}, if $v$ is a degree-1 vertex, $G$ consists of a single edge, and the statement is obvious. If $v$ is a degree-2 vertex, \cref{le:1-bend} guarantees the existence of a $v$-constrained optimal orthogonal representation of $G$ with an angle larger than $90^\circ$ at $v$ on the external face. Clearly, if $v$ is on the external face, then both its two edges are on the external face. Therefore, for any edge $e$ incident to $v$ we can construct the SPQR-tree $T$ of $G$ rooted at $e$ and we can apply on $T$ the $O(n)$-time algorithm in the proof of Theorem~\ref{th:gd2018-enhanced}. Since in this case the child-root $\mu$ of $T$ is an $S$-node (because $v$ has degree 2), that algorithm guarantees that the angle at $v$ on the external face is an angle of either $180^\circ$ or $270^\circ$ when it merges a representation of $G_\mu$ with~$e$.    




\section{Conclusions and Open Problems}\label{se:conclusions}
We have solved a long standing open problem by proving that an orthogonal representation of a planar 3-graph with the minimum number of bends can be computed in $O(n)$ time in the variable embedding setting. Furthermore our construction has at most one bend per edge, which is also optimal. We mention some open problems that we find interesting to investigate.

\begin{enumerate}
\item A key ingredient of our linear time result is the evidence that a bend-minimum orthogonal representation of a planar 3-graph does not need to ``roll-up'' too much. This may be true also for some families planar 4-graphs. For example, can one apply some of the ideas in this paper to compute bend-minimum orthogonal representations of series-parallel 4-graphs? The most efficient algorithm for this problem has time complexity $O(n^3 \log n)$~\cite{DBLP:journals/siamcomp/BattistaLV98}.

\item The HV-planarity testing problem asks whether a given planar graph admits a rectilinear drawing with prescribed horizontal and vertical orientations of the edges. This problem is NP-complete also for planar 3-graphs. An $O(n^3 \log n)$-time solution is known for series parallel 3-biconnected~\cite{dlp-hvpac-19}. Can our techniques be extended to design an $o(n^3 \log n)$-time algorithm?

\item Given the strong similarities between the rectilinear and the upward planarity testing problems~\cite{DBLP:journals/siamcomp/GargT01}, we wonder whether our ideas can be extended to devise efficient upward planarity testing algorithms for special families of graphs such as, for example, triconnected DAGs.
\end{enumerate}

\bibliographystyle{siamplain}
\bibliography{bibliography-arxiv}

\end{document}